\documentclass[12pt]{article}
\usepackage{savesym}
\usepackage{cite}
\usepackage{wasysym}
\usepackage{amscd}
\usepackage{graphicx}
\usepackage{pifont}
\usepackage{float}
\usepackage{subfig}
\usepackage[all]{xy}
\usepackage{multicol}
\usepackage{amsfonts}
\usepackage{picture}
\usepackage{amssymb}
\savesymbol{iint}
\savesymbol{iiint}
\usepackage{amsmath}
\usepackage{amsthm}
\restoresymbol{TXF}{iint}
\restoresymbol{TXF}{iiint}
\usepackage{color}
\usepackage{clay}
\usepackage{hyperref}
\hypersetup{colorlinks=true}
\hypersetup{linkcolor=black}
\hypersetup{citecolor=black}
\hypersetup{urlcolor=black}
\usepackage{setspace}
\usepackage{wrapfig}
\usepackage{verbatim}
\newtheorem{thm}{Theorem}
\newtheorem{prop}{Proposition}

\numberwithin{equation}{section}

\newcommand{\fX}{{\mathfrak X}}

\newcommand{\N}{{\mathcal N}}
\newcommand{\Z}{{\mathbb Z}}
\newcommand{\C}{{\mathbb C}}
\newcommand{\R}{{\mathbb R}}
\newcommand{\I}{{\mathrm i}}
\newcommand{\re}{{\mathrm {Re}}}
\newcommand{\im}{{\mathrm {Im}}}
\newcommand{\BPS}{{\mathrm {BPS}}}
\newcommand{\RG}{{\mathbf {RG}}}
\newcommand{\abs}[1]{\lvert#1\rvert}
\newcommand{\HH}{{\mathcal H}}
\newcommand{\cO}{{\mathcal O}}
\newcommand{\inprod}[1]{\langle#1\rangle}
\newcommand\fro{{\overline{\underline{\Omega}}}}

\begin{document}

\date{February, 2014}

\institution{Fellows}{\centerline{${}^{1}$Society of Fellows, Harvard University, Cambridge, MA, USA}}
\institution{Austin}{\centerline{${}^{2}$Department of Mathematics, University of Texas at Austin, Austin, TX, USA}}

\title{Line Defects, Tropicalization, and Multi-Centered Quiver Quantum Mechanics}

\authors{Clay C\'{o}rdova\worksat{\Fellows}\footnote{e-mail: {\tt cordova@physics.harvard.edu}}  and Andrew Neitzke \worksat{\Austin}\footnote{e-mail: {\tt neitzke@math.utexas.edu}} }

\abstract{We study BPS line defects in $\N=2$ supersymmetric four-dimensional field theories.  We focus on theories of ``quiver type,'' those for which the BPS particle spectrum can be computed using quiver quantum mechanics. For a wide class of models, the renormalization group flow between defects defined in the ultraviolet and in the infrared is bijective.  Using this fact, we propose a way to compute the BPS Hilbert space of a defect defined in the ultraviolet, using only infrared data.  In some cases our proposal reduces to studying representations of a ``framed'' quiver, with one extra node representing the defect.  In general, though, it is different. As applications, we derive a formula for the discontinuities in the defect renormalization group map under variations of moduli, and show that the operator product algebra of line defects contains distinguished subalgebras with universal multiplication rules.  We illustrate our results in several explicit examples.}

\maketitle

\tableofcontents

\bibliographystyle{utphys}

\enlargethispage{\baselineskip}

\section{Introduction and Summary}
\label{intro}

Non-local operators are useful tools in quantum field theory.  In the most well-known example, the Wilson loop and its generalizations can be used as order parameters to study phases of gauge theories \cite{Wilson:1974sk, 'tHooft:1977hy}.   More generally, line defects --- objects extended along a timelike world-line --- model infinitely massive particles, and the response of a field theory to such sources can teach us about the spectrum of excitations of the original system \cite{Kapustin:2005py, Drukker:2010jp, Aharony:2013hda, Billo:2013jda}.

In this paper we study line defects in the context of four-dimensional supersymmetric field theories with $\mathcal{N}=2$ supersymmetry.  In such field theories there exist supersymmetric line defects preserving one-half of the supercharges.  Such defects
have been previously studied e.g. in \cite{Kapustin:2006hi,Drukker:2009id,Drukker:2009tz,Drukker:2010jp,Gaiotto:2010be}.  Familiar examples are provided by supersymmetric Wilson-'t Hooft lines in supersymmetric gauge theories, but the concept is more general: all evidence accumulated so far suggests that every $\mathcal{N}=2$ theory, even if it has no known Lagrangian, supports a wide variety of BPS line defects.  The Hilbert space of the quantum field theory in the presence of such supersymmetric defects supports a set of supersymmetric states known as \emph{framed BPS states} \cite{Gaiotto:2010be}.  As for the usual BPS states, these framed BPS states saturate a BPS bound and thus they are
the lowest-energy states in their charge sectors.

An important special case of supersymmetric line defects occurs in abelian gauge theories.  Such a theory has an integral electromagnetic charge lattice $\Gamma$, and supersymmetric line defects are labelled by a charge $\gamma \in \Gamma$.  Physically, the defect corresponding to a charge $\gamma$
models the insertion of an infinitely massive BPS dyon carrying charge $\gamma$.
As the abelian theory is free, the spectrum of BPS excitations around each defect is easy to describe, and there is exactly one framed BPS state, describing the vacuum in the presence of the defect.

In general, our method for studying line defects in an interacting $\mathcal{N}=2$ theory is to deform onto the Coulomb branch and follow the defect to the infrared.  The low-energy effective field theory after this flow is a free abelian gauge theory.  Correspondingly, the IR limit of a given line defect is a superposition of line defects of the abelian theory \cite{Gaiotto:2009hg}.   The coefficients in this superposition are given precisely by the counts of framed BPS states.  We extract the framed BPS state of lowest energy, which may be thought of as the ground state of the theory in the presence of the defect.  This ground state carries a charge, and the assignment of this charge to a UV line defect determines a renormalization group map:
\begin{equation}
\mathbf{RG}: \{\mathrm{UV \ Line \ Defects}\}\rightarrow \{\mathrm{IR \ Line \ Defects}\}=\Gamma. \label{RGmapintro}
\end{equation}
Many of the results of this paper are a consequence of the remarkable properties of this renormalization group flow.  In particular the map $\mathbf{RG}$ is conjecturally bijective \cite{Gaiotto:2010be}, and hence yields IR techniques for studying UV line defects.  In this paper we provide further evidence for this idea by illustrating a method for reconstructing UV line defects (or at least their complete spectra of framed BPS states) from their IR charges, in a wide class of $\mathcal{N}=2$ theories.

The correspondence \eqref{RGmapintro} has significant implications in both physics and mathematics.  Physically, it implies that screening of BPS UV line defects does not occur: any two UV defects can be distinguished by low-energy observations.  Mathematically, the $\mathbf{RG}$ map has an interpretation in terms of the geometry of the Coulomb branch.  The four-dimensional theory reduced on a circle has a moduli space $\fX$, and the expectation value of a line defect wrapping the circle leads to a holomorphic function on this space.  The conjectured bijection \eqref{RGmapintro} can then be interpreted as stating that the set of such functions can be labelled by integral points of a \emph{tropical} version of $\fX$; this seems to give a physical realization of a proposal of \cite{MR2233852,MR2349682} (as also discussed in \cite{Gaiotto:2010be}.)  We sketch these ideas in \S \ref{tropical}, but leave a full exploration of this subject to future work.

In the remainder of \S \ref{genlineops} we survey the general properties of BPS line defects including their framed BPS spectra and the renormalization group flow described above.  In addition, we review how framed BPS states can be used to compute the operator product algebra of line defects defined in \cite{Kapustin:2006pk} and further studied in \cite{Kapustin:2006hi, Drukker:2009tz, Gomis:2009ir, Gomis:2009xg, Gaiotto:2010be, Cecotti:2010fi, Ito:2011ea, Saulina:2011qr, Moraru:2012nu, Xie:2013lca}.

In \S \ref{quiverconstruct} we introduce our main class of examples, $\mathcal{N}=2$ theories of quiver type \cite{Douglas:1996sw, Douglas:2000ah,  Douglas:2000qw, Alim:2011kw}.  These are $\mathcal{N}=2$ systems whose BPS spectra may be computed from non-relativistic multi-particle quantum mechanics encoded by a quiver.  Not all $\mathcal{N}=2$ field theories admit such a simple description of their BPS states.  However, included in this class are gauge theories with matter, many strongly interacting non-lagrangian  conformal field theories, and theories described by M5-branes on punctured Riemann surfaces \cite{Fiol:2000pd, Fiol:2000wx, Fiol:2006jz, Cecotti:2011rv, Cecotti:2011gu, Alim:2011ae,    DelZotto:2011an, Xie:2012dw, Cecotti:2012va,  Cecotti:2012sf, Cecotti:2012gh, Saidi:2012gi, Cecotti:2012jx, Xie:2012jd, Cecotti:2012se, Cecotti:2012kv, Cecotti:2013lda, Cecotti:2013sza, Galakhov:2013oja}.  In any $\N=2$ theory of quiver type, we ask: how can one generalize the quantum-mechanical description of BPS states to framed BPS states?

There is a natural candidate answer to this question.  To model a defect, one may extend the quiver by adding one new node, with charge dictated by the $\mathbf{RG}$ map.  We refer to these extended quivers as \emph{framed quivers}.  Framed quivers have appeared previously in work on framed BPS states \cite{Chuang:2013wt, Cirafici:2013bha}, and in the context of Donaldson-Thomas theory \cite{MR2403807, MR2592501, Ooguri:2008yb}.   A surprising observation of this paper is that \emph{a naive calculation of framed BPS states from a framed quiver in general produces an incorrect spectrum}.\footnote{In the restricted class of examples studied in \cite{Chuang:2013wt, Cirafici:2013bha} the framed quiver produces the correct answer. }

To understand why the naive computation gives the wrong answer, we recall that in theories of quiver type there are two distinct descriptions of BPS states.  These complementary perspectives arise respectively from quantization of the Higgs and Coulomb branches of the quiver quantum mechanics \cite{Denef:2002ru}.  The Higgs branch provides a simpler mathematical formalism, and relates BPS states to cohomology of moduli spaces of quiver representations \cite{Douglas:1996sw, Douglas:2000ah, Douglas:2000qw}.  It describes a configuration of BPS particles with nearly coincident positions.   By contrast, on the Coulomb branch, BPS states are obtained by quantization of classically stable multi-particle configurations, and the resulting wavefunctions physically describe a kind of $\mathcal{N}=2$ molecule where BPS states are bound together by electromagnetic and scalar exchange, typically orbiting each other at non-zero semiclassical radius \cite{Denef:2000nb}.

There is a subtle relationship between these two pictures of BPS particles.  As parameters are varied, one or the other descriptions is more natural.  In simple cases, one
can follow a state represented by a particular wavefunction on the Higgs branch to a corresponding multi-centered configuration on the Coulomb branch; however, in general,
a wavefunction on the Higgs branch may correspond on the Coulomb branch to a degenerate molecule where semiclassically some of the particles have coincident positions \cite{Denef:2007vg, deBoer:2008zn}.   This is a field theory version of the supergravity problem of scaling configurations of multi-centered black holes.

Now let us return to the problem of interest in this paper, the computation of framed BPS states.  In the examples we consider, we find that the correct spectrum is obtained by counting only a subset of the Higgs branch states for our framed quiver.  Roughly speaking, these are the states whose spatial description on the Coulomb branch involves a \emph{core charge} consisting of an infinitely massive defect, and a \emph{halo charge} consisting of finite mass BPS particles, where moreover the \emph{halo particles remain at non-zero spatial separation from the core}.  A typical framed BPS state is illustrated in Figure \ref{fig:ptolemyintro}.
\begin{figure}[here!]
  \centering
  \includegraphics[width=0.4\textwidth]{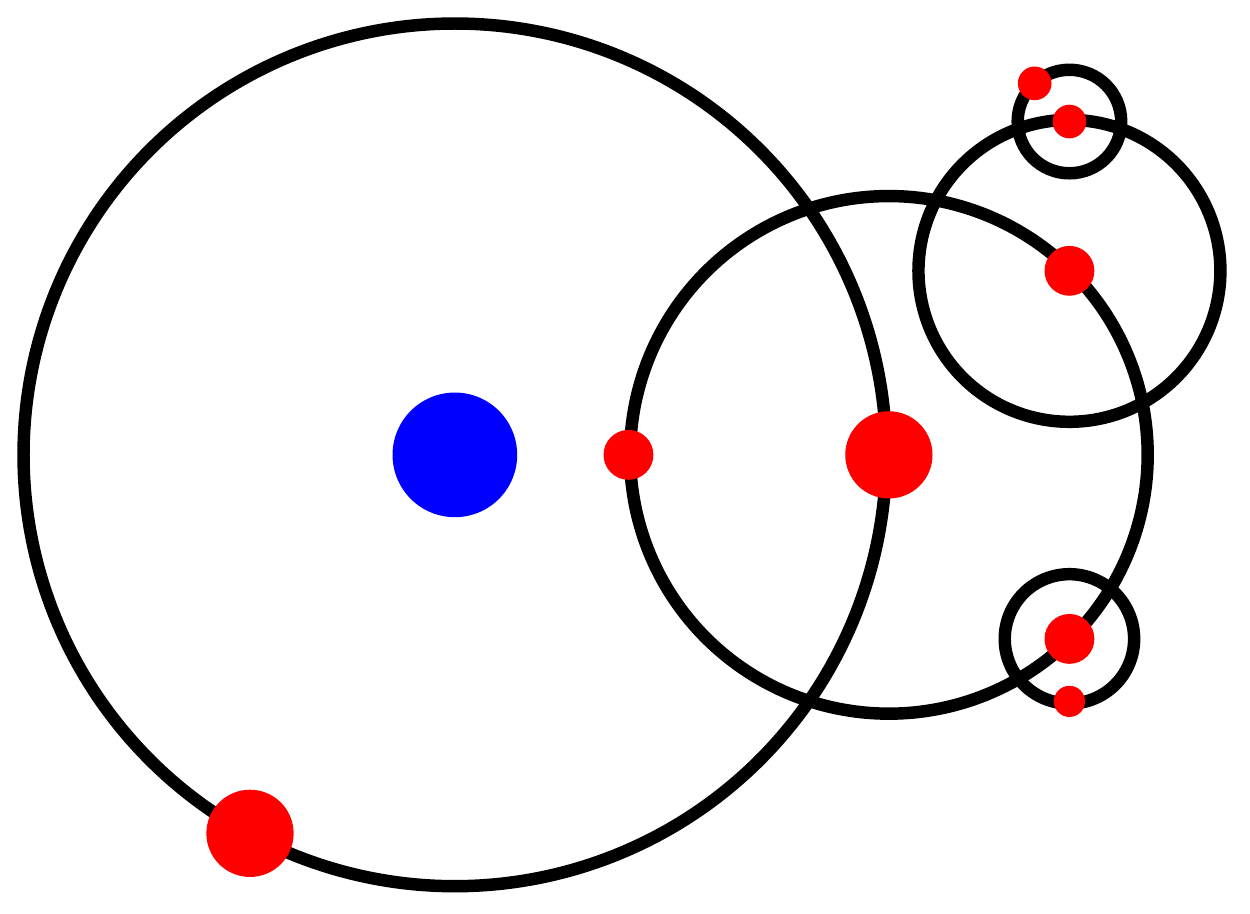}
  \caption{A semiclassical cartoon of a framed BPS state.  The blue object indicates the infinitely massive core charge.  It is modeled in the quiver description by a new node with charge determined by the $\mathbf{RG}$ map.  The red objects are ordinary BPS states of finite mass.  They comprise the halo charge which is bound to the defect, and orbit the core in a multi-centered Ptolemaean system.  The core charge is spatially separated from any constituent of the halo.}
  \label{fig:ptolemyintro}
\end{figure}
In other words, from the spectrum computed by the framed quiver, we discard those states where the halo and core charges become spatially coincident.

This rough idea may be made precise using a localization formula for the counting of Coulomb branch states \cite{Manschot:2011xc, Manschot:2012rx, Manschot:2013sya}, which enables us to consistently separate the unwanted configurations.  From this, we obtain our main proposal, in \S \ref{framing}: a conjectural algorithm for calculating the framed BPS states of any BPS line defect in any theory of quiver type.  Our proposal satisfies two consistency checks.  First, it obeys the wall-crossing formula \cite{KS1, KS2, KS3, JS, J} which governs decays of framed BPS states.  Second, it reproduces known results in explicit examples.  In \S \ref{framedquivers} we also describe conditions under which our proposal reduces to a more traditional calculation using the framed quiver.

In \S \ref{properties} we study some interesting implications of our proposal.  A first result concerns the behavior of the $\mathbf{RG}$ map as moduli are varied.  Locally in moduli space the $\mathbf{RG}$ map is constant; however, at certain codimension one loci, dubbed \emph{anti-walls}, the ground state of a defect may become non-unique.  Typically, the degenerate ground states have distinct electromagnetic charges, and hence exactly at the anti-wall the $\mathbf{RG}$ map is ill-defined.  On either side of the anti-wall, the degeneracy is lifted, but in crossing the wall the $\mathbf{RG}$ map changes discontinuously.  For a large class of anti-walls we explain how these discontinuities fit into the more general scheme of quantum-mechanical dualities described by quiver mutations.  This enables us to derive a uniform formula for the discontinuities at these anti-walls in terms of certain piecewise linear transformations on the infrared charge lattice $\Gamma$, matching a proposal from \cite{Gaiotto:2010be}.

This result implies that the collection of UV line defects of an $\N =2$ theory of quiver type is endowed with a natural piecewise linear structure.  More concretely, there is an atlas of possible ``charts'', each isomorphic to the charge lattice $\Gamma$, which via the inverse of $\mathbf{RG}$, may be used to label BPS UV line defects.  The maps which govern the changes of coordinates are exactly the piecewise linear discontinuities described above.

In \S \ref{dualseeds} we identify a class of defects with simple properties.  At a fixed generic point in moduli space, we label UV defects by their ground state charge as $L_{\gamma}$.  We identify a distinguished cone, $\check{C}$, inside the charge lattice and using the inverse of $\mathbf{RG}$ it defines a cone inside the set of BPS UV line defects.  We prove that line defects with ground state charge in $\check{C}$ obey the universal algebra
\begin{equation}
\gamma_{1}, \gamma_{2} \in \check{\mathcal{C}}\Longrightarrow L_{\gamma_{1}} * L_{\gamma_{2}}=y^{\langle \gamma_{1}, \gamma_{2}\rangle}L_{\gamma_{1}+\gamma_{2}}, \label{universalintro}
\end{equation}
where in the above, $*$ is the non-commutative OPE of line defects with formal non-commutative parameter $y$, and $\langle \cdot, \cdot \rangle$ indicates the symplectic product on the charge lattice.  In the context of theories of quiver type, \eqref{universalintro} implies a nontrivial recursion relation on the cohomology of framed quiver moduli spaces.

In \S \ref{mutateddualseeds}, we extend the result \eqref{universalintro} by making use of quantum mechanical dualities encoded by quiver mutations.   We show that there are in fact many distinct cones each obeying the simple subalgebra.  The geometry of these cones of line defects provides significant information about the defect OPE algebra.

Finally, in \S \ref{examples}, we illustrate our results with explicit examples.  We focus on systems described by the so-called Kronecker quivers: those with two nodes, and an arbitrary number $k$ of unidirectional arrows.  For $k=1,$ the quiver describes the BPS states of the Argyres-Douglas CFT defined in \cite{Argyres:1995jj, Argyres:1995xn}.  For $k=2$ it describes the BPS states in $SU(2)$ super Yang-Mills studied in \cite{Seiberg:1994rs, Seiberg:1994aj}.  For $k>2$ this quiver is not known to occur as a complete description of the spectrum of any UV theory; however, it does occur as a subsector of many $\mathcal{N}=2$ theories of quiver type, including in particular $SU(n)$ super Yang-Mills with $n>2$ \cite{Galakhov:2013oja}.

In \S \ref{dualseedsex} we work out the geometry of the distinguished cones of line defects.  The cases $k=1$, $k=2$ and $k>2$ exemplify three possible phenomena.  When $k=1$ the cones are finite in number and fill the entire charge lattice.  When $k=2,$ there are infinitely many cones and they accumulate.  The complement of the set of  cones is a single ray.  Finally, when $k>2,$ there are infinitely many cones which accumulate and whose complement is an open set.

In \S \ref{explicitspectra} we evaluate explicit examples of framed spectra and OPEs.  For the case of the Argyres-Douglas theory, we evaluate all OPEs and hence determine the complete non-commutative algebra of line defects.  For the case of $SU(2)$ super Yang-Mills, we determine the spectrum supported by a Wilson line in an arbitrary representation of $SU(2)$.  This example is especially significant.  A naive calculation of the framed BPS spectrum using a framed quiver produces an incorrect result.  By contrast, using our conjecture for the framed BPS states we are able to reproduce the expected semi-classical spectrum and OPE algebra, thus providing a nontrivial check on our proposal.

\medskip

Finally, continuing a thread from \cite{Gaiotto:2010be, Cecotti:2010fi, Cirafici:2013bha}, we should remark that many of the statements about line defects in this paper closely resemble statements which have appeared in the cluster literature.  Here is one example. The idea that cluster algebras should have particularly nice ``canonical bases'' was
part of the original motivation of the notion of cluster algebra \cite{MR1887642}, appears in some form in the work of Fock and Goncharov already mentioned \cite{MR2349682,MR2233852} and in many other places.  The canonical basis for a cluster algebra seems to be something close to the set of simple line defects in an $\N=2$ theory.
Now, in \cite{MR2295199} one finds the remark that the canonical basis is expected to contain all cluster monomials, i.e. all monomials of the form $x_1^{a_1} \cdots x_n^{a_n}$ where the $x_i$ are variables in a single cluster and all $a_i \ge 0$. This remark closely resembles one of the observations made in this paper, that at least in $\N=2$ theories of quiver type, there are some canonically defined subsets (cones) in the set of simple line defects, consisting of defects obtained by operator products of the form $L_1^{a_1} \cdots L_n^{a_n}$.  However, the precise details of the relation between these statements are not clear yet, at least to us.  It would be very interesting to understand the connection more precisely.

\section{BPS Line Defects}
\label{genlineops}

In this section we review the general theory of BPS line defects in four-dimensional $\N=2$ field theories.  These theories have a Coulomb branch where the IR physics is governed by an abelian gauge theory \cite{Seiberg:1994aj, Seiberg:1994rs}.  We let $u$ denote a point in the Coulomb branch and $\Gamma_{u}$ the total electromagnetic and flavor charge lattice.  This lattice is equipped with an antisymmetric bilinear pairing
\begin{equation}
 \langle \cdot, \cdot \rangle : \Gamma_{u} \times \Gamma_{u} \rightarrow \Z,
\end{equation}
which is the Dirac inner product on charges.  The physics is controlled by the central charge map
\begin{equation}
Z(u): \Gamma_{u} \rightarrow \C,
\end{equation}
which is a linear function of the charge variable.\footnote{In the following we frequently suppress the dependence of the charge lattice and central charge on the point $u$.}

In \S \ref{susyalg} we review the basic definition of BPS line defects and their associated Hilbert spaces of framed BPS states.

In \S \ref{abelian}, as a prelude to the discussion of defect renormalization group flow, we describe in more explicit terms the BPS line defects of free abelian gauge theories.

In \S \ref{rgdecompose}, we turn to a detailed study of defect renormalization group flow.  We introduce the map $\RG$ of equation \eqref{RGmapintro}, which assigns a low-energy label to each UV line defect by extracting the charge of its ground state.  We describe the behavior of the map $\RG$ explicitly for the familiar case of $SU(2)$ Yang-Mills theory where the UV line defects can be explicitly enumerated.  We also sketch some connections of $\RG$ to tropical geometry.

Finally, in \S \ref{OPEsec} we review the theory of line defect OPE algebras, and explain how to extract the OPE coefficients from the framed BPS spectrum.

\subsection{Definition of Supersymmetric Line Defects}
\label{susyalg}

Choose a rest frame, which we hold fixed for the rest of the paper. The $\N=2$ supersymmetry algebra contains a distinguished family of subalgebras $A_\zeta$, labeled by $\zeta \in \C^\times$.  Each $A_\zeta$ has real dimension four, and contains time translations but no other translations. See \cite{Gaiotto:2010be} for the generators.  $U(1)_R$ acts nontrivially on the set of subalgebras $A_\zeta$, taking $\zeta \mapsto e^{\I \alpha} \zeta$ (mirroring its action on the central charge by $Z \mapsto e^{\I \alpha} Z$).

For any $\zeta \in \C^\times$, define a \emph{$\zeta$-supersymmetric line defect} to be a defect in our theory which preserves the following symmetries.
\begin{itemize}
\item The supersymmetry algebra $A_\zeta$ (thus the defect is one-half BPS).
\item The spatial $SO(3)$ rotation symmetry about the point where the defect is inserted.
 \item The $SU(2)_R$ symmetry.
\end{itemize}

When $\abs{\zeta} = 1$, the algebra $A_\zeta$ and associated $\zeta$-supersymmetric line defects have a simple physical interpretation.  $A_\zeta$ consists exactly of those supercharges which annihilate a BPS particle at rest, and whose central charge satisfies $Z / \zeta \in \R_+$.   Hence if  $\abs{\zeta} = 1,$ we may interpret $\zeta$-supersymmetric line defects as representing heavy external BPS source particles. If $\abs{\zeta}\neq1$, the corresponding defects do not have this kind of interpretation.

\subsubsection{Defect Hilbert Space and Framed BPS States}
\label{framedhilbert}

Fix a $\zeta$-supersymmetric defect $L$ and a point $u$ of the Coulomb branch. Let $\HH_{L,u}$ denote the Hilbert space of the theory with the defect $L$ inserted. $\HH_{L,u}$ is a representation of the remaining supersymmetry $A_\zeta$.  The BPS bound in the presence of the defect is\footnote{Our conventions here differ from those  of \cite{Gaiotto:2010be} by a sign.} \cite{Gaiotto:2010be}
\begin{equation} \label{eq:bps-bound}
M \ge \re(Z / \zeta).
\end{equation}
States which are annihilated by all of $A_\zeta$, or equivalently states which saturate \eqref{eq:bps-bound}, are called \emph{framed BPS states} and make up a sub-Hilbert space $\HH^\BPS_{L,u}$.  The framed BPS spectrum is one of the main objects of study in this paper.

From the point of view of the original $\N=2$ theory, framed BPS states might be thought of as one-half BPS, just as the defect $L$ itself.  However, strictly speaking they are not states of the original theory at all, but rather states of the modified theory with $L$ inserted.  As a result, the framed BPS states do \emph{not} transform in representations of the super Poincar\'{e} generators broken by the defect $L$; those generators are simply not present in the theory with $L$ inserted.  Note in particular that the broken generators include spatial translations and thus the framed BPS states do not have any traditional interpretation in terms of particles.

As usual, the Hilbert space decomposes into sectors labeled by the total electromagnetic and flavor charge
\begin{equation} \label{eq:h-decomp}
 \HH_{L,u} = \bigoplus_{\gamma \in \Gamma_u} \HH_{L,u,\gamma}, \qquad \HH^\BPS_{L,u} = \bigoplus_{\gamma \in \Gamma_u} \HH^\BPS_{L,u,\gamma}.
\end{equation}
The framed BPS degeneracies may be assembled into an index, the \emph{framed protected spin character} defined in \cite{Gaiotto:2010be}.  Let $J_{3}$ denote a Cartan generator of the $SO(3)$ rotation group, and $I_{3}$ a Cartan generator of the $SU(2)_{R}$ symmetry.  The framed protected spin character is given by
\begin{equation}
\fro(L, \gamma,  u, y) = {\rm Tr}_{\HH_{L, u,\gamma}^{\rm BPS}} y^{2J_3} (-y)^{2I_3}. \label{frodef}
\end{equation}
This quantity is invariant under small deformations of the modulus $u$.  At walls of marginal stability, it jumps according to the wall-crossing formula \cite{KS1, KS2, KS3, JS, J}.\footnote{In \eqref{eq:h-decomp} and \eqref{frodef}, there is an unpleasant proliferation of subscripts and arguments. We drop some of them when we can do so without creating too much confusion.}

There is an important positivity phenomenon for the framed protected spin characters:  in a UV complete $\N=2$ theory in four dimensions, all framed BPS states should
have $I_3 = 0$.  This is the so-called ``no-exotics conjecture'' formulated in \cite{Gaiotto:2010be}, and proven for the case of $SU(n)$ super Yang-Mills in \cite{Chuang:2013wt}.

If the no-exotics conjecture is correct, the framed protected spin character \eqref{frodef} of any BPS line defect becomes an Laurent series in $y$ with non-negative integral coefficients where the coefficient of $y^n$ literally counts, \emph{without signs}, framed BPS states with angular momentum $\frac{n}{2}$ in the $x^3$ direction.  In particular, specializing to $y=1$ the framed protected spin character reduces to the dimension dimension of the framed BPS Hilbert space in a given charge sector:
\begin{equation}
\fro(L, \gamma,  u, y)|_{y=1}\equiv \fro(L, \gamma,  u)=\dim \HH_{L,u,\gamma}^{\BPS}.
\end{equation}
We assume the absence of exotics in the remainder of this paper, most crucially in our discussion of defect renormalization group flow in \S \ref{rgdecompose}.

\subsubsection{Classifying Line Defects}
\label{simpledefects}

In $\N=2$ super Yang-Mills, basic examples of supersymmetric line defects are provided by supersymmetric Wilson and 't Hooft lines, or more generally Wilson-'t Hooft lines; these defects are labeled by pairs $\alpha = (\lambda,\lambda^*)$ consisting of a weight and coweight modulo simultaneous Weyl reflection as described in \cite{Kapustin:2005py}.

For a general field theory, we will similarly introduce a discrete label $\alpha$ which keeps track of the type of line defect.  However, to get a useful classification we have to restrict our attention to special line defects, as emphasized in \cite{Gaiotto:2010be}. For example, in $\N=2$ super Yang-Mills with gauge group $G$, the line defects labeled by $\alpha = (\lambda, 0)$ are the Wilson lines corresponding to \emph{irreducible} representations of $G$, not \emph{arbitrary} representations.

We make a similar distinction in a general field theory. Thus we call a $\zeta$-supersymmetric line defect $L$ \emph{composite} if the theory with $L$ inserted splits up into superselection sectors, i.e. if there exist nontrivial line defects $L_1$ and $L_2$ such that
\begin{equation}
\inprod{L \cO} = \inprod{L_1 \cO} + \inprod{L_2 \cO}
\end{equation}
where $\cO$ represents any combination of $\zeta$-supersymmetric operators.  In this case we write simply
\begin{equation}
L = L_1 + L_2.
\end{equation}
An equivalent way of expressing the same concept is to note that the framed BPS Hilbert space is a representation of the $\zeta$-supersymmetric operators.  If the line defect $L$ is composite, then the associated representation is reducible and we write
\begin{equation}
 \HH_{L,u} = \HH_{L_1,u} \oplus \HH_{L_2,u}.
\end{equation}

We call $L$ \emph{simple} if it is not composite.  Our index $\alpha$ will run over all simple $\zeta$-supersymmetric line defects.  More precisely, we will assume that for each fixed choice of $\alpha$ there is a continuous family of simple line defects $L(\alpha, \zeta)$ parameterized by $\zeta \in \C^\times$,\footnote{Actually, in some asymptotically free theories there is monodromy acting on the family of line defects when $\zeta$ goes around zero \cite{Gaiotto:2010be}.  Thus, strictly speaking, in those theories $L(\alpha, \zeta)$ depends not only on $\zeta \in \C^\times$ but on a choice of $\log \zeta$.  We will suppress this dependence in the notation.} where $L(\alpha, \zeta)$ is $\zeta$-supersymmetric.   This statement can be verified directly in large classes of field theories, such as the theories of class $S[A_1]$ \cite{Gaiotto:2010be}.  As we describe in \S \ref{rgdecompose}, the existence of the $\mathbb{C}^{\times}$ family would also follow from the expected properties of the defect renormalization group flow.

\subsection{BPS Line Defects in Abelian Gauge Theories}
\label{abelian}

For a general $\N=2$ quantum field theory, it may difficult to explicitly describe the collection of simple BPS line defects parameterized by the label $\alpha$.  However, in the special case of an abelian gauge theory,  the discussion is much more concrete: simple supersymmetric line defects are parameterized by the lattice $\Gamma$ of electromagnetic and flavor charges.

Indeed, for every charge $\gamma \in \Gamma$, and every $\zeta \in \C^\times$, there is a corresponding abelian $\zeta$-supersymmetric Wilson-'t Hooft line defect $L(\gamma, \zeta)$.\footnote{In the general discussion above we used the letter $\alpha$ for the labels of line operators, while we now use $\gamma$; the reason for this change of notation will be apparent when we discuss RG flows in Section \ref{rgdecompose} below.} Moreover, in the case where $\gamma$ is a pure electric charge, it is easy to write the $\zeta$-supersymmetric Wilson line explicitly.  For notational convenience let us write it in the case of a rank one theory, where the bosonic sector consists of a single $u(1)$ gauge field $A$ and a complex scalar $\phi$. Then the supersymmetric Wilson line defect is
\begin{equation} \label{eq:susy-wilson-abelian}
 L(\gamma, \zeta) = \exp \, \left[i \gamma \int A + \frac{1}{2} (\zeta^{-1} \phi + \zeta \bar\phi) \right]
\end{equation}
where the integral runs over the time direction.  One checks directly that this defect is indeed invariant under the algebra $A_\zeta$.  More general Wilson-'t Hooft line defects can be obtained from this one by electric-magnetic duality.

In particular, \eqref{eq:susy-wilson-abelian} shows directly that in abelian theories line operators come in $\mathbb{C}^{\times}$ families parameterized by $\zeta$, as we proposed more generally above.

\subsubsection{Framed BPS States in Abelian Theories}
\label{bpsabelian}

We may also obtain an explicit description of the framed BPS Hilbert space in the presence of a dyonic defect $L(\gamma, \zeta)$.

As mentioned in \S \ref{susyalg}, when $|\zeta| = 1$ one way of thinking about $\zeta$-supersymmetric line defects is as heavy one-half BPS source particles whose central charge satisfies $Z/\zeta \in \mathbb{R}_{+}$.  In the case of an abelian theory we can realize this idea concretely as follows.  Let $T$ denote an abelian $\N=2$ theory, and $\gamma \in \Gamma$ a charge in that theory.  Fix an electric-magnetic duality frame for which $\gamma$ is a purely electric charge. Then let $T[\gamma, M\zeta]$ denote an augmented theory obtained by coupling $T$ to one hypermultiplet field $\psi$, with charge $\gamma$ and complex mass $M\zeta$, where $M$ is a real positive parameter.

The $\psi$ particle number is conserved in this augmented theory, and we focus on the sector consisting of one-$\psi$-particle states.  We consider the limit $M\rightarrow \infty,$ with $\zeta$ fixed.  In this regime, the $\psi$ particles become very massive, with central charge along the ray determined by $\zeta$. Therefore there is a correspondence
\begin{center}
\mbox{one-$\psi$-particle rest states of $T[\gamma, \infty\zeta]$} \ $\leftrightarrow$ \ \mbox{states of $T$ with $L(\gamma,\zeta)$ inserted}.
\end{center}
In particular, this gives an alternative way of viewing the framed BPS states in the theory $T$:  they are ordinary one-$\psi$-particle BPS states of $T[\gamma, \infty\zeta]$, at rest.  The term at rest means we diagonalize the translation generators, and also choose a particular component of the multiplet generated by the broken supersymmetry generators:  thus on both sides of the above, the states are only in representations of the $SO(3)$ group of rotations, not a super-extension thereof.

Finally, note that as $T$ is a free abelian theory and the augmented theory is arbitrarily weakly coupled, the Hilbert space is the perturbative Fock space and we can describe  the one-$\psi$-particle part explicitly.  There is a  single BPS state, carrying charge $\gamma$.  Thus our correspondence says that
\begin{equation}
\dim \HH_{L(\gamma,\zeta),u,\gamma'}^{\BPS} = \delta_{\gamma, \gamma'}. \label{defectabelianH}
\end{equation}
This relationship between framed BPS states in the presence of a defect and ordinary one-particle states of an augmented theory foreshadows our construction of framed BPS states in interacting $\N=2$ theories of quiver type in \S \ref{framing} below.

\subsection{Defect Renormalization Group Flow}
\label{rgdecompose}

Our aim in this section is to understand  how the renormalization group acts on BPS line defects. Thus we return to the situation of a general UV complete $\N=2$ theory, at a point $u$ on its Coulomb branch.  The IR physics is then described by an abelian gauge theory and possesses dyonic BPS line defects discussed in \S \ref{abelian}.

Fix a simple line defect $L(\alpha, \zeta)$ in the UV theory.  Flowing to the IR, we obtain a line defect of the IR theory. However, importantly, there is no reason why this IR line defect should be simple:  typically it will decompose as a nontrivial sum of simple line defects \cite{Gaiotto:2009hg}.  The reason for this is that the IR theory has fewer local operators than the UV, so it is easier for a line defect to be decomposable in the IR.  In other words, the theory with the line defect inserted may break up into superselection sectors in the IR, which in  the UV theory are coupled by massive fields.

A concrete example is provided by the $\zeta$-supersymmetric Wilson line defect attached to a representation $R$ of the gauge group $G$ in pure $\N=2$ gauge theory.  In the IR, $G$ is broken to $U(1)^r$, and the representation $R$ decomposes as a direct sum of $\dim(R)$ $1$-dimensional representations of $U(1)^r$. Correspondingly, the single UV Wilson line defect $L(\alpha, \zeta)$ decomposes (in the classical limit) into $\dim (R)$ IR Wilson line defects.\footnote{Interestingly, quantum mechanically, this result is modified:  even in the weakly coupled region of the Coulomb branch, the decomposition of a UV Wilson line defect can include IR Wilson-'t Hooft line defects as well as the expected IR Wilson defects \cite{Gaiotto:2010be}. See also \S \ref{wilsonlineex} below.} The operators which couple these IR line defects into a single UV line defect are the $W$ bosons, which are integrated out in flowing to the IR.

While $L(\alpha, \zeta)$ need not be simple in the IR, it can be expanded in terms of simple defects in the IR theory.  We claim that this expansion can be written in terms of the framed protected spin characters as
\begin{equation}
 L(\alpha, \zeta) \rightsquigarrow \sum_{\gamma \in \Gamma} \fro(\alpha, \gamma, \zeta, u) L(\gamma,\zeta). \label{eq:defect-expansion}
\end{equation}
This decomposition follows from matching dimensions of BPS Hilbert spaces in each charge sector.

To prove this statement, consider first the right-hand-side of the above.  We make use of \eqref{defectabelianH} which states that the framed BPS Hilbert space of the IR abelian theory in the presence of the defect $L(\gamma,\zeta)$ consists of a single state of charge $\gamma$.  Therefore, in the charge sector $\gamma,$ the framed BPS Hilbert space of the composite defect specified in \eqref{eq:defect-expansion} has dimension given by the coefficient of the defect  $L(\gamma,\zeta),$ namely $\fro(\alpha, \gamma, \zeta, u)$.

Similarly we can also consider left-hand-side of the claimed decomposition.  According to the absence of exotics assumption of \S \ref{framedhilbert} the quantity $\fro(\alpha, \gamma, \zeta, u)$ is the dimension of the BPS Hilbert space of the defect  $ L(\alpha, \zeta)$ in charge sector $\gamma$.  This dimension is defined in the UV.  However $\fro(\alpha, \gamma, \zeta, u)$ is also an index, and is invariant under RG flow.  Thus establishing \eqref{eq:defect-expansion}.

As is clear from the previous argument, we may also express the effect of the RG flow at the level of Hilbert spaces as
\begin{equation}
 \HH^\BPS_{L(\alpha, \zeta)} \rightsquigarrow \bigoplus_{\gamma \in \Gamma} \HH^\BPS_{L(\alpha,\zeta),u,\gamma } \otimes \HH^{\BPS}_{L(\gamma,\zeta)}.   \label{eq:hilbert-decomp}
\end{equation}
This expression points to a refined version of the decomposition \eqref{eq:defect-expansion} which keeps track of the spin content.  Indeed, the IR operators $L(\gamma,\zeta)$ are singlets under the group $SO(3)$ of spatial rotations.  However there is no reason why the framed BPS states of the UV operators $L(\alpha, \zeta)$ need to be singlets of $SO(3)$, and in general they are not.  Moreover, the $SO(3)$ content is also invariant under the RG flow, like the dimension.  Thus we expect that \eqref{eq:hilbert-decomp} continues to hold even if we regard  both sides as $SO(3)$ representations.  We can then uplift to a spin-sensitive version of the decomposition \eqref{eq:defect-expansion} by making use of the full data of the framed protected spin characters:
\begin{equation}
 L(\alpha, \zeta) \rightsquigarrow \sum_{\gamma \in \Gamma} \fro(\alpha, \gamma, \zeta, u, y) L(\gamma,\zeta).   \label{eq:ref-defect-expansion}
\end{equation}

In \eqref{eq:defect-expansion}, the equivalence relation $\rightsquigarrow $ specified by RG flow implies that the given defects are are equal in any correlation function of $\zeta$ supersymmetric operators of the IR abelian theory.  Similarly, in \eqref{eq:ref-defect-expansion} the defects are equal in any correlation function of $\zeta$ supersymmetric operators of the IR abelian theory with an additional explicit insertion of $y^{2J_{3}}(-y)^{2I_{3}}.$

\subsubsection{IR Labels for Line Defects}
\label{troplicallabels}

In this section, we introduce one of the fundamental concepts of the defect renormalization group flow, namely an IR labeling scheme for UV defects.

As we have just reviewed, a UV line defect $L(\alpha, \zeta)$ generally decomposes at low energy into a sum of IR line defects $L(\gamma,\zeta)$.  There is however a sense in which one of the terms in this decomposition is dominant, namely the $L(\gamma, \zeta)$ for which $\Re(Z_\gamma / \zeta)$ is smallest.  One way to understand this dominance is to note that according to the BPS bound \eqref{eq:bps-bound}, framed BPS states obey
\begin{equation}
E = \re(Z / \zeta). \label{Elabels}
\end{equation}
Thus the IR line defect with smallest $\re(Z / \zeta)$ corresponds to the framed BPS state with minimum energy, which may be viewed as the global ground state of the original UV  defect $L(\alpha, \zeta)$.

We let $\RG(\alpha,\zeta,u)$ denote the IR charge $\gamma$ of the ground state of $L(\alpha, \zeta)$.  This defines a map
\begin{equation}
\RG(\cdot, \zeta, u): \{\mathrm{UV \ Line \ Defects}\}\rightarrow \{\mathrm{IR \ Line \ Defects}\}=\Gamma, \label{RGmapbody}
\end{equation}
which assigns to each UV line defect, an IR label determined by the simple dyonic defect associated to its ground state.

For a given UV defect $L(\alpha, \zeta),$ we sometimes refer to the image $\RG(\alpha,\zeta,u)$ as the \emph{core charge} of the defect.  This terminology anticipates a physical picture developed in \S \ref{framing} where an IR observer views the defect as an infinitely massive dyonic core charge, corresponding to the dominant term in the expansion \eqref{eq:hilbert-decomp}, and the remaining terms  in \eqref{eq:hilbert-decomp} are viewed as excitations.

Now let us assume that the point $(\zeta,u)$ is generic.  It was observed in \cite{Gaiotto:2010be} that in many $\N=2$ theories the $\RG(\cdot, \zeta, u)$ map has the following properties.
\begin{itemize}
\item For a given UV defect $L(\alpha, \zeta),$ the ground state is unique and carries no spin, i.e. the corresponding coefficient $\fro(\alpha, \gamma, \zeta, u, y) $ equals one.

\item The map $\RG(\cdot,\zeta,u)$ is a bijection between the set of UV line defect labels $\alpha$ and the set of IR line defect labels $\gamma.$
\end{itemize}
The bijectivity of the $\RG$ map is particularly striking.  It implies that a UV line defect $L(\alpha, \zeta)$ can be reconstructed uniquely from its IR ground state $\RG(\alpha, \zeta,u).$  More physically, this means that at least for supersymmetric defects the phenomenon of screening does not occur: any two UV defects can be distinguished by low-energy observations.  Thus, for these defects, the renormalization group flow, which is typically irreversible, is in fact reversible.

In more practical terms, the bijectivity of the $\RG$ map has several useful consequences.  First, it implies the conjecture we mentioned in \S \ref{simpledefects}, that UV line defects come in families parameterized by $\zeta\in \mathbb{C}^{\times}$ (since we have already noted in \S \ref{abelian} that IR ones do).  Second, it implies that one can reconstruct a UV line defect $L(\alpha)$ from the collection of IR line operators into which it decomposes, or equivalently from its framed BPS spectrum.  Third, it is a fundamental ingredient in any approach to computing the framed BPS states of UV line defects using IR physics.

The above discussion is accurate provided that we remain at a fixed generic point $(\zeta,u)$.  However if the parameters $(\zeta,u)$ are varied we encounter an important subtlety in the $\RG$ map between UV and IR defect labels.  The IR label of a given UV defect $L(\alpha, \zeta)$ may jump discontinuously at those loci in parameter space where the ground state of the defect becomes degenerate.  As the energy of each ground state is given by the saturated BPS bound \eqref{Elabels}, such an exchange of dominance would occur when $\re(Z_{\gamma_1} / \zeta) = \re(Z_{\gamma_2} / \zeta)$, or letting $\gamma = \gamma_1 - \gamma_2$, $\re(Z_\gamma / \zeta) = 0$.  Thus, define an \emph{anti-wall} to be a locus in the $(u, \zeta)$ parameter space where some $\gamma \in \Gamma$ has $\re(Z_\gamma / \zeta) = 0$.  The map $\RG(\cdot, \zeta, u)$ is piecewise constant away from anti-walls, but can jump at anti-walls.

We should emphasize that this is \emph{not} the framed wall-crossing phenomenon: anti-walls are where some $\re(Z_\gamma / \zeta) = 0$; in contrast, framed wall-crossing occurs at the walls where  some $\im(Z_\gamma / \zeta) = 0$ for some $\gamma$.  At the latter type of  walls some framed BPS states can appear  or disappear, but these are never the ones of lowest energy, so the ground state does not change.\footnote{In the context of  the asymptotic expansion of line defect vevs, the anti-walls are anti-Stokes lines, while walls are Stokes lines; the two play very different roles in the theory.}

If we want to understand UV line operators using IR computations, it would be useful to understand how this jumping of labels occurs.  A partial result in this direction was obtained in \cite{Gaiotto:2010be} in theories of class $S[A_1]$.  Specifically, consider the restricted set of anti-walls where the charge $\gamma$ satisfying $\re(Z_\gamma / \zeta) = 0$ has $\Omega(\gamma)=1$.  Let $\RG(\alpha)_{\pm}$ denote the value of the $\RG$ map on either side of the anti-wall, with $+$ sign indicating the region where $\re(Z_\gamma / \zeta)>0$.  Then the discontinuity takes the form of a universal piecewise linear transformation on the charge lattice $\Gamma$:
\begin{equation}
\RG(\alpha)_{-} = \RG(\alpha)_{+} + \mathrm{Max}\{ \langle \gamma , \RG(\alpha)_{+}\rangle,0\}\gamma. \label{tropicallabeljump}
\end{equation}
An analogous formula for the jump
at more general anti-walls or in more general theories is not known, and it would be very interesting to develop a more general understanding of this feature of
the $\RG$ map.  We provide some hints in this direction in \S \ref{puresu2trop}.

In \S\ref{quiverconstruct}-\ref{examples} below, we find evidence that the properties of $\RG$ enumerated above, namely uniqueness of the ground state, and bijectivity, are true for every $\N=2$ theory of quiver type.  More specifically, we will assume these properties and see that a consistent picture emerges, allowing us in particular to reconstruct physically well behaved candidate UV defect vevs from IR calculations.  In that context, we will also obtain an independent derivation of the discontinuity formula \eqref{tropicallabeljump}.

\subsubsection{The Case of Pure $SU(2)$ Gauge Theory}
\label{puresu2trop}

A useful explicit example of the $\RG$ map occurs in the case of Wilson line defects in pure $\N=2$ gauge theory with gauge group $G = SU(2)$.  Such defects are classified by a finite-dimensional representation $R$ of $SU(2)$.

Let $a$ denote the central charge of a purely electric state;
$a$ is a function of the Coulomb branch
parameter $u$.  In \cite{Gaiotto:2010be} it was shown that when $R$ is the fundamental representation, at  any $(u,\zeta)$ for which $\Re(a/\zeta) \neq 0$, the Wilson line defect supports three framed BPS states.\footnote{In \S \ref{wilsonlineex} below, we derive this result independently using quiver quantum mechanics.} Two of these carry purely electric $U(1)$ charge and correspond to the two weights of the representation $R$, i.e. they simply come from  decomposing $R$ under $U(1) \subset SU(2)$.  The remaining framed BPS state is a dyon, whose electric charge depends on $(u,\zeta)$:  there are various walls in the $(u,\zeta)$ parameter space where one dyonic bound state decays and simultaneously another, with a different charge, forms. Moreover this dyon is always heavier than the purely electric states. Thus the lightest framed BPS state is one of the two electrically charged ones, corresponding to the ``lowest weight'' of the representation $R$.  A similar story holds for a general representation $R$ (and presumably also for other gauge groups $G$, although it has not been investigated in detail.)

As we have remarked in the previous section, in many $\N=2$ theories --- including this one --- for generic $(u,\zeta)$ the map $\RG$ is a bijection between the sets of UV line defect labels $\alpha$ and IR charges $\gamma$. What we have just described is that, when restricted to the set of Wilson line defects, $\RG$ becomes the map which assigns to each representation $R$ a purely electric charge $\gamma$ corresponding to the lowest weight of $R$. Note however that this does not exhaust the set of purely electric charges in the IR theory: rather, it covers only one Weyl chamber, i.e. only the positive electric charges (in some basis).   The negative charges are also in the image of the $\RG$ map, but do not correspond to Wilson lines.

So far we have discussed what happens at generic $(u,\zeta)$, which is all we will use in the rest of the paper.  As an aside, though, it is worth noting an interesting phenomenon which occurs at the locus $\Re(a/\zeta) = 0$ in the weakly coupled region.  This locus lies on an anti-wall, so we would expect that the $\RG$ map is not defined exactly there.  What if we look very near this anti-wall, on one side or the other? Naively there is a further difficulty:  there is a collection of infinitely many anti-walls which accumulate here, corresponding to the infinite set of BPS dyons in the weakly coupled region \cite{Seiberg:1994rs}.  Thus any finite point will always be separated from $\Re(a/\zeta) = 0$ by some anti-walls (in fact infinitely many of them). Nevertheless, let us use the notation $(u, \zeta)_+$ formally to represent a point ``infinitesimally close'' to the anti-wall $\Re(a/\zeta) = 0$, which we reach by crossing the infinite chain of dyon anti-walls but stopping short of $\Re(a/\zeta) = 0$.  We will try to make sense of the $\RG$ map at the formal point $(u,\zeta)_+$.

One approach to defining it is to start at some generic $(u,\zeta)$ --- where we understand the $\RG$ map ---and then compose with the infinite chain of jumps \eqref{tropicallabeljump} corresponding to the infinite chain of anti-walls. Introduce electric and magnetic charges $(e,m)$ around weak coupling.  The electric charge $e$ is an arbitrary integer, while the magnetic charge $m$ is an arbitrary even integer.\footnote{In these conventions, the Dirac pairing takes the form $\langle (e, m), (e',m')\rangle=\frac{1}{2}(em'-e'm)$.}  The charges of the infinite collection of dyons are of the form $(2n, -2)$ for $n \ge 0.$

Consider a UV line defect with IR label $(e,m).$  At the anti-wall corresponding to the $n$-th dyon, the IR labels of line defects jump by \eqref{tropicallabeljump}
\begin{equation}
(e,m) \mapsto (e,m) + \mathrm{Max}\{e+nm,0\} (2n,-2).
\end{equation}

Now we may ask whether the composition of these infinitely many transformations is convergent. It turns out that it is, but it is not a bijection: rather it maps the whole charge lattice onto the union of the half-space, where $m < 0$ and $e$ is arbitrary, with the ray, where $m= 0$ and $e \leq 0$. Thus it seems to make sense to define the $\RG$ map at the point $(u, \zeta)_+$, but unlike what happens at ordinary points $(u, \zeta)$, here $\RG$ is not a bijection.

Moreover, at the point $(u, \zeta)_+$ the image of  $\RG$ agrees exactly with the usual UV labeling of line defects.  Indeed  from the UV point of view, 't Hooft-Wilson lines are labelled by a pair  $(e,m) \in \Z \times 2\Z$ modulo the simultaneous Weyl equivalence $(e,m) \sim (-e,-m)$ \cite{Kapustin:2005py}. Thus the $\RG$ map at $(u,\zeta)_+$ just maps such equivalence classes into a preferred fundamental domain.

Now what will happen if we cross the anti-wall at $\Re(a/\zeta) = 0$, passing to a point $(u,\zeta)_-$  which is again ``infinitesimally close'' but now on the other side?  We do not have the tools to study this directly, but we can again approach it indirectly, by moving through the $(u,\zeta)$ parameter space along a path which goes into the strong coupling region and back again out  to weak coupling.  If we do so, we find a simple result:  the IR labels at $(u,\zeta)_+$ and  $(u,\zeta)_-$ are related by a transformation which preserves the open half-space $(e,m)$ with $m < 0$ but acts on the ray $(e,0)$ with $e \leq0, $ as  $(e,0) \mapsto (-e,0)$.\footnote{This reflection phenomenon was also noticed by D. Gaiotto.} One might think of this as a kind of Weyl reflection:  it acts nontrivially only on the IR labels of pure Wilson defects, and on these it swaps the highest and lowest weights of the corresponding $SU(2)$ representations.

\subsubsection{Tropical Duality}
\label{tropical}

A formula very similar to \eqref{tropicallabeljump} arose in mathematical work of Fock and Goncharov \cite{MR2233852,MR2349682}, which seems to be closely related to the physics of line defects. We sketch this relationship here, but leave a more complete analysis to future work.

Fock and Goncharov studied a particular class of complex spaces called cluster varieties. A cluster variety over $\C^\times$ is a space obtained by gluing together complex tori, each isomorphic to $(\C^\times)^n$, using transition functions of a particularly simple form.  In particular, the transition functions involve only the operations $+$, $\times$ and $/$, but not the operation $-$.

As a result of these restricted transition functions, one can define various different versions of cluster varieties by replacing $\C^\times$ in the above construction by other algebraic structures.  We denote these spaces by $\fX[S]$ where $S$ is any of the possible places where the coordinates are valued.  For instance, the case of $\C^{\times}$ defined above leads to a cluster variety $\fX[\C^{\times}]$.  One may define a real version of the cluster variety, $\fX[\R^{\times}]$, by using coordinates valued in $\R^\times$.  Similarly, one may define a positive version, $\fX[\R^\times_+]$, by using coordinates valued in $\R^\times_+$.

More exotically, one may also consider the coordinates to take values in a semifield, an algebraic structure where the $-$ operation does not even exist.  Relevant examples are the \emph{tropical} semifields $\R_t,$ and $\Z_t$.  As sets these have the same elements as $\R \cup \{-\infty\}$ and $\Z \cup \{-\infty\}$ respectively, but they have addition, $\oplus,$ and multiplication, $\otimes,$ defined by
\begin{equation}
a \oplus b =\mathrm{Max}\{a,b\}, \hspace{.5in} a \otimes b =a+b.
\end{equation}
The transition functions for a cluster variety make sense for coordinates valued in $\R_{t}$ and $\Z_{t}$ if we replace ordinary addition by $\oplus,$ and ordinary multiplication by $\otimes$.

Tropical semifields arise naturally when studying the behavior geometric spaces at large values of their coordinates.  For instance, consider a function $f(a_{1}, \cdots, a_{n})$ of ordinary real or complex variables $a_{i}$ which can be written without subtraction operations.  Then, one obtains a tropical form of the function, $f_{t},$ which is a function of tropical variables $\mathfrak{a}_{i},$ by the limiting procedure
\begin{equation}
f_{t}(\mathfrak{a}_{1}, \mathfrak{a}_{2}, \cdots, \mathfrak{a}_{n})\equiv \lim_{\epsilon \rightarrow \infty} \frac{1}{\epsilon} \log \left[f \left(\exp(\epsilon \mathfrak{a}_{1}), \exp(\epsilon \mathfrak{a}_{2}), \cdots, \exp(\epsilon \mathfrak{a}_{n})\right)\right]. \label{troplim}
\end{equation}
This limit has the effect of replacing ordinary addition by $\oplus,$ and ordinary multiplication by $\otimes.$  In particular, given any cluster variety $\fX[\C^\times],$ the tropical cluster space $\fX^\vee[\R_t]$ may be thought of as parametrizing points at infinity of $\fX[\C^\times].$

Returning to the general theory, Fock and Goncharov proposed that given any cluster variety $\fX[\C^\times]$ there should exist a canonical vector space basis for the ring of global regular functions on $\fX[\C^\times]$.  Moreover this basis should be naturally parameterized by the integer tropical points of a different cluster variety $\fX^\vee$, i.e. by $\fX^\vee[\Z_t]$. A similar statement is also supposed to hold with the roles of $\fX$ and $\fX^\vee$ interchanged. Thus we have a kind of \emph{tropical duality} between $\fX$ and $\fX^\vee$.  Moreover, the relation between $\fX[\C^\times]$ and $\fX^\vee[\C^\times]$ is expected to be some version of mirror symmetry \cite{ghki,goncharovtoappear}.

What does this have to do with four-dimensional $\N=2$ supersymmetric field theories?  Suppose we have such a theory and a given simple line defect $L(\alpha)$.  We consider compactifying on a two-torus, with $L(\alpha)$ wrapped on one of the cycles, say the $A$ cycle.  This compactification can be viewed in two different ways; comparing the two will lead to the desired duality.

Suppose we first compactify on the $A$ cycle and flow to the IR, to get a three-dimensional theory.   As the original defect wrapped the $A$ cycle, it descends to a local operator $O(\alpha)$.  The resulting theory can be described as a sigma model into a hyperkahler space \cite{Seiberg:1996nz}.   This space has the same essential features as a cluster variety, and is provably a cluster variety in many cases \cite{Gaiotto:2008cd}. Indeed, it has patches with $\C^{\times}$ valued coordinates $X_{\gamma}$, and the changes of coordinate maps are given by Kontsevich-Soibelman symplectomorphisms, the simplest of which is a cluster transformation:
\begin{equation}
 X_{\gamma'} \to X_{\gamma'}(1 + X_\gamma)^{\inprod{\gamma,\gamma'}}. \label{kstrans}
\end{equation}
With this in mind we call this space $\fX[\C^\times]$. The expectation value of the supersymmetric local operator $O(\alpha)$ is represented in the sigma model by a distinguished holomorphic function $F(\alpha)$ on $\fX[\C^\times]$. Upon further compactification on the $B$ cycle we arrive finally at a two-dimensional sigma model into $\fX[\C^\times]$.

Now let us begin again in four dimensions and perform the compactification in the opposite order:  so we begin by compactifying on the $B$ cycle and flowing to the IR. We again get a three-dimensional theory, this time with a line defect $\tilde L(\alpha)$. The three-dimensional theory can again be described as a sigma model into a hyperkahler space, again a cluster variety, which we call $\fX^\vee[\C^\times]$. The spaces, $\fX$ and $\fX^\vee$ are in general distinct.  For instance, if the $A$ and $B$ cycles have different radii, then  $\fX$ and $\fX^\vee$ are not even isometric.

The line defect $\tilde L(\alpha)$ can be thought of as imposing some boundary condition on the sigma model fields, requiring that they have a singularity as we approach the defect.  Roughly speaking, such  a singularity condition amounts to specifying a particular direction to infinity in $\fX^\vee[\C^\times]$.  On the other hand, as sketched in \eqref{troplim}, the tropical cluster space $\fX^\vee[\R_t]$ parameterizes such directions to infinity.  Thus it seems reasonable to propose that the line defect $\tilde L(\alpha)$ will indeed correspond to a point of $\fX^\vee[\Z_t]$. Upon further compactification on the $A$ cycle we will again arrive at a two-dimensional sigma model, this time into $\fX^\vee[\C^\times]$.

Assuming this proposal is correct, then we see that a simple line defect $L(\alpha)$ induces on the one hand a canonical function $F(\alpha)$ on $\fX[\C^\times]$, and on the other hand a point of $\fX^\vee[\Z_t]$. Moreover, it is known that the two orders of compactification give rise to sigma models which are mirror to one another \cite{Bershadsky:1995vm, Harvey:1995tg}.  Thus $\fX[\C^\times]$ and $\fX^\vee[\C^\times]$ are mirror to one another.  This matches well with the expected picture of tropical duality described by Fock and Goncharov; moreover our field-theoretic way of realizing this picture is similar to one described in \cite{goncharovtoappear}.

Concretely, the singularity of the IR fields near a line defect $L(\alpha)$ should be determined by its IR label $\RG(\alpha)$.  Thus what we are saying is that the IR label $\RG(\alpha)$ of a given line defect $L(\alpha)$ should be globally understood as being  valued in a tropical cluster space $\fX^\vee[\Z_t]$.  (This  should be viewed as a slight elaboration of a proposal which appeared in \cite{Gaiotto:2010be}). This tropical space has charts each of which identifies the set of line defects  with the charge lattice $\Gamma$, and the formula  \eqref{tropicallabeljump} for the discontinuity of the $\RG$ map gives a change of coordinates from one chart to another.

\subsection{Line Defect OPE Algebra}
\label{OPEsec}

In this section, we review the definition of the OPE algebra of line defects and its connection to framed BPS states following \cite{Gaiotto:2010be}.  In the case of $\N=4$ gauge theory, this algebra was defined in \cite{Kapustin:2006pk}.  It has been further studied in a variety of contexts in \cite{Kapustin:2006hi, Drukker:2009tz, Gomis:2009ir, Gomis:2009xg, Gaiotto:2010be, Cecotti:2010fi, Ito:2011ea, Saulina:2011qr, Moraru:2012nu, Xie:2013lca}.

Suppose we have two $\zeta$-supersymmetric line defects $L = L(\alpha,\zeta)$ and $L' = L(\alpha',\zeta)$, located at two points $x$, $x'$ of the spatial $\R^3$. This combined configuration preserves the supersymmetry algebra $A_\zeta$ and a $U(1)$ rotational symmetry around the axis connecting $x$ to $x'$. Thanks to the supersymmetry, the correlation functions between $L$, $L'$ and other $A_\zeta$-invariant defects are independent of $x'$ \cite{Gaiotto:2010be}. In particular the limit $x' \to x$ is actually nonsingular and independent of how $x'$ approaches $x$ (since any two directions can be continuously connected while keeping $x' \neq x$). Moreover, in this limit the rotational symmetry is enhanced to $SO(3)$, so we have obtained a new $\zeta$-supersymmetric line defect, which we could call either $LL'$ or $L'L$.

The defect $LL'$ may be composite, but we can expand it in terms of simple line defects as
\begin{equation} \label{eq:ope}
L(\alpha,\zeta) L(\alpha',\zeta) = \sum_{\beta} c(\alpha,\alpha',\beta) L(\beta, \zeta).
\end{equation}
Let us assume moreover that this expansion is unique.  If so, then the coefficients $c(\alpha,\alpha',\beta)$ are non-negative integers and can be interpreted as structure constants for OPE of line defects. We emphasize that they are defined in the UV and in particular do not depend on a point $u$ of the Coulomb branch.  Moreover, since the defects depend continuously on $\zeta$, it follows that the structure constants do as well; on the other hand they are integers, and hence must actually be independent of $\zeta$.

One way to understand the uniqueness and positivity result for the structure constants $c(\alpha,\alpha',\beta)$ is to uplift the defect OPE to the level of Hilbert spaces.  The configuration with both $L(\alpha)$ and $L(\alpha')$ inserted is independent of the separation between the defects.  Thus, the BPS Hilbert space associated to the insertion of the two defects is the tensor product of the individual Hilbert spaces.   The Hilbert space version of the defect OPE then asserts that this tensor product may be expressed uniquely as a direct sum.
\begin{equation}
\HH_{L(\alpha)}^{BPS}\otimes \HH_{L(\alpha')}^{BPS}=\bigoplus_{\beta} N_{\alpha, \alpha'}^{\beta}\otimes \HH_{L(\beta)}^{BPS}. \label{vectorope}
\end{equation}
The ``coefficients," $N_{\alpha, \alpha'}^{\beta}$ appearing in the above are now themselves vector spaces.  There is no known direct definition of these coefficient spaces.  It would be interesting to formulate one, perhaps by studying the field theory in the presence of three defects.

Assuming the existence of the vector space valued OPE coefficients, we can now motivate the positivity of the structure constants $c(\alpha,\alpha',\beta).$  The spaces $N_{\alpha, \alpha'}^{\beta}$ should be representations of the $SU(2)_{R}$ symmetry group and the structure constants $c(\alpha,\alpha',\beta)$ should be extracted as an index
\begin{equation}
c(\alpha,\alpha',\beta)=\mathrm{Tr}_{N_{\alpha, \alpha'}^{\beta}}(-1)^{2I_{3}}.
\end{equation}
Then, if a no-exotics conjecture a la \S \ref{framedhilbert} holds for the $N_{\alpha, \alpha'}^{\beta}$, i.e. if these spaces are all trivial representations of the $R$-symmetry, we obtain the simplification
\begin{equation}
c(\alpha,\alpha',\beta)=\mathrm{dim}N_{\alpha, \alpha'}^{\beta},
\end{equation}
which is indeed non-negative.

As usual, the a good starting point for understanding the operator products is the case of abelian gauge theory.  In that situation, \eqref{eq:ope} collapses to
\begin{equation} \label{eq:ope-simple}
 L(\gamma, \zeta) L(\gamma', \zeta) = L(\gamma + \gamma', \zeta).
\end{equation}
In other words, the structure constants are simply given by the Kronecker delta,
\begin{equation}
c(\gamma, \gamma', \gamma'') = \delta_{\gamma + \gamma', \gamma''}.
\end{equation}
This equation merely states that for the dyonic defects of the abelian theory the charge of the source is additive, a statement usually referred to a Gauss' law.  We can mimic the algebra of these abelian defects by introducing a collection of formal variables $X_{\gamma}$ for $\gamma \in \Gamma$ and imposing the multiplication law\footnote{The notational similarity between these $X_{\gamma}$ and the cluster coordinates of \eqref{kstrans} is not an accident.  Indeed the cluster coordinates $X_{\gamma}$ are the expectation values of the dyonic defects of the abelian theory \cite{Gaiotto:2010be}.}
\begin{equation}
X_\gamma X_{\gamma'} = X_{\gamma + \gamma'}. \label{algebraxg}
\end{equation}

Using these results for abelian defects, together with the defect renormalization group flow of \S \ref{rgdecompose}, we can give a scheme for computing operator products between line defects in a general $\N=2$ theory.  To each defect $L$ we associate a generating function in  the formal variables $X_\gamma$
\begin{equation}
F(L) = \sum_\gamma \fro(L, \gamma) X_\gamma, \label{generatorcommute}
\end{equation}
which is a generating functional version of the defect decomposition of $L$ stated in \eqref{eq:defect-expansion}.

As remarked above the correlation functions of $\zeta$-supersymmetric operators in the presence of a pair of BPS defects $L$ and $L'$ are in fact independent of the separation between the defects.  Thus, the OPE can be computed in the UV by going to short distances, or in the IR abelian theory by going to long distances.  In particular, it follows that the generating functionals must obey
\begin{equation} \label{eq:FL-ope}
F(LL') = F(L) F(L').
\end{equation}
As a nice consistency check, observe that as a consequence of the framed wall-crossing phenomenon, the generating functional \eqref{generatorcommute} jumps at walls of marginal stability.  However, according the wall-crossing formula of \cite{KS1, KS2}, the transformation of the indices $\fro(L,\gamma)$ is an automorphism of the algebra \eqref{algebraxg}.   Thus, the algebra obeyed by the generating functionals is in fact an invariant of the UV field theory.

Assuming the bijectivity of the map $\RG$ described in \S \ref{troplicallabels}, we can reconstruct $LL'$ from $F_{LL'}$. Thus, if we are able to compute $F_L$ for every $L$, \eqref{eq:FL-ope} gives an algorithm for computing the operator product coefficients.  This method was used in \cite{Gaiotto:2010be} for some theories of class $S[A_1]$; we will use it \S \ref{explicitspectra} below as well.

\subsubsection{Noncommutative Defect OPEs}
\label{OPEnoncomm}

The OPE multiplication discussed in the previous section admits an important noncommutative deformation, as follows.  Fix an axis in the spatial $\R^3$, say the $x^3$-axis. Suppose that we are interested in computing correlation functions between $\zeta$-supersymmetric line defects $L$ and $L'$ as before, but we restrict all of these defects to be inserted only along this axis.  In that case the theory with the defects inserted has a $U(1)$ rotation symmetry around this axis; let $J_3$ denote the generator of this $U(1)$, and as before let $I_3$ denote one of the rotation generators in $SU(2)_R$.

Now we consider correlation functions between $\zeta$-supersymmetric line defects with an extra overall insertion of the operator $(y)^{2J_3}(-y)^{2 I_3}.$  All correlation functions will be promoted to functions of $y$, which, when specialized to $y $ equals one, reduce to the previous case with no extra insertion. Such correlation functions are still independent of the points where the defects are inserted.  Thus in particular the limit $x' \to x$ along the $x^3$-axis is nonsingular as before, but now it may depend on the direction in which $x' \to x$, since there are only two possible directions along the $x^3$-axis, and no way to interpolate from one to the other.

We thus obtain a noncommutative OPE, where the order in which we take the product corresponds to the ordering along the $x^3$-axis:
\begin{equation} \label{eq:ope-nc}
L(\alpha,\zeta) * L(\alpha',\zeta) = \sum_{\beta} c(\alpha,\alpha',\beta,y) L(\beta, \zeta)
\end{equation}
This is a noncommutative deformation of the product \eqref{eq:ope}.  The coefficients $c(\alpha,\alpha',\beta,y) $ are now valued in $\Z_{\ge 0}[y,y^{-1}]$.  They are evidently symmetric under the simultaneous exchange $\alpha \leftrightarrow \alpha'$, $y \leftrightarrow y^{-1}$.

In terms of the putative vector space valued OPE coefficients $N_{\alpha, \alpha'}^{\beta}$ appearing in \eqref{vectorope}, the $y$-dependent structure constants $c(\alpha,\alpha',\beta,y)$ should be obtained from a protected spin character
\begin{equation}
c(\alpha,\alpha',\beta,y)=\mathrm{Tr}_{N_{\alpha, \alpha'}^{\beta}}(y)^{2J_{3}}(-y)^{2I_{3}}.
\end{equation}
And again a no-exotics type conjecture motivates the fact that the structure constants are valued in $\Z_{\ge 0}[y,y^{-1}].$

In abelian theories, just as in the commutative case, the OPE is particularly simple.  It is the noncommutative analog of this simple result \eqref{eq:ope-simple}, derived in \cite{Gaiotto:2010be}
\begin{equation}
 L(\gamma, \zeta) * L(\gamma', \zeta) = y^{\inprod{\gamma,\gamma'}} L(\gamma + \gamma', \zeta),
\end{equation}
or in other words,
\begin{equation}
c(\gamma, \gamma', \gamma'', y) = y^{\inprod{\gamma,\gamma'}} \delta_{\gamma + \gamma', \gamma''}.
\end{equation}
This expression reflects the fact that, in comparing the configuration of a pair of separated defects $ L(\gamma, \zeta)$ and $ L(\gamma', \zeta),$ to the situation with the single defect $ L(\gamma + \gamma', \zeta),$ there is a quantized angular momentum carried in the electromagnetic fields and pointing along the axis separating the pair of sources.

As in the commutative case it is natural to study these OPEs by attaching a generating function to each line defect,
\begin{equation}
F_L = \sum_\gamma \fro(L, \gamma;y) X_\gamma,
\end{equation}
where now the product law is a noncommutative Heisenberg algebra
\begin{equation}
X_\gamma * X_{\gamma'} = y^{\inprod{\gamma,\gamma'}} X_{\gamma + \gamma'}. \label{heisenberg}
\end{equation}
As before, the operator product can be computed either in the UV or IR and we obtain the relation
\begin{equation} \label{eq:FL-ope-nc}
F_{LL'} = F_L * F_{L'}.
\end{equation}
Thus, if we are able to compute $F_L$ for every $L$, \eqref{eq:FL-ope-nc} gives an algorithm for computing the $y$-deformed operator product coefficients.  We study examples of this noncommutative algebra in \S \ref{explicitspectra}.

\section{Framed BPS States in Theories of Quiver Type}
\label{quiverconstruct}

In this section we introduce a technique for studying line defects in a large class of $\mathcal{N}=2$ theories: those of quiver type \cite{Douglas:1996sw, Douglas:2000ah, Douglas:2000qw, Alim:2011kw}.  Thus for the remainder of this paper we assume that the (unframed) BPS states of the theory may be described by non-relativistic quiver quantum mechanics.  Not all $\N=2$ theories admit such a simple description of their BPS states. However, the class of quiver type theories is quite broad and includes gauge theories with hypermultiplet matter, Argyres-Douglas type conformal field theories, and theories described by M5-branes on punctured Riemann surfaces  \cite{Fiol:2000pd, Fiol:2000wx, Fiol:2006jz, Cecotti:2011rv, Cecotti:2011gu, Alim:2011ae,    DelZotto:2011an, Xie:2012dw, Cecotti:2012va,  Cecotti:2012sf, Cecotti:2012gh, Saidi:2012gi, Cecotti:2012jx, Xie:2012jd, Cecotti:2012se, Cecotti:2012kv, Cecotti:2013lda, Cecotti:2013sza, Galakhov:2013oja}.

We begin in \S \ref{quivereview} by briefly recalling the quiver theory of ordinary BPS states.  Particularly important is the existence of two distinct descriptions of such states: the Higgs branch, associated to quiver representation theory, and the Coulomb branch, associated to quantization of the classical space of supersymmetric multi-centered particle configurations.

In \S \ref{framing} we describe our proposal \eqref{framedquivconj} for computing the framed BPS states of a line defect in theories of quiver type: such states are to be computed using the Coulomb branch of an extended particle system.  We utilize a recent localization formula \cite{Manschot:2011xc, Manschot:2012rx, Manschot:2013sya} to give precise meaning to our conjecture.

In \S \ref{framedquivers} we describe a related construction of an extended framed quiver.  We further state circumstances under which our conjectured multi-centered bound state formula for framed BPS degeneracies reduces to a more familiar quiver representation theory problem utilizing the framed quiver.

\subsection{Quiver Review}
\label{quivereview}

Fix an $\mathcal{N}=2$ theory and a vacuum $u$ on its Coulomb branch.  The Hilbert space of the theory supports BPS states carrying electromagnetic charges $\gamma \in \Gamma$.  For each fixed occupied charge, the BPS Hilbert space in that sector is a representation of $su(2)_{R} \times su(2)_{J}$, the algebra of $R$-symmetries and rotations.  This representation takes the form
\begin{equation}
\left[\left(\frac{1}{2},0\right)\oplus \left(0,\frac{1}{2}\right)\right] \otimes \mathcal{H}_{\gamma}.
\end{equation}
The representation in brackets above gives the center of mass degrees of freedom of a BPS particle.  The space $\mathcal{H}_{\gamma}$ comprises the internal degrees of freedom of the particle.  We count states by forming an index, the unframed protected spin character
\begin{equation}
\Omega(\gamma, y)\equiv \mathrm{Tr}_{\mathcal{H}_{\gamma}}y^{2J_{3}}(-y)^{2I_{3}}.
\end{equation}
$\Omega(\gamma, y)$ is protected by supersymmetry under small deformations of the vacuum $u$.  It is this quantity that we aim to reproduce from quiver quantum mechanics.

The starting point for a constructing a quiver description of the BPS spectrum is to take the non-relativistic limit.  In this regime, particle number of the massive BPS states is conserved and hence the description simplifies.  For a large class of $\mathcal{N}=2$ theories, the resulting non-relativistic system may be usefully organized into a simple effective field theory.  This is achieved by identifying a finite number of massive BPS particle states which form the elementary fields in the non-relativistic limit.  All such particles are assumed to contain no internal degrees of freedom and form half-hypermultiplet representations of the superalgebra.  The remaining BPS particles are then viewed as non-relativistic bound states of the elementary quanta.

In general, there is a complicated relationship between the elementary states of the non-relativistic theory and the original ultraviolet fields defining the $\mathcal{N}=2$ system.  For example, if the UV quantum field theory is a non-abelian gauge theory, a non-relativistic elementary field might be a monopole, a highly non-linear function of the original degrees of freedom.  Thus, the identification of the elementary fields in the non-relativistic limit requires physical insight into the dynamics.  Nevertheless, for a wide class of models, the effective non-relativistic theory can be written explicitly.

In all such cases, the resulting system is a supersymmetric quantum mechanics with four supercharges.  Supersymmetry constrains the form the interactions of the non-relativistic elementary fields and yields a simple relationship between the spectrum of supersymmetric bound states computed in the non-relativistic theory and the BPS states of the original ultraviolet theory.

To describe the connection, we split the set of one-particle states of the Hilbert space of the $\mathcal{N}=2$ field theory into a set of \emph{particles} and \emph{antiparticles}.  This splitting is not unique.  To specify it, what is required is a choice of half-plane inside the central charge plane.  Such a half-space may be labeled by a choice of angle $\vartheta$ defining its boundary as
\begin{equation}
\mathfrak{h}_{\vartheta}=\left\{Z\in \mathbb{C}|\vartheta+\pi>\arg(Z)>\vartheta \right \}, \hspace{.5in} \vartheta \in[0,2\pi).
\end{equation}
Given a fixed angle $\vartheta$, particles are defined to to be those states whose central charge lies in $\mathfrak{h}_{\vartheta}$, while antiparticles are defined to be those states whose central charge lies in the complement of $\mathfrak{h}_{\vartheta}$.

The spectrum of bound states computed by a non-relativistic quantum mechanics is not the complete set of one-particle BPS representations in the Hilbert space of the theory.  Instead, the spectrum determined by the quantum mechanics is merely the set of particles. The fact that the quantum mechanics computes only the set of particles entails no loss of information due to CPT invariance of the UV field theory.  The antiparticles consist precisely of states with equal degeneracies and spin but opposite electromagnetic charges from the particles.

The structure of the non-relativistic quantum mechanics is conveniently encoded in terms of a directed graph $Q$ known as a quiver.  The quiver has the following features.
\begin{itemize}
\item Nodes:  For each elementary non-relativistic particle there is a node.  Each node is equipped with a charge $\gamma_{i}\in \Gamma$.  The minimum number of such nodes is equal to the rank of the lattice $\Gamma$.
\item Arrows: Between any two nodes $i$ and $j$ there are a number of arrows.  The net number of arrows is given by the symplectic product $\langle \gamma_{i}, \gamma_{j} \rangle$. Such arrows become fields in the non-relativistic quantum mechanics and model the forces between the elementary BPS states.
\item Superpotential:  If the arrows of the quiver may be concatenated to form an oriented cycle, the data must be supplemented by a potential $\mathcal{W}$ which is a formal function of such cycles.
\item Central Charges: Each node is equipped with its central charge $Z_{i}=Z(\gamma_{i})$.  These are constrained to satisfy $Z_{i} \in \mathfrak{h}_{\vartheta}$.
\end{itemize}

Given this data, the problem of determining all BPS particles of the theory maps to the problem of determining the supersymmetric ground states in the quantum mechanics encoded by the quiver.  There are two distinct regimes where this problem may be studied, which we describe in detail below.

\subsubsection{Higgs Branch}
\label{Higgsreview}

The ``Higgs branch'' of the quiver quantum mechanics refers to the effective description of the bifundamental (arrow) fields obtained when the vector multiplet fields, describing the positions of the particles, are integrated out.

We can phrase the calculation of ground states of this system in the language of quiver representation theory.   Fix a charge $\gamma \in \Gamma$.  We wish to know whether this charge is occupied by BPS particles, and if so, to determine their degeneracy and spin content.  To begin, we express $\gamma$ in terms of the charges of the elementary BPS states:
\begin{equation}
\gamma=\sum_{i}n_{i}\gamma_{i}, \hspace{.5in}n_{i}\in \mathbb{Z}_{\geq 0}.
\end{equation}
The vector $\vec{n}=(n_{1}, n_{2}, \cdots)$ of non-negative integers is referred to as the dimension vector of $\gamma$.

Next, we construct the moduli space of stable quiver representations $\mathcal{M}_{\vec{n}}$.  In detail this moduli space is constructed by the following procedure  \cite{Douglas:1996sw, Douglas:2000ah, Douglas:2000qw}.
\begin{itemize}
\item To the $i$th node of the quiver we assign a complex vector space $V_{i}$ with complex dimension $n_{i}$ .   For each arrow from node $i$ to node $j$ we assign a complex linear map $\Phi: V_{i}\rightarrow V_{j}$.  The set of vector spaces and linear maps is known as a representation of the quiver $Q$ with dimension vector $\vec{n}$.
\item Constrain the maps $\Phi$ associated to arrows by requiring that the superpotential is extremized
\begin{equation}
\frac{\partial{\mathcal{W}}}{\partial{\Phi}}=0, \hspace{.4in} \forall \Phi.
\end{equation}
Representations satisfying this condition are said to be flat.
\item Impose stability.  Define the central charge of a representation $R$ by linearity
\begin{equation}
Z_{R}=\sum_{i}n_{i}Z_{i}.
\end{equation}
Given any representation $R$, a subrepresentation $S\subset R$ is a representation of $Q$ with vector spaces $v_{i}$ and morphisms $\phi$ such that $v_{i}$ is a vector subspace of $V_{i}$ and $\phi$ the restriction of $v_{i}$ of the associated linear map $\Phi$.  Every representation $R$ has two trivial subrepresentations given by the zero representation which assigns a zero-dimensional vector space to each node, and $R$ itself.  $R$ is called stable if for all nontrivial subrepresentations
\begin{equation}
\arg \left(Z_{R}\right)>\arg \left(Z_{S}\right).
\end{equation}
\item The set of stable flat representations is acted upon naturally by the group $\prod_{i}Gl(V_{i},\mathbb{C})$.  If $\Phi$ is a map from $V_{i}$ to $V_{j}$ the action is
\begin{equation}
\Phi \rightarrow G_{j}^{-1}\Phi G_{i}, \hspace{.5in} G_{i}\in Gl(V_{i},\mathbb{C}).
\end{equation}
We define the moduli space $\mathcal{M}_{\vec{n}}$ as the quotient of the set of stable flat representations by the the above group action.
\end{itemize}

The moduli space $\mathcal{M}_{\vec{n}}$ is the space of classical supersymmetric Higgs branch ground states of the non-relativistic quantum mechanics involving $n_{i}$ particles of type $i$.  To determine the BPS spectrum of the theory in the charge sector $\gamma$ we must now quantize the space $\mathcal{M}_{\vec{n}}$.  Mathematically this means that we must pass to cohomology.  The moduli space $\mathcal{M}_{\vec{n}}$ is a K\"{a}hler manifold of complex dimension $d$, and hence its cohomology admits a Hodge decomposition with Hodge numbers $h^{p,q}(\mathcal{M}_{\vec{n}})$. The Cartan generators $I_{3}$ and $J_{3}$ of $su(2)_{R}\times su(2)_{J}$ are realized on the cohomology as
\begin{equation}
2J_{3}=p+q-d, \hspace{.5in}2I_{3}=p-q.
\end{equation}
From this we conclude that the desired index of BPS states $\Omega(y,\gamma)$ can be extracted from a specialization of the Hodge polynomial
\begin{equation}
\Omega_{\mathrm{Higgs}}(\gamma, y)=\sum_{p,q=0}^{d}h^{p,q}\left(\mathcal{M}_{\vec{n}}\right)\left(-1\right)^{p-q}y^{2p-d}. \label{indexform}
\end{equation}
Equation \eqref{indexform} states a precise relationship between the BPS indices of the four dimensional quantum field theory and the mathematical invariants of quiver moduli spaces.\footnote{For most cases of interest in this paper the moduli spaces are compact and smooth and the discussion above is accurate.  In complete generality however, the moduli spaces are non-compact and possess singularities.  In that case the extraction of the invariants from the moduli space is more subtle and requires the machinery of \cite{KS1, KS2, KS3}.}  It provides the Higgs branch prediction for the protected spin characters.

\subsubsection{Coulomb Branch}
\label{Coulombreview}

The ``Coulomb branch'' of the quiver quantum mechanics refers to the effective description of the vector multiplet fields obtained when the bifundamental (arrow) fields are integrated out.  In the regime of validity of this approximation, wavefunctions extracted by quantization describe multi-centered particle configurations.  Such configurations have been studied in the context of supergravity in \cite{Denef:2000nb, Denef:2002ru, Denef:2007vg}.

We again fix a charge $\gamma \in \Gamma$ with dimension vector $\vec{n}=(n_{1}, n_{2}, \cdots)$.  The $i$th node of the quiver supports a $U(n_{i})$ vector multiplet.  There are three real scalars in this multiplet each transforming in the adjoint of $U(n_{i}).$ We view them as a three-component vector $\vec{r}_{i}.$   Additionally each $U(n_{i})$ gauge group has an associated real Fayet-Iliopoulos parameter $\chi_{i}$ extracted from the central charge as
\begin{equation}
\chi_{i}=\mathrm{Im}\left(\frac{n_{i}Z_{i}}{\sum_{i}n_{i}Z_{i}}\right).
\end{equation}
By construction, these parameters satisfy the constraint $\sum_{i}\chi_{i}=0.$

On the Coulomb branch of classical ground states each non-abelian gauge group $U(n_{i})$ is broken to $U(1)^{n_{i}}$.  The scalars are commuting and may be simultaneously diagonalized.  The eigenvalues of the scalars are denoted by $\vec{r}_{i,j}$ where the index $j$ ranges from $1$ to $n_{i}$.  The eigenvalue $\vec{r}_{i,j}$ physically describes the classical position of the $j$-th particle of type $i$.  There is a residual gauged permutation symmetry $S_{n_{i}}$ at each node which descends from the Weyl group of each gauge group.  This discrete quotient accounts for the fact that the $n_{i}$ particles of type $i$ are identical.

The values of the scalars $\vec{r}_{i,j}$ parameterize the Coulomb branch.  As derived in \cite{Denef:2000nb}, the $\vec{r}_{i,j}$ are subject to the constraints that for each node $i$,
\begin{equation}
\sum_{k\neq i}\sum_{j=1}^{n_{i}}\sum_{\ell=1}^{n_{k}} \frac{\langle \gamma_{i},\gamma_{k}\rangle}{|\vec{r}_{i,j}-\vec{r}_{k,\ell}|}=\chi_{i}. \label{denefeq}
\end{equation}
As a consequence of the constraint on the FI parameters, the sum of these equations vanishes.  We define the Coulomb branch manifold $\mathcal{C}_{\vec{n}}$ as the solutions to \eqref{denefeq} modulo the action of the permutation gauge groups.
\begin{equation}
\mathcal{C}_{\vec{n}}=\left(\left\{\vec{r}_{ij}\in \mathbb{R}^{3}|\sum_{k\neq i}\sum_{j=1}^{n_{i}}\sum_{\ell=1}^{n_{k}} \frac{\langle \gamma_{i},\gamma_{k}\rangle}{|\vec{r}_{i,j}-\vec{r}_{k,\ell}|}=\chi_{i}\right\}\right)/\prod_{i}S_{n_{i}}
\end{equation}
An $SO(3)$ group of rotations acts on $\mathcal{C}_{\vec{n}}$ via simultaneous rotations on the position variables $\vec{r}_{ij}.$

The manifold $\mathcal{C}_{\vec{n}}$ is the classical moduli space that we quantize to extract BPS states.  As the constraint equations \eqref{denefeq} depend only on the relative positions, $\mathcal{C}_{\vec{n}}$ naturally contains a copy of $\mathbb{R}^{3}.$ The quantization of this non-compact space leads to plane wave states of definite momentum describing the motion of the center of mass of the multi-particle system.

To determine the states from the remaining manifold $\mathcal{C}_{\vec{n}}$ modulo the center of mass $\mathbb{R}^{3}$, we note that the fermionic superpartners to the position fields $\vec{r}_{ij}$ define a spinor bundle $E$ over the moduli space.  The BPS states extracted from the Coulomb branch are defined to be in one-to-one correspondence with the harmonic spinors in $E$.  Such harmonic spinors arise in representations of the group $SO(3)$ which is interpreted as the angular momentum action on the associated state. See \cite{deBoer:2008zn, Manschot:2011xc} for additional details.

\subsubsection{The Higgs-Coulomb Relationship}
\label{MPSsec}

We have now described two distinct methods for extracting BPS states from a quiver quantum mechanics system: the Higgs branch and associated quiver representation theory, and the Coulomb branch and associated harmonic spinors.  For quivers without oriented loops, it is known that the two prescriptions yield equivalent predictions for BPS states \cite{Denef:2002ru}.  However, more generally in stating that an $\N=2$ theory is of quiver type, we mean, by definition, that a Higgs branch calculation in quiver quantum mechanics reproduces the correct BPS states predicted by the UV field theory.\footnote{For a rigorous derivation of this fact in the case of theories of type $S[A_{1}],$ and $SU(n)$ SYM coupled to fundamental hypermultiplet mater see \cite{Bridgeland1} and \cite{Chuang:2013wt} respectively.}  Thus in theories of quiver type we have
\begin{equation}
\Omega_{\mathrm{Higgs}}(\gamma,y)=\Omega(\gamma,y),
\end{equation}
where the right-hand-side denotes the protected spin characters of the UV field theory in question.  The connection with the Coulomb branch is subtle and is described below.

The difficulty arises from the fact that when a quiver diagram has oriented loops, the Coulomb branch equations \eqref{denefeq} typically admit \emph{scaling solutions} where groups of three or more position variables become small simultaneously.  These are non-compact regions of the moduli space $\mathcal{C}_{\vec{n}}$ and when such regions exist the Coulomb branch calculation becomes incomplete.\footnote{Indeed in a scaling region of moduli space, the assumption that the bifundamental fields are massive and may be integrated out is invalid and hence additional light degrees of freedom become relevant. }  Said differently, the Coulomb branch description of BPS states does not give a prediction for the number of quantum wavefunctions which are supported at the end of the scaling region where, semiclassically, the centers in question coincide.

One may encapsulate the ambiguity described above into a \emph{single-centered degeneracy} $\Omega_{S}(\gamma,y)$ which is the protected spin character of BPS states with charge $\gamma$ whose Coulomb branch description consists of an object with a single position variable.  More explicitly, we may consider a region of parameter space where all central charges $Z_{i}$ are nearly aligned.  In that case the Coulomb branch description of BPS states becomes parametrically accurate and all BPS particles have a molecular description, similar to a hydrogen atom, with a characteristic Bohr radius.  As the central charges are further aligned all BPS particles except for the single-centered states contributing to $\Omega_{S}(\gamma,y)$ have the property that their characteristic size diverges.

The Coulomb branch does not give any prediction for the values of the $\Omega_{S}(\gamma,y).$  However, if these single-centered invariants are specified, the Coulomb branch may be quantized uniquely to obtain an unambiguous prediction for the set of multi-centered states.  The utility of this result is that, while the complete BPS degeneracies $\Omega(\gamma,y)$ depend on the FI parameters $\chi_{i}$ and undergo wall-crossing, the single-centered degeneracies $\Omega_{S}(\gamma,y)$ do not jump across walls of marginal stability.

Recently, a recursive algorithm for extracting the multi-centered total degeneracies $\Omega(\gamma,y)$ from the single-centered degeneracies $\Omega_{S}(\gamma,y)$ was derived \cite{Manschot:2011xc, Manschot:2012rx, Manschot:2013sya}.  This result is similar to the application of the Harder-Narasimhan recursion formula \cite{MR0364254} derived in \cite{Reineke1} for quiver representation theory.

In the special case of a primitive dimension vector, i.e. when the integers $(n_{1}, n_{2}, \cdots)$ have no common divisor, we may express the result of \cite{Manschot:2013sya} as follows.  Let $\gamma\in \Gamma$ indicate the charge of interest, and let $n$ be a positive integers.  By the notation $\{\alpha_{1}, \alpha_{2}, \cdots, \alpha_{n}\}$ we mean a partition of the charge $\gamma$ into $n$ pieces $\alpha_{i},$ each of which has a non-negative integral expansion in terms of the node charges of the quiver in question.  The Coulomb branch degeneracy $\Omega_{\mathrm{Coulomb}}(\gamma,y)$ is expressed in terms of contributions from all such partitions as
\begin{equation}
\Omega_{\mathrm{Coulomb}}(\gamma,y)=\Omega_{S}(\gamma,y)+\sum_{n>1}\sum_{\{\alpha_{1}, \cdots, \alpha_{n}\}} U(\{\alpha_{i}\}, \chi, y)\prod_{k=1}^{n} \sum_{m_{i}|\alpha_{i}}\frac{y-y^{-1}}{m_{i}(y^{m_{i}}-y^{-m_{i}})}V(\alpha_{i}/m_{i},y^{m_{i}}).  \label{MPS}
\end{equation}
The functions $U$ and $V$ appearing above, depend only on the arguments indicated explicitly.  $U$ is extracted from the geometry of Coulomb moduli spaces $\mathcal{C}_{\vec{n}},$ while $V$ is determined recursively from the geometry of $\mathcal{C}_{\vec{n}},$ as well as the values of the postulated single-centered degeneracies.  For a detailed discussion see \cite{Manschot:2013sya}.

One significant feature of the above is that for any choice of single-centered degeneracies $\Omega_{S}(\gamma,y)$ the definition \eqref{MPS} satisfies the Kontsevich-Soibelman wall-crossing formula and has the property that the single-centered invariants $\Omega_{S}(\gamma,y)$ are stable across all walls of marginal stability \cite{Manschot:2010qz, Sen:2011aa, Reineke5}.

Let us now return to the Higgs branch.  The quiver representation theory problem does not suffer from any of the ambiguities plaguing the Coulomb branch.  In particular, in $\mathcal{N}=2$ theories of quiver type, the Higgs branch yields complete results for the protected spin characters $\Omega(\gamma,y).$  If we enforce the equality
\begin{equation}
\Omega(\gamma,y)=\Omega_{\mathrm{Higgs}}(\gamma,y)=\Omega_{\mathrm{Coulomb}}(\gamma,y),
\end{equation}
then via comparison with \eqref{MPS}, the Higgs branch predicts values for the single-centered invariants $\Omega_{S}(\gamma,y)$.  As these quantities are stable under wall-crossing they define a new class of quiver invariants.  The states contributing to $\Omega_{S}(\gamma,y)$  are sometimes referred to as \emph{pure Higgs} states.  Several examples of these invariants are described in \cite{Denef:2007vg, Bena:2012hf, Lee:2012sc, Lee:2012naa}.  They satisfy two basic properties:
\begin{itemize}
\item If $\gamma_{i}$ is the charge of a node of the quiver $Q$, then $\Omega_{S}(\gamma_{i},y)=1.$
\item If the quiver $Q$ has no oriented loops and $\gamma\neq \gamma_{i}$ for all nodes then $\Omega_{S}(\gamma,y)=0.$
\end{itemize}

In general there is no known method, other than direct comparison to \eqref{MPS}, for determining these single-centered invariants.  We view a formulation of these quantities in terms of intrinsic Higgs branch data as an interesting open problem for future work.  For the remainder of this paper we assume that the single-centered degeneracies of a theory of quiver type have been computed, and we explain how to use them to determine the framed BPS spectra associated to an arbitrary line defect.

\subsection{Framed BPS States from Multiparticle Quantum Mechanics}
\label{framing}

In this section we present our proposal for the framed BPS states in theories of quiver type.  The basic observation which makes our analysis possible is that the quiver quantum mechanics description of the BPS spectrum is an intrinsically abelian effective field theory.  Thus, to study the framed BPS states using quivers we make use of the abelian line operators described in \S \ref{bpsabelian}.

Fix a line defect $L(\alpha, \zeta)$ defined in the ultraviolet.   We assume that $\zeta \in \mathbb{C}^{\times}$ lies on the unit circle. As stated in \S \ref{susyalg}, in that case the superalgebra preserved by the line defect is the same as that preserved by a massive BPS particle of with central charge of phase $\zeta$.  At a point $u$ on the Coulomb branch an infrared observer describes this defect, to leading order, as an infinitely massive dyon with electromagnetic core charge $\gamma_{c} \in \Gamma$ determined by the UV-IR renormalization group flow map described in \S \ref{troplicallabels},
\begin{equation}
\gamma_{c} = {\mathbf{RG}}(\alpha, \zeta, u). \label{coredef}
\end{equation}
We thus wish to couple the BPS quiver which governs the ordinary BPS states to an infinitely massive particle of charge $\gamma_{c}$.  We attempt to include the remaining framed BPS states appearing in the IR decomposition of the UV defect as excitations bound to the core dyon.

Since we work directly in the infrared description of the $\mathcal{N}=2$ theory, coupling our theory to a an infinitely massive dyon is straightforward:  we simply couple to a dyon of finite mass $M$ and take the limit $M\rightarrow \infty$.  This may be viewed as a generalization to the case of an interacting theory, of the construction in \S \ref{bpsabelian} of framed BPS states in abelian gauge theories.

To carry out this procedure explicitly, we begin by extending our charge lattice $\Gamma$ to include an extra direction.  The extended charge lattice is therefore
\begin{equation}
\Gamma \oplus \mathbb{Z}[\gamma_{F}].
\end{equation}
We further extend the symplectic product trivially by declaring that $\gamma_{F}$ has vanishing pairing with all charges in the extended lattice.  Thus, $\gamma_{F}$ is a new flavor charge in our theory.  All ordinary BPS states, present before coupling to the probe, carry zero units of the $\gamma_{F}$ flavor charge.  However the defect itself is modeled as a new particle which carries one unit of flavor charge and hence has total charge $\gamma_{c} +\gamma_{F}$.  We extend the central charge function to the extended lattice by declaring
\begin{equation}
Z(\gamma_{c}+\gamma_{F})=M\zeta, \hspace{.5in}M\in \mathbb{R}_{>0}. \label{limitZ}
\end{equation}
The parameter $M$ thus controls the mass of the defect, and we will work in the limit $M\rightarrow \infty$.

We now state a precise proposal for the framed BPS states bound to the defect utilizing the concept of single-centered invariants described in \S \ref{MPSsec}.  Since our UV field theory is of quiver type, each charge $\gamma$ in the original charge lattice $\Gamma$ may be assigned a single-centered degeneracy $\Omega_{S}(\gamma,y)$.  We extend these single-centered invariants to the lattice $\Gamma \oplus \mathbb{Z}[\gamma_{F}]$ by
\begin{equation}
\forall  \ \gamma \in \Gamma \hspace{.5in}\Omega_{S}(\gamma+n\gamma_{F},y)=\begin{cases}\Omega_{S}(\gamma,y) & \mathrm{if}  \  n=0, \\
1 & \mathrm{if}  \ n=1 \ \mathrm {and} \ \gamma=\gamma_{c}, \\
0 & \mathrm{otherwise}.
\end{cases}
\label{extendedsinglecenter}
\end{equation}
Framed BPS states supported by the defect carry exactly one unit of flavor charge $\gamma_{F}.$ Hence, they have a total electric magnetic charge of the form
\begin{equation}
\gamma_{\mathrm{total}}=\gamma_{F}+\gamma_{c}+\gamma_{h}=\left(\gamma_{F}+\gamma_{c}\right)+\sum_{i}n_{i}\gamma_{i}, \hspace{.5in} n_{i}\in \mathbb{Z}_{\geq0}. \label{halodef}
\end{equation}
The above formula defines the \emph{halo charge,} $\gamma_{h},$ which admits a non-negative expansion in terms of the node charges $\gamma_{i}$ of the original quiver.   We conjecture that the framed BPS indices of the defect with core charge $\gamma_{c}$ are the Coulomb branch multi-centered bound states computed using \eqref{extendedsinglecenter} as the input single-centered degeneracies,
\begin{equation}
\fro(\alpha, \gamma, \zeta, y)=\Omega_{\mathrm{Coulomb}}(\gamma+\gamma_{F}, y), \label{framedquivconj}
\end{equation}
where the right-hand-side is computed using the formula \eqref{MPS} and the parameter $\zeta$ enters as the phase of the central charge in \eqref{limitZ}.

We may immediately note an important consistency check on this conjecture.  The framed BPS states of a given line defect can undergo wall-crossing as moduli are varied, and the discontinuities in the framed protected spin characters are governed by the wall-crossing formula \cite{KS1, KS2, KS3, JS, J}.   Thus we may ask:  as moduli are varied, do the proposed degeneracies defined in \eqref{framedquivconj} jump according to the wall-crossing formula?  The answer is that they do.  Indeed as remarked in \S \ref{MPSsec}, a feature of \eqref{MPS} is that for any input single-centered degeneracies, including those defined in \eqref{extendedsinglecenter}, the resulting multi-particle degeneracies vary according to the wall-crossing formula \cite{Manschot:2010qz, Sen:2011aa, Reineke5}.

Physically, equation \eqref{framedquivconj} implies a simple intuitive picture for the framed BPS states associated to a line defect $L(\alpha,\zeta).$  There is a universally stable state carrying the core charge $\gamma_{c}$.  This state is single-centered, arising from the postulated non-zero value $\Omega_{S}(\gamma_{c}+\gamma_{F},y)$ in \eqref{extendedsinglecenter}.  It describes the bare defect, and as it is single-centered, appears as a microscopic object to an infrared observer.  The remaining framed BPS states are finite mass excitations bound to this core.  They carry a halo charge $\gamma_{h}$ and may be viewed as a collection of ordinary BPS states orbiting the core as illustrated in Figure \ref{fig:ptolemyintro}.

In particular, the input single-centered degeneracies \eqref{extendedsinglecenter} imply that there are no framed BPS states where any constituent of the halo collapses forms a single-centered object with the core: such a state would yield a non-zero contribution to some single-centered degeneracy $\Omega_{S}(\gamma_{c}+\gamma_{h}+\gamma_{F},y)$ which by hypothesis vanishes unless the halo charge is trivial.  Thus, all ordinary BPS states orbit the defect at a non-zero radius which may be made parametrically large by approaching a locus of moduli space where the central charges $Z_{i}$ all align with the defect ray at phase $\zeta.$

In the remainder of this paper we will provide further evidence for our conjecture \eqref{framedquivconj} and study its consequences.

\subsubsection{Framed Quivers}
\label{framedquivers}

Our proposal \eqref{framedquivconj} for the framed BPS degeneracies of a line defect is stated intrinsically in terms of multi-particle Coulomb branch quantum mechanics.  There is a closely related Higgs branch construction utilizing quiver representation theory of an augmented quiver.  In general, the quiver representation theory problem described below does not compute the same framed degeneracies as our conjecture \eqref{framedquivconj}.   However, there is a subclass of line defects above \eqref{firstcone} and studied in detail in \S\ref{dualseeds} where the quiver representation theory does compute the correct framed degeneracies.

Let $L(\alpha,\zeta)$ denote an ultraviolet line defect with core charge $\gamma_{c}$ defined as in \eqref{coredef} via the renormalization group map.  Starting from the original quiver $Q$ we construct a new quiver as follows.
\begin{itemize}
\item Adjoin to $Q$ a new node.  The new node models the bare defect.  Arrows to and from the new node are determined by the symplectic products $\langle \gamma_{c}, \gamma_{i}\rangle$.  We denote the quiver resulting from this procedure as $Q[\gamma_{c}]$, and refer to it as the \emph{framed quiver} with framing charge $\gamma_{c}$.  Similarly the node representing the defect is referred to as the \emph{framing node}.
\item The framed quiver $Q[\gamma_{c}]$ may have new cycles formed by arrows entering or exiting the framing node.  In this situation we modify the superpotential to include new terms $\delta \mathcal{W}$.  We assume that these terms are chosen generically.  Thus, we set $\delta \mathcal{W}$ to be equal to a formal sum of all cycles involving arrows to and from the framing node.  The coefficients in the summation are to be chosen generically.
\item Extend the central charge function to the framing node as in \eqref{limitZ}.
\end{itemize}

Utilizing the framed quiver $Q[\gamma_{c}]$ we may extract a set of states which have the electric magnetic charges compatible with a core plus halo expansion as in \eqref{halodef}.  Thus, we examine representations of $Q[\gamma_{c}]$ with dimension vector of the form $(1,n_{1}, n_{2}, \cdots),$ where the first entry indicates the framing node.  We define
\begin{equation}
\fro_{\mathrm{Higgs}}(\alpha,\gamma, \zeta,y)=\Omega_{\mathrm{Higgs}}^{Q[\gamma_{c}]}(\gamma,y), \label{higgswrong}
\end{equation}
where the right-hand-side indicates the unframed BPS degeneracies for the framed quiver constructed from core charge $\gamma_{c},$ and where $\zeta$ specifies the phase of the central charge as in \eqref{limitZ}.

For general line defects $L(\alpha,\zeta),$ the degeneracies \eqref{higgswrong} do not agree with the framed BPS states predicted by our conjecture \eqref{framedquivconj}.  Moreover, in \S \ref{wilsonlineex} we analyze in detail the example of Wilson lines in $SU(2)$ super Yang-Mills theory and show that while our proposal \eqref{framedquivconj} produces states and operator product algebra expected on physical grounds, the degeneracies \eqref{higgswrong} computed by the framed quiver are incorrect.

Nevertheless, the framed quiver is still useful for many purposes.  The discrepancy between the degeneracies $\fro_{\mathrm{Higgs}}(\alpha,\gamma, \zeta,y)$ and those predicted by $\eqref{framedquivconj}$ occurs because the Higgs branch calculation of quiver representations implies a set of single-centered invariants $\Omega_{S}^{Q[\gamma_{c}]}(\gamma,y)$ which violate our assumptions \eqref{extendedsinglecenter}.  However, it is clear that the quiver representation theory of $Q[\gamma_{c}]$ may imply more non-trivial single-centered invariants but never fewer.  These additional single-centered states may themselves combine to make additional bound states which contribute to the total index, but in general it follows that the Higgs branch calculation may produce more BPS states but never fewer.  In particular, assuming our formula \eqref{framedquivconj} we have the useful implication
\begin{equation}
\fro_{\mathrm{Higgs}}(\alpha,\gamma, \zeta,y)=0\Longrightarrow \fro_{}(\alpha,\gamma, \zeta,y)=0.
\end{equation}

Another important fact about framed quivers is that there is a large class of line defects for which no spurious single-centered degeneracies occur and hence the framed quiver does compute the correct spectrum,
as follows.

\textbf{Fact:}
 Suppose that a line defect $L(\alpha,\zeta),$ has a core charge $\gamma_{c} = {\mathbf{RG}}(\alpha, \zeta, u)$ satisfying $\langle \gamma_{c},\gamma_{i}\rangle \geq 0$ for all node charges $\gamma_{i}$ of the quiver $Q$.  Then we have:
\begin{equation}
\fro_{\mathrm{Higgs}}(\alpha,\gamma, \zeta,y)=\fro_{}(\alpha,\gamma, \zeta,y). \label{firstcone}
\end{equation}

To prove the above statement we examine the Coulomb branch equations \eqref{denefeq} defining the allowed multi-centered configurations.  We fix the overall center of mass by placing the core charge at the origin.  Then, if $i$ labels the nodes of the unframed quiver, and $\vec{r}_{i,j}$ the positions of the associated centers, one constraint reads
\begin{equation}
\sum_{i}\sum_{j=1}^{n_{i}} \frac{\langle \gamma_{c},\gamma_{i}\rangle}{|\vec{r}_{i,j}|}=\chi. \label{denefcone}
\end{equation}
Note that on left-hand-side of the above, each term in the summation is non-negative by our assumptions about the core charge, while the right-hand side is a parameter that depends on the central charge evaluated at the given point in moduli space.

As the value of $\chi$ is varied, multi-centered BPS states arising from quantization of the Coulomb branch may change discontinuously.  However, the single-centered degeneracies captured by $\Omega_{S}(\gamma,y)$ are stable under such variations.  Thus to study them we may freely choose a convenient value of $\chi.$  Let us therefore choose $\chi$ to be negative.  Then, the only solution to \eqref{denefcone} is the trivial one where all $n_{i}$ vanish, and there are no  $\vec{r}_{i,j}$ to speak of.  In particular this shows that there are no spurious single-centered degeneracies hence proving \eqref{firstcone}.

Note that in fact our argument has produced a stronger result:  for any defect with core charge satisfying the hypotheses stated above \eqref{firstcone}, there is a chamber where the framed BPS spectrum consists of a single BPS state describing the isolated core charge.  We study this phenomenon and other aspects of  the defects satisfying \eqref{firstcone} in more detail in \S \ref{dualseeds}.

\section{Properties of Line Defects in Quiver Type Theories}
\label{properties}

In this section we describe the structure present in the spectrum of framed BPS states when an $\mathcal{N}=2$ theory is of quiver type.

In \S \ref{seeds} we review quiver mutation, a natural class of dualities in quiver quantum mechanics.  Consequences of quiver mutation described there are jumps in the single-centered degeneracies, and a derivation of the tropical formula \eqref{tropicallabeljump} describing the discontinuities in the map $\mathbf{RG}(\alpha, u,\zeta)$ as the modulus $u$ is varied.

In \S \ref{dualseeds} we investigate general properties of line defects whose framed BPS states are exactly captured by framed quiver representations.  As illustrated there, such line defects have universal OPEs.

\subsection{Cones, Seeds, and Quiver Mutation}
\label{seeds}

An important aspect of all theories of quiver type is that their BPS states, both framed and unframed, have a conical structure.  We extend the coefficients in the charge lattice to real numbers and consider the cone
\begin{equation}
\mathcal{C}=\left\{\sum_{i} r_{i}\gamma_{i}|r_{i}\in \mathbb{R}_{\geq0}\right\} \subset \Gamma\otimes_{\mathbb{Z}}\mathbb{R}. \label{conedef}
\end{equation}
By construction the integral points of $\mathcal{C}$ contain the charges of all occupied BPS states computed by the quiver.  We refer to $\mathcal{C}$ as the \emph{cone of particles}.  The existence of this cone is one of the fundamental properties of theories whose BPS states admit a quiver description.\footnote{Indeed certain theories, for example $\mathcal{N}=4$ SYM, are not of quiver type because their BPS particles do not lie in a cone.} In practice it is often useful to visualize $\mathcal{C}$ by projecting it onto the $Z$-plane by making use of the linear central charge map $Z:\Gamma\rightarrow\mathbb{C}$.   In this way, we obtain a cone $Z(\mathcal{C})$ in the half-space $\mathfrak{h}_{\theta}$.

It is fruitful to explore the dependence of this conical structure on the parameters of the ultraviolet $\mathcal{N}=2$ theory which have hitherto been fixed.  One such parameter is the modulus $u$ labeling the vacuum state.  Another is the angle $\vartheta$ labeling the choice of half-space $\mathfrak{h}_{\vartheta}$ defining the split into particles and antiparticles.  The cone described above depends on both parameters in a piecewise constant fashion.

An illuminating example of this dependence can be seen by fixing the modulus $u$ and rotating the angle $\vartheta.$  By construction, the projection of the cone of particles, $Z(\mathcal{C}),$ lies in the interior of $\mathfrak{h}_{\vartheta}$ and hence for small changes in the angle $\vartheta,$ $Z(\mathcal{C})$ remains in the interior of $\mathfrak{h}_{\vartheta}$.  Thus, the cone is invariant under small deformations in $\vartheta$.  However, eventually as $\vartheta$ is varied, the half-space reaches a critical angle $\vartheta_{c}$ where a ray $Z_{\gamma}$ on the boundary of $Z(\mathcal{C})$ coincides with the boundary of $\mathfrak{h}_{\vartheta_{c}}$.  At this critical angle, the cone does not exist in the open subspace $\mathfrak{h}_{\vartheta_{c}}$.  As the boundary of $\mathcal{C}$ is spanned by nodes of the quiver, the ray $Z_{\gamma}$ is occupied by some BPS particle.  As we further increase $\vartheta,$ this particle falls out of the half-space $\mathfrak{h}_{\vartheta}$ and changes its
identity to an antiparticle.  This process is illustrated in Figure \ref{fig:cones}.
\begin{figure}[here!]
  \centering
  \subfloat[$\vartheta <\vartheta_{c}$]{\label{fig:cone2}\includegraphics[width=0.45\textwidth]{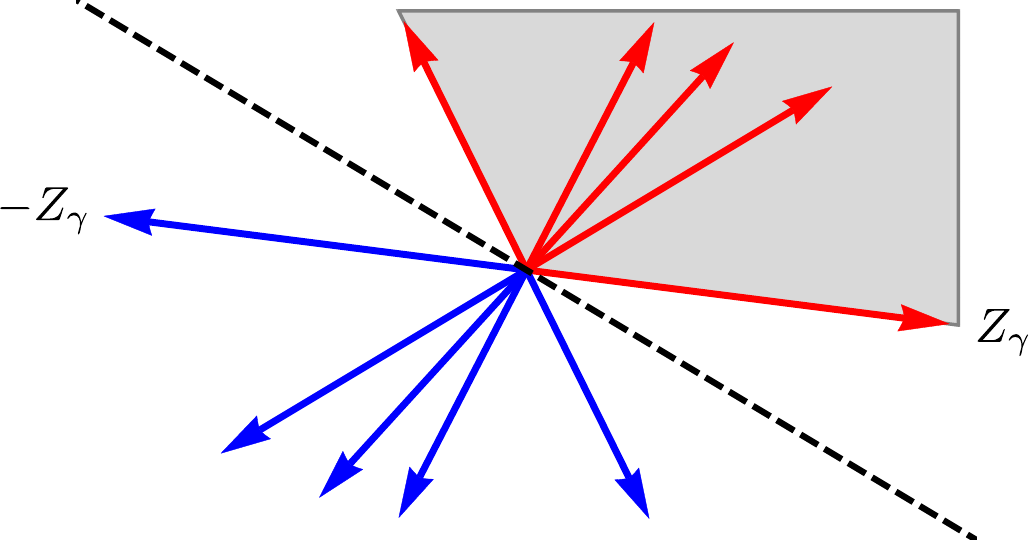}}
  \hspace{.4in}
  \subfloat[$\vartheta>\vartheta_{c}$]{\label{fig:cone1}\includegraphics[width=0.45\textwidth]{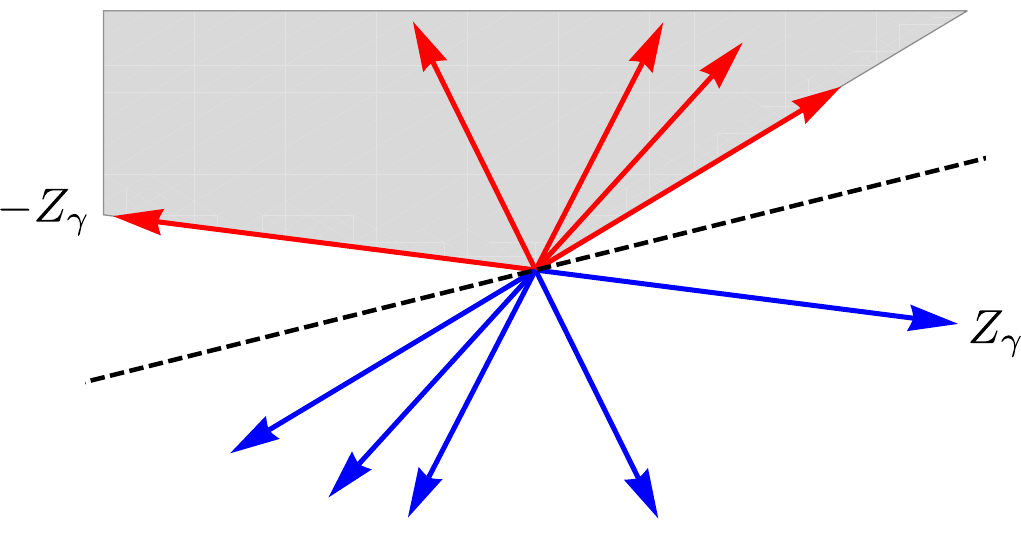}}
  \caption{A jump in the cone of particles induced by a change in the particle half-space.  Red lines indicate the central charges of particles, while blue lines denote central charges of antiparticles.  The gray shaded region indicates the projection of the cone of particles, and the boundary of the particle half-space is indicated with a dashed line.  In passing from (a) to (b) the particle with central charge $Z_{\gamma}$ changes its identity to an antiparticle and the cone jumps. At the critical value $\vartheta_{c},$ the dashed line coincides with the ray $Z_{\gamma}$ and the quiver does not exist.}
  \label{fig:cones}
\end{figure}

As the example above indicates, the cone of particles may jump upon variation of parameters.  This implies that the quiver description of the spectrum is also subject to discontinuities.  Indeed when the cone changes, the initial quiver itself, as a description of the BPS particles, must break down for the simple reason that some node of the quiver has become an antiparticle.  However, often it is possible to give a new quiver description of the new set of particles.

In general, a theory of quiver type possesses many distinct quiver descriptions of its spectrum.  At fixed values of the pair $(u,\vartheta)$ the quiver, if it exists, is unique.  The charges of the nodes of the quiver define a canonical subset of $\Gamma_{u}$ which we refer to as a \emph{seed}.   For any fixed seed $\mathfrak{s}$ the integral lattice $\mathbb{Z}\mathfrak{s}$ spanned by the seed contains the possible charges of unframed BPS states.  In the case of a framed BPS states, this is also the lattice where possible halo charges are valued.  The quotient $\Gamma_{u}/\mathbb{Z}\mathfrak{s}$ is of finite order.\footnote{A typical example where the quotient is non-trivial is $SU(2)$ SYM.  There all unframed BPS states have even electric charge, but the charge lattice $\Gamma_{u}$ contains elements with arbitrary electric charge \cite{Kapustin:2005py}.}

When $u$ is fixed and $\vartheta$ is varied, we may encounter regions of the $\vartheta$ space where, as with $\vartheta_{c}$ above, the cone ceases to exist.  Let $\Delta$ denote the union of all angles where the cone cannot be defined.  The set $[0,2\pi)-\Delta$ is separated into various disconnected components, which are cyclically ordered by increasing boundary angle $\vartheta \in [0,2\pi)-\Delta$.  Each component possesses a distinct seed, and an associated distinct quiver presentation of the spectrum.  A basic problem in a given $\N=2$ theory of quiver type is to describe all such seeds in the lattice $\Gamma_{u}$.

In simple situations, such as that illustrated in Figure \ref{fig:cones}, it is possible to describe the relationship between pairs of seeds in adjacent components of $[0,2\pi)-\Delta$.  There is a an angle $\vartheta_{c}$ at which the cone is ill-defined for the reason that a single BPS hypermultiplet state has central charge along the ray $\vartheta_{c}$, and further, there are unoccupied wedges in the central charge plane both above and below the boundary of $\mathfrak{h}_{\vartheta_{c}}$. Let $\gamma_{i}$ indicate the charge of the hypermultiplet particle exiting the half-space.  It may exit the particle half-space on the right $(R)$ or on the left $(L)$.  After this state has changed identity to an antiparticle, the new seed is related to the original seed by a transformation $\mu_{Ri}$ or $\mu_{Li}$ respectively.  These transformations are known as \emph{mutations}.  They are defined by the following formulae.
\begin{equation}
\mu_{Ri}\left(\gamma_{j}\right)= \begin{cases} -\gamma_{j} & j=i \\ \gamma_{j}-\mathrm{Min}\{\langle \gamma_{i},\gamma_{j}\rangle,0\}\gamma_{i} & j\neq i \end{cases}\hspace{.4in}\mu_{Li}\left(\gamma_{j}\right)= \begin{cases} -\gamma_{j} & j=i \\ \gamma_{j}+\mathrm{Max}\{\langle \gamma_{i},\gamma_{j}\rangle,0\}\gamma_{i} & j\neq i \end{cases}\label{mutrule}
\end{equation}
By construction, left and right mutation are inverses:
\begin{equation}
\mu_{Li} \circ \mu_{Ri}=1, \hspace{.5in}\mu_{Ri}\circ \mu_{Li}=1.
\end{equation}
After mutation, one may construct a new quiver description of the spectrum built on the mutated seed.\footnote{The superpotential of the new quiver may also be determined from that of the old quiver.  See \cite{Alim:2011kw} for details.}

Analogous remarks may be made about the behavior of the cone when $\vartheta$ is fixed and the modulus $u$ is varied.  In this case the central charges of BPS particles may exit $\mathfrak{h}_{\vartheta}$ resulting in a discontinuity.  In simple situations one obtains two quiver descriptions of the spectrum related by mutation.  However, in general, the moduli space may contain whole regions where no quiver exists.

At each point in the moduli space where a quiver exists, we may again form the collection of seeds in $\Gamma_{u}.$  We may compare the seeds at distinct moduli $u$ by parallel transporting them to a common base point $u_{0}$ in the moduli space.\footnote{This parallel transport depends on a choice of path, and hence is subject to monodromy in the local system of lattices $\Gamma_{u}$ fibered over the moduli space.}    As we explain in \S \ref{dualseeds}, the resulting collection of seeds in a fixed lattice $\Gamma_{u_{0}}$ provides significant information about OPE of line defects of the ultraviolet $\mathcal{N}=2$ field theory.

Mathematically it is known that the transformation on seeds defined by mutation determines a derived equivalence on categories of quiver representations \cite{KS1, Keller2}.  Physically this operation describes a duality in non-relativistic quantum mechanics.  From the perspective of the Coulomb branch, some results related to quiver mutation have been described in \cite{Denef:2000nb, Andriyash:2010yf}.  However, the general theory remains to be systematically developed.  We anticipate that quiver mutation continues to describe equivalences in this case as well, and assume this for the remainder of the paper.

\subsubsection{Jumps in Single-Centered Invariants}
\label{singlejumpsec}

Quiver mutations have interesting implications for the concept of single-centered invariants discussed in \S \ref{MPSsec}.  As described there, the single-centered degeneracy $\Omega_{S}(\gamma,y)$ is stable when crossing walls of marginal stability where the total degeneracy $\Omega(\gamma,y)$ may jump.  On the other hand, when the quiver is mutated, the total degeneracies are invariant while the single-centered degeneracies must jump.

One can see the necessity for such discontinuities simply from the two basic properties of single-centered degeneracies listed at the conclusion of \S \ref{MPSsec}.  In particular, the single-centered degeneracy of each node of the quiver must be one.  However, in the course of a quiver mutation the charges of nodes jump according to the formula \eqref{mutrule}.  Thus at the very least a mutation operation $\mu$ must modify the single-centered degeneracies attached to node charges $\gamma_{i}$ as
\begin{equation}
\Omega_{S}(\gamma_{i},y)\rightarrow \Omega_{S}(\mu(\gamma_{i}),y). \label{singlejump}
\end{equation}
The above leaves open the transformation properties of the non-trivial single-centered degeneracies attached to charges which are not nodes.  A general formula for these jumps would be desirable.

There is a simple physical model for the jumps \eqref{singlejump}.  Consider for simplicity a two-centered configuration at a generic point in moduli space, and let $\gamma_{1}$ and $\gamma_{2}$ be the associated charges of the centers.  Then, according to the Coulomb branch formula \eqref{denefeq} the two centers may form a bound state whose semiclassical radius $r$ takes the form
\begin{equation}
r = \frac{\langle\gamma_{1}, \gamma_{2}\rangle}{\chi}.
\end{equation}
In particular for small $\chi,$ the radius is parametrically large and hence an infrared observer can safely identify the associated states as multi-centered.

Now let us imagine adiabatically varying $\chi$ by moving in the moduli space.  We carry out this variation while holding fixed the complex mass of the state as specified by the total central charge, and fix the half-space $\mathfrak{h}_{\vartheta}$ to be the points with positive $\mathrm{Im}(Z).$ As $\chi$ is increased, the radius of the bound state decreases until it reaches a minimum at the locus where the central charges of the constituents has become real.  The infrared observer may then change their perspective on this state, from a composite two centered object, to a single-centered atomic object.  Similar phenomena have been described in the context of attractor flow geometries in \cite{Andriyash:2010yf}.

\subsubsection{RG Discontinuities from Quiver Mutation}
\label{framedmutation}

The conical structure, and seed mutation formulas described in \S\ref{seeds} have important implications for the behavior of the line defect renormalization group flow map $ {\mathbf{RG}}(\alpha, \zeta, u)$ as a function of the modulus $u.$   This map defines the core charge of an ultraviolet line defect which enters centrally in the construction of the framed BPS states.  As explained in \S \ref{troplicallabels}, when the modulus $u$ is varied the core charge is locally constant, but may jump at anti-walls, the loci where $\Re (Z_{\gamma_{h}}/\zeta)$ vanishes for some halo charge $\gamma_{h}$ of a framed BPS state.

In a theory of quiver type, we may give a simple universal formula for these jumps in the special case where the halo charge in question is a node of the quiver.  Indeed, all discontinuities in the quiver description of framed BPS states may be understood as a special case of the general discontinuities described in \S \ref{seeds}.  In particular, the anti-wall, where the discontinuity in the renormalization flow map occurs, must be identified with the boundary of the particle half-space $\mathfrak{h}_{\vartheta}$ where the discontinuity in a quiver description of the spectrum occurs.

From this discussion we conclude that a framed quiver produced from renormalization group flow possesses a preferred choice of particle half-space, namely that where the boundary of $\mathfrak{h}_{\vartheta}$ and the ray defined by $\zeta$ are orthogonal
\begin{equation}
\vartheta = \arg (\zeta)-\frac{\pi}{2}. \label{rghalfdef}
\end{equation}
This identification follows because the loci where the infrared label of the defect jumps are exactly those where quiver mutation occurs.  The structure of the particle half-space $\mathfrak{h}_{\arg (\zeta)-\frac{\pi}{2}}$ is indicated in Figure \ref{fig:rghalf}.
\begin{figure}[here!]
  \centering
  \includegraphics[width=0.4\textwidth]{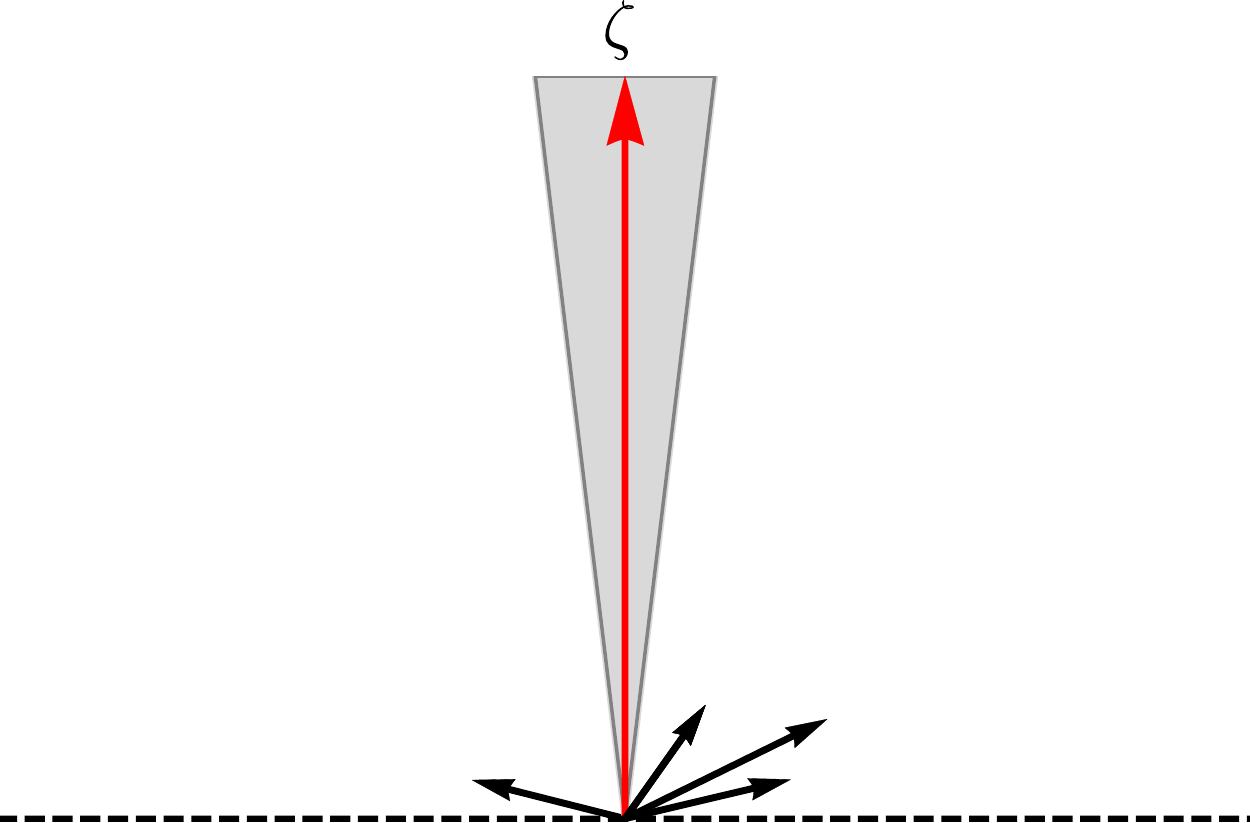}
  \caption{The particle half-space for the framed quiver implied by the RG map.  The red ray indicates the central charge of the defect $\gamma_{c}+\gamma_{F}$, while black rays are the central charges of ordinary BPS states.  The gray shaded region indicates the cone containing the framed BPS states.  In the framing limit, the length of the red ray tends to infinity and the width of the cone is infinitesimal.  The boundary of the half-space is indicated by the dashed line.  As moduli are varied and a central charge of an ordinary BPS state exits the particle half-space, the core charge changes by a mutation.}
  \label{fig:rghalf}
\end{figure}

As moduli are varied, the central charge of an ordinary node of the quiver may reach the boundary of this half-space.  When it does, the seed must be changed by mutations given by \eqref{mutrule}.  However, in the framed quiver, the framing node is a node like any other and therefore the when the the $i$th element of the seed exists the half-space on the right or left, the core charge jumps accordingly
\begin{equation}
\mu_{Ri}(\gamma_{c})=\gamma_{c}-\mathrm{Min}\{\langle \gamma_{i},\gamma_{c}\rangle,0\}\gamma_{i}, \hspace{.5in}\mu_{Li}(\gamma_{c})=\gamma_{c}+\mathrm{Max}\{\langle \gamma_{i},\gamma_{c}\rangle,0\}\gamma_{i}. \label{rgjump}
\end{equation}
The above is exactly the tropical transformation \eqref{tropicallabeljump}.  It is valid in any theory of quiver type.

We can obtain physical insight into the tropical transformation formula by placing it in the general context of the jumping of single-centered degeneracies described in \S \ref{singlejumpsec}.  The mutation rule for the discontinuity in the $\mathbf{RG}$ map is a precise version of the statement that the core charge is modified because, from the infrared point of view, it has fused with a previously bound halo.  Indeed, the difference between $\gamma_{c}$ and its mutated version $\mu(\gamma_{c})$ is always the charge of a halo of an occupied framed BPS state.  However, the mutation formula also reveals an important subtlety in the intuitive picture.  During a mutation all the seed charges change via \eqref{mutrule}, not just the core charge at the framing node.  Hence, when the ground state of the framed Hilbert space is changed we must also change our description of the excitations from those halo charges spanned by the original seed to those spanned by the mutated seed.

\subsection{Cones of Line Defects and Their OPEs}
\label{dualseeds}

In this section, we study the properties of defects defined in \S \ref{framedquivers} where our general proposal \eqref{framedquivconj} for the framed BPS states reduces to a representation theory problem on the framed quiver.  As we demonstrate below, such defects have universal OPEs.

Fix a choice of modulus $u$ and $\zeta \in \mathbb{C}^{\times}.$   According to the discussion of \S\ref{troplicallabels} the renormalization map $\mathbf{RG}(\cdot,\zeta,u),$ which extracts the core charge of a UV line defect,  is a bijection between the set of UV defects and the set of IR line defects is a large class of $\N=2$ theories.  Let us assume that it is a bijection in theories of quiver type.  Then we may identify the set of UV line defects with the lattice $\Gamma_{u}$.

From the analysis of \S \ref{framedmutation}, we know that the $\RG$ map selects a preferred particle half-space and hence a preferred seed $\mathfrak{s}.$  Let  $\mathcal{C}$ its indicate its associated cone as given by \eqref{conedef}.  We define a  dual cone, $\check{\mathcal{C}},$ by the condition that it pair positively with all elements in the seed
\begin{equation}
\check{\mathcal{C}}= \left \{\phantom{ \int} \hspace{-.14in} \check \gamma \in \Gamma_{u}\otimes_{\mathbb{Z}}\mathbb{R}| \langle \check{\gamma},\gamma\rangle\geq 0, \forall \gamma \in \mathcal{C}\right \}.
\end{equation}
Via the inverse of the renormalization map, the integral points of the dual cone $\check{\mathcal{C}}$ form a subset in the space of UV line defects.  According to the argument given at the conclusion of \S\ref{framedquivers}, for all defects with core charge $\gamma_{c}$ in $\check{\mathcal{C}},$  the framed quiver $Q[\gamma_{c}]$ computes the exact framed degeneracies.

A significant feature of the dual cone is that the UV line defects in $\check{\mathcal{C}}$ have universal OPEs.  Let $\alpha_{i}$ be the UV line defect label corresponding to a core charge $\gamma_{i}$,
\begin{equation}
\mathbf{RG}(\alpha_{i},\zeta, u)=\gamma_{i}.
\end{equation}
Then the OPE of line defects in the dual cone is simply additive in the core charge.  In other words, if $\gamma_{i} \in \check{\mathcal{C}},$ we have
\begin{equation}
\gamma_{1}+\gamma_{2}=\gamma_{3} \iff L(\alpha_{1},\zeta) *L(\alpha_{2},\zeta)=y^{\langle \gamma_{1}, \gamma_{2}\rangle}L(\alpha_{3},\zeta).  \label{coneope}
\end{equation}

To demonstrate the result \eqref{coneope}, recall from \S \ref{OPEsec} that the OPE of line defects is independent both of $\zeta$ and of the modulus $u$.  The latter enters our problem through the central charge configuration of the ordinary nodes of the quiver, while the former enters through the central charge of the framing node.  It follows that the OPE of line defects is in fact independent of any stability condition on the framed quiver and hence to compute the OPE coefficients we may assume a convenient configuration of central charges.  In particular, we may use this freedom to assume that
\begin{equation}
\arg(\zeta)<\arg(Z_{i}), \label{simplestab}
\end{equation}
for all nodes $i$ of the quiver $Q$.

Now, consider the topology of the framed quiver $Q[\gamma_{i}]$ when the core charge $\gamma_{i}$ lies in $\check{\mathcal{C}}$.  This takes the form indicated in Figure \ref{fig:framedsdual}.
\begin{figure}[here!]
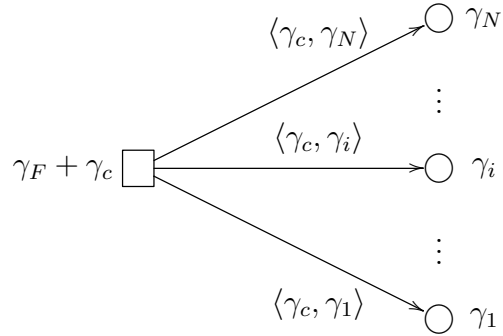

  \centering
$\xy
(10,0)+*{\bigcirc}="a"; (10,10)+*{\bigcirc}="b";  (10,-10)+*{\bigcirc}="d"; ; (13,-10)+*{\gamma_{1}}; (13,0)+*{\gamma_{i}}; (13,10)+*{\gamma_{N}}; (-20,0)*+{\color{white}{1}}*\frm{-} ="c"; (-15,0)+*{\gamma_{F}+\gamma_{c}}; (2,2)+*{\langle \gamma_{c}, \gamma_{i}\rangle};  (2,-9)+*{\langle \gamma_{c}, \gamma_{1}\rangle};  (2,9)+*{\langle \gamma_{c}, \gamma_{N}\rangle}; (10,5)+*{\vdots}; (10,-5)+*{\vdots};
\ar @{->} "c"; "a"
\ar @{->} "c"; "b"
\ar @{->} "c"; "d"
\endxy$
  \caption{Topology of a framed quiver $Q[\gamma_{c}]$ when the core charge occupies $\check{\mathcal{C}}$.  The framing node is indicated by a square, while the ordinary nodes of $Q$ are circles.  Arrows to between the nodes of $Q$ are suppressed. }
  \label{fig:framedsdual}
\end{figure}
In particular, the framing node is a source: arrows may exit this node but none enter it.

Let $R$ denote any framed representation of the quiver $Q[\gamma_{c}].$  Then $R$ admits a subrepresentation $S$ which is the identity away from the framing node and has vanishing support at the framing node.  From \eqref{simplestab}, we deduce that
\begin{equation}
\arg(S)>\arg(R).
\end{equation}
Thus $S$ is destabilizing unless it is the trivial zero representation.  Hence $R$ can be stable only if it is the single representation with vanishing halo charge.

The above yields a calculation of all stable framed BPS states of a UV line operator in $\check{\mathcal{C}}$ in the chamber \eqref{simplestab}.  This chamber coincides with the that identified in  \S\ref{framedquivers} from the Coulomb branch point of view.  The framed spectrum here is completely trivial. It consists of a single state with core charge $\gamma_{i}$ and vanishing spin. Thus, written in terms of the Heisenberg algebra of \S \ref{OPEnoncomm} the generating functional is simply
\begin{equation}
F(\alpha_{i},\zeta)=X_{\gamma_{i}}.
\end{equation}
We read off the desired OPE result \eqref{coneope} from the multiplication of these generating functionals.

While the line defects in $\check{\mathcal{C}}$ have a simple spectrum in the chamber specified by \eqref{simplestab}, in other chambers they may support a non-trivial set of states.  Nevertheless, the OPE result \eqref{coneope} continues to hold and provides interesting constraints on the spectrum.  As the states are encoded by the cohomology of quiver moduli, we may phrase this result purely in that language.  Let $\mathcal{M}_{\gamma}^{Q[\gamma_{c}]}$ denote the moduli of stable framed quiver representations with framing core charge $\gamma_{c}$ and halo charge $\gamma$.  And let $d_{\gamma}$ indicate the complex dimension of $\mathcal{M}_{\gamma}^{Q[\gamma_{c}]}.$  The generating functional of framed BPS states for $\gamma_{c}\in \check{\mathcal{C}}$ is given by
\begin{equation}
F(\gamma_{c}, u, \zeta, y)= \sum_{\gamma \in \mathcal{C}} \Omega_{\mathrm{Higgs}}^{Q[\gamma_{c}]}(\gamma,y)X_{\gamma_{c}+\gamma}
=\sum_{\gamma \in \mathcal{C}}\sum_{p,q=0}^{d_{\gamma}}h^{p,q}(\mathcal{M}_{\gamma}^{Q[\gamma_{c}]})(-1)^{p-q}y^{2p-d_{\gamma}}X_{\gamma_{c}+\gamma}
\end{equation}
where on the right-hand-side, the $\zeta$ and $u$ dependence is implicit in the dependence of the quiver moduli spaces on the central charges specifying the stability condition.  By multiplying these functionals together and comparing to \eqref{coneope} we obtain, for $\gamma_{c_{i}}\in \check{\mathcal{C}}$
\begin{equation}
\Omega_{\mathrm{Higgs}}^{Q[\gamma_{c_{1}}+\gamma_{c_{2}}]}(\gamma,y)=\sum_{\substack{\gamma_{1},\gamma_{2}\in \mathcal{C} \\ \gamma_{1}+\gamma_{2}=\gamma}}\Omega_{\mathrm{Higgs}}^{Q[\gamma_{c_{1}}]}(\gamma_{1},y)\Omega_{\mathrm{Higgs}}^{Q[\gamma_{c_{2}}]}(\gamma_{2},y)y^{\langle \gamma_{1},\gamma_{2}\rangle+\langle \gamma_{1},\gamma_{c_{2}}\rangle+\langle \gamma_{c_{1}},\gamma_{2}\rangle} \label{quiverrecursion}
\end{equation}
The above result is a non-trivial recursion relation on the cohomology of framed quiver representations.  We exhibit an explicit example of this formula in \S \ref{ADframedcomp}.

\subsubsection{Mutated Dual Cones}
\label{mutateddualseeds}

The simple features of the dual cones of line defects may be extended to a larger class of defects by utilizing the mutation technology of \S \ref{seeds}.

We utilize the map $\mathbf{RG}(,\zeta,u)$ to identify the set of line defects with the charge lattice $\Gamma$.  The lattice $\Gamma$ may in turn be identified with $\mathbb{Z}^{\mathrm{rk}(\Gamma)}$ by expressing any given element $\gamma$ in terms of the seed elements $\gamma_{i}$ of the preferred seed $\mathfrak{s}$ of \S \ref{framedmutation}.  Embedded inside $\mathbb{Z}^{\mathrm{rk}(\Gamma)}$ is a dual cone $\check{\mathcal{C}}_{\mathfrak{s}}$ of line defects obeying the simple algebra described in the previous section.

Now suppose that $\zeta$ or the moduli $u$ are varied.  Then the quiver may transform by a mutation $\mu$.  The mutated seed $\mu(\mathfrak{s})$ also has an associated dual cone of line defects $\check{\mathcal{C}}_{\mu(\mathfrak{s})}$, and line defects with core charge in $\check{\mathcal{C}}_{\mu(\mathfrak{s})}$ obey the simple OPE algebra \eqref{coneope}.  As a set, $\check{\mathcal{C}}_{\mu(\mathfrak{s})}$ exists inside the collection of UV defects of the fixed $\mathcal{N}=2$ field theory.  However to exhibit its embedding in $\mathbb{Z}^{\mathrm{rk}(\Gamma)}$, identified with the coefficient expansion of the core charge with respect to the original seed $\mathfrak{s}$, we must account for the fact that in the course of a mutation the map $\mathbf{RG}$ has undergone a discontinuity as described in \S \ref{framedmutation}.  Thus, we mutate back to the original seed and obtain $\check{\mathcal{C}}_{\mu(\mathfrak{s})}$ embedded inside $\mathbb{Z}^{\mathrm{rk}(\Gamma)}$.

For each seed related to $\mathfrak{s}$ by a sequence of mutations we obtain as above a dual cone embedded in $\mathbb{Z}^{\mathrm{rk}(\Gamma)}$.  Each such dual cone obeys the simple algebra stated in \eqref{coneope}.  To elucidate their geometry it is useful to be more explicit about their embedding.

Consider for example a seed related to $\mathfrak{s}$ by a left mutation $\mu_{Li}$ defined in equation \eqref{mutrule}.  Then, as a subset of $\mathbb{Z}^{\mathrm{rk}(\Gamma)}$, the dual cone $\check{\mathcal{C}}_{\mu_{L_{i}}(\mathfrak{s})}$ is defined to be those core charges $\gamma_{c}$ of the form
\begin{equation}
\gamma_{c}=\gamma+\langle \gamma,\gamma_{i}\rangle\gamma_{i},
\end{equation}
where the charge $\gamma$ satisfies the incidence relations defining $\check{\mathcal{C}}_{\mu_{L_{i}}(\mathfrak{s})}$ with respect to the seed $\mu_{L_{i}}(\mathfrak{s})$
\begin{equation}
\langle\gamma,\mu_{L_{i}}(\gamma_{j})\rangle \geq0 \Longrightarrow \langle \gamma, \gamma_{j} \rangle =\begin{cases}  \leq 0 & \mathrm{If} \ i=j \\
\geq \langle \gamma_{i}, \gamma\rangle \mathrm{Max}\left(\langle\gamma_{i},\gamma_{j} \rangle, 0\right) & \mathrm{If} \ i\neq j
\end{cases}
\end{equation}
From these inequalities we deduce that those core charges contained in the intersection $\check{\mathcal{C}}_{\mathfrak{s}}\bigcap \check{\mathcal{C}}_{\mu_{L_{i}}(\mathfrak{s})}$ are characterized by
\begin{equation}
\langle \gamma_{c}, \gamma_{i} \rangle =0, \hspace{.5in} \langle \gamma_{c}, \gamma_{j} \rangle \geq 0, \hspace {.2in}\forall i \neq j.
\end{equation}
This is a codimension one face of each of the cones in question.

The argument above may be repeated for any pair of seeds related by a mutation to demonstrate that the associated dual cones meet along a codimension one face.  More generally, we may examine the collection of seeds $\mathfrak{s}_{I}$ related to $\mathfrak{s}$ by an arbitrary sequence of mutations.  The cones $\check{\mathcal{C}}_{\mathfrak{s}_{I}}$ may either meet along codimension one faces as above, or, due to non-trivial topology in the space of mutations, they may coincide.  The totality of these cones and their relative geometry provides significant information about the OPE algebra of line defects.  We study explicit examples of these cones in \S \ref{dualseedsex}.

\subsubsection{Adapted Quivers, Cyclic, and Cocyclic Representations}
\label{adapted}

In \S \ref{framedmutation}, we have seen how the defect renormalization map $\mathbf{RG}$ selects a preferred particle half-space \eqref{rghalfdef}.  However, in explicit computations it is frequently useful to utilize the known behavior of the framed quiver $Q[\gamma_{c}]$ under changes in $\vartheta$ to change to a different half-space $\mathfrak{h}_{\vartheta}$.  In this section we describe the simplifications that may occur if $\mathfrak{h}_{\vartheta}$ is chosen judiciously, and relate our constructions of framed quivers and spectra to known concepts in representation theory.  As with the totality of \S\ref{dualseeds}, in the following we assume that the core charge $\gamma_{c}$ lies in a dual cone $\check{\mathcal{C}}_{\mathfrak{s}}$ so that the framed quiver construction of \S\ref{framedquivers} accurately computes the framed spectra.

We begin with the following definitions concerning framed quivers and their stability conditions.
\begin{itemize}
\item A triple $(Q[\gamma_{c}], \zeta, \mathfrak{h}_{\vartheta})$ is \emph{left adapted} if the framing ray $\zeta$ is the leftmost extremal occupied ray in the half-space $\mathfrak{h}_{\vartheta}$.
\item A triple $(Q[\gamma_{c}], \zeta, \mathfrak{h}_{\vartheta})$ is \emph{right adapted} if the framing ray $\zeta$ is the rightmost extremal occupied ray in the half-space $\mathfrak{h}_{\vartheta}$.
\end{itemize}
Examples of left and right adapted stability conditions are illustrated in Figure \ref{fig:adapt}.  For a given UV line defect $L(\alpha,\zeta)$, it may happen that the framed quiver implied by the map $ {\mathbf{RG}}(\cdot, \zeta, u)$ is left or right adapted.  However, more generally the left and right adapted framed quivers are simply useful computational devices.  Frequently, we may relate calculations carried out with adapted quivers to those carried out with the physical framed quiver implied by $ {\mathbf{RG}}(\cdot, \zeta, u)$ by a sequence of seed mutations and wall-crossings.

\begin{figure}[here!]
  \centering
  \subfloat[Left Adapted Stability Condition]{\label{fig:Ladapt}\includegraphics[width=0.45\textwidth]{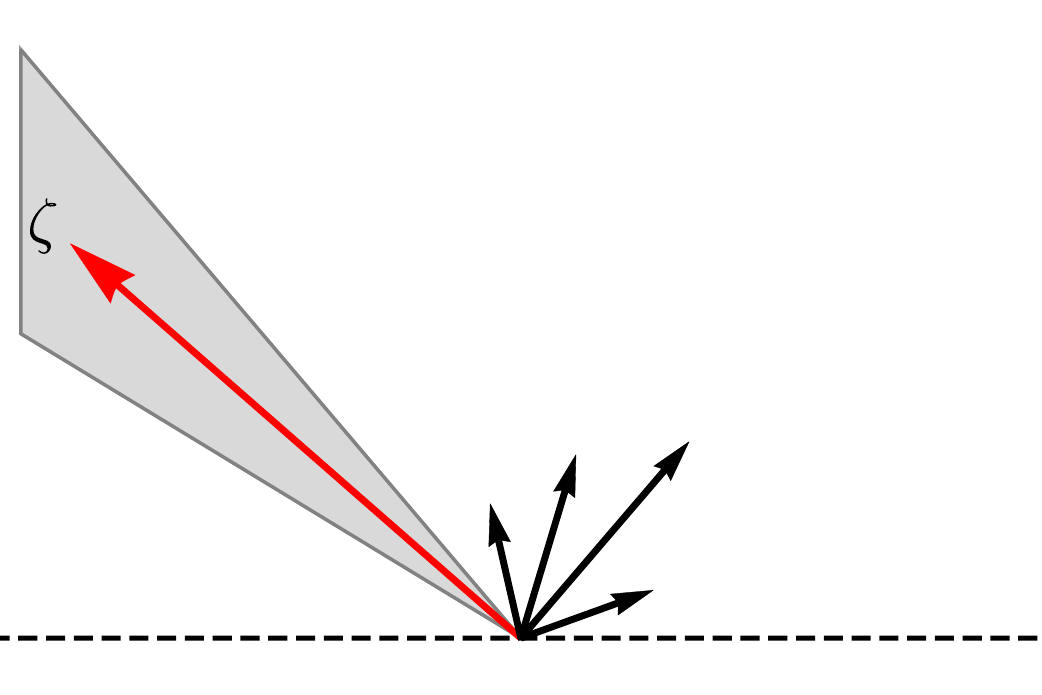}}
  \hspace{.4in}
  \subfloat[Right Adapted Stability Condition]{\label{fig:Radapt}\includegraphics[width=0.45\textwidth]{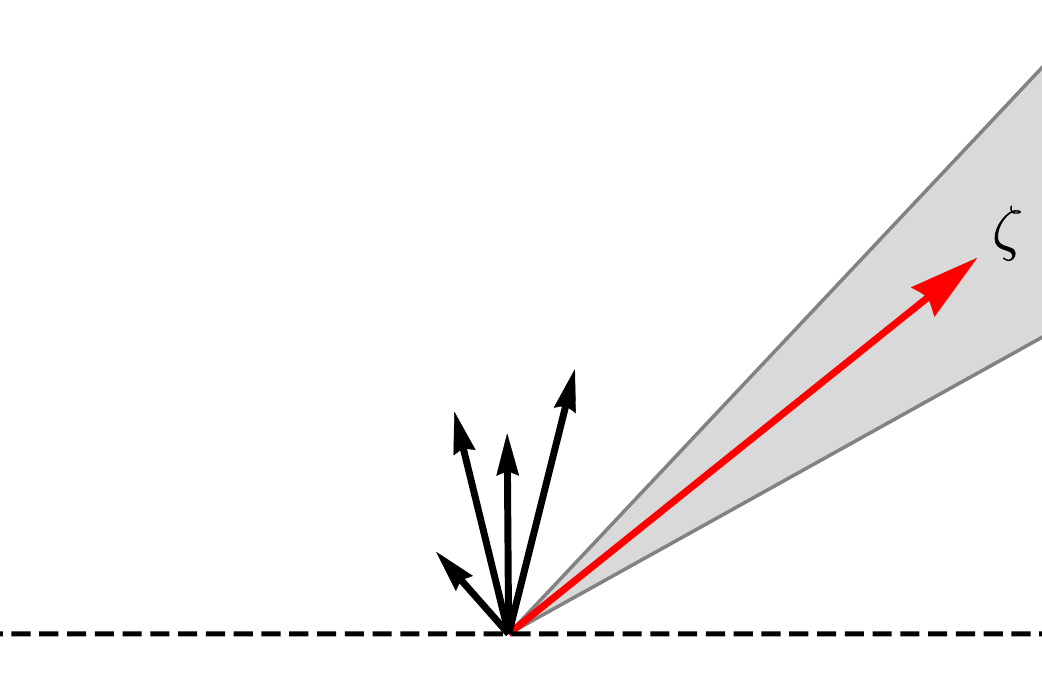}}
  \caption{Adapted stability conditions plotted in the particle half-space.  The framing ray is indicated in red and possible halo rays are indicated in black.  The gray shaded region is the infinitesimal cone containing framed BPS states.  In (a) the stability condition is left adapted.  In (b) the stability condition is right adapted.}
  \label{fig:adapt}
\end{figure}

The utility of these definitions is that if a framed quiver and stability condition are adapted, then the notion of a stable framed representation simplifies to a purely algebraic condition on quiver representations.  Indeed, consider a general framed quiver and framed representation $R$.  Let $V$ denote the one-dimensional vector space supported at the framing node.  Then $R$ has a distinguished set of subrepresentations, namely those which contain the subspace $V,$ and there are two natural classes of framed representations defined utilizing $V$.
\begin{itemize}
\item A pair $(R,V)$ is \emph{cyclic} if the only non-zero subrepresentation of $R$ containing $V$ is $R$.
\item A pair $(R,V)$ is \emph{cocyclic} if all non-zero subrepresentations of $R$ contain $V$.
\end{itemize}
Note that cyclic and cocyclic representations are algebraic concepts that make no reference to any notion of stability condition.

The importance of these distinguished classes of framed representations is that, if the framed quiver and stability condition is left adapted, then a framed representation is stable if and only if it is cyclic.  And similarly, if a framed quiver and stability condition is right adapted, then a framed representation is stable if and only if it is cocyclic.  This result is a direct consequence of the ray geometry illustrated in Figure \ref{fig:adapt} together with the definition of the framing limit stability condition \eqref{limitZ}.

Moduli spaces of cyclic and cocyclic quiver representations are particularly amenable to analysis using localization techniques.  They have appeared previously in physics in the related context of counting D-brane bound states in non-compact Calabi-Yau geometries where a heavy non-compact brane forms the framing node  \cite{MR2403807, MR2592501, Ooguri:2008yb}.

\section{Examples of Framed Spectra and OPEs}
\label{examples}

In this section we study explicit examples of framed quivers and spectra, and associated line defect OPEs.  We focus on the class of $\mathcal{N}=2$ theories whose unframed quivers have two nodes and $k>0$ arrows, the so-called Kronecker quivers $Q_{k}.$  For $k=1$ the UV theory is the Argyres-Douglas CFT defined in \cite{Argyres:1995jj, Argyres:1995xn}.  For $k=2$ the UV theory is $SU(2)$ super Yang-Mills studied in \cite{Seiberg:1994rs, Seiberg:1994aj}.  For $k>2$ this quiver is not known to occur as a complete description of the spectrum of any UV theory, however it occurs as a subsector many $\mathcal{N}=2$ theories of quiver type, including in particular  $SU(n)$ super Yang-Mills with $n>2$ \cite{Galakhov:2013oja}.

Throughout the remainder of this section we fix an initial point $u$ on the moduli space and an initial seed $\gamma_{1}, \gamma_{2}.$   Associated to this seed is a pair of variables $X_{\gamma_{1}},$ and $X_{\gamma_{2}}$ generating the Heisenberg algebra of \S  \ref{OPEnoncomm}
\begin{equation}
X_{\gamma_{1}}X_{\gamma_{2}}=y^{2k}X_{\gamma_{2}}X_{\gamma_{1}}.
\end{equation}
Using the defining relations \eqref{heisenberg}, all remaining variables $X_{\gamma}$ in the Heisenberg algebra may be written as monomials in the two generators above.  The quivers of interest are illustrated in Figure \ref{fig:kronecker}.

\begin{figure}[here!]
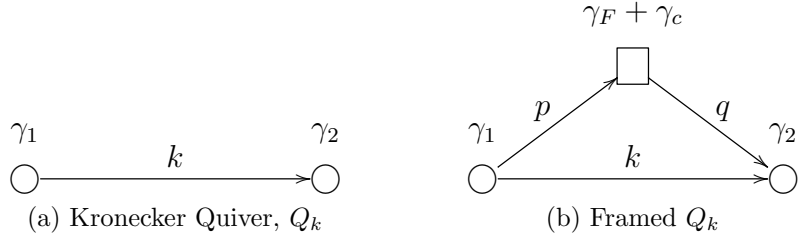

  \centering
  \subfloat[Kronecker Quiver, $Q_{k}$]{\label{fig:qk}\xy
(-10,0)+*{\bigcirc}="a"; (10,0)+*{\bigcirc}="b";  (-10,3)+*{\gamma_{1}}; (10,3)+*{\gamma_{2}}; (0,1.5)+*{k};
\ar @{->} "a"; "b"
\endxy}
  \hspace{.6in}
 \subfloat[Framed $Q_{k}$]{\label{fig:qkf}\xy
(-10,0)+*{\bigcirc}="a"; (10,0)+*{\bigcirc}="b";  (-10,3)+*{\gamma_{1}}; (10,3)+*{\gamma_{2}}; (0,15)*+{\color{white}{1}}*\frm{-} ="c"; (0,11)+*{\gamma_{F}+\gamma_{c}}; (0,1.5)+*{k};  (-6,4.5)+*{p};  (6,4.5)+*{q};
\ar @{->} "a"; "b"
\ar @{->} "a"; "c"
\ar @{->} "c"; "b"
\endxy}
  \caption{Example quivers.  In (a) the Kronecker quiver $Q_{k}$.  In (b) a framed Kronecker quiver.  The framing node is represented as a square.  The core charge is encoded by the arrows to and from the framing node indicated by $(p,q)\in \mathbb{Z}\oplus \mathbb{Z}$.}
  \label{fig:kronecker}
\end{figure}

\subsection{Seeds and Dual Cones}
\label{dualseedsex}
We first exhibit the partition of the space of line operators into cones dual to seeds as described in \S \ref{dualseeds}.

To carry out the computation it is fruitful to observe that the action of mutation on the Kronecker quiver $Q_{k}$ is again $Q_{k}.$  Thus, in this example we may view mutations as linear transformations on the charge labels of the nodes of the quiver.  Let $M_{i}$ for $i=1,2$ indicate the transformation on charge labels defined by left mutation at the source node, $i=1,$ and the sink node, $i=2,$ respectively.  Then $M_{i}$ are given by the following $SL(2,\mathbb{Z})$ transformations
\begin{equation}
M_{1}=\left(\begin{array}{cc}k & -1 \\1 & 0\end{array}\right), \hspace{.5in} M_{2}= \left(\begin{array}{cc}0 & 1 \\ -1 & 0\end{array}\right). \label{mutationmatrices}
\end{equation}
Similarly, right mutations are given by the inverse transformations.

The integer $k$ controls the eigenvalues $\lambda_{\pm}$ of $M_{1},$ and there are three qualitatively distinct cases.
\begin{itemize}
\item $k=1$.  The eigenvalues are complex roots of unity, $\lambda_{\pm}=\frac{1\pm i\sqrt{3}}{2}$. The matrix $M_{1}$ is a periodic, elliptic transformation.
\item $k=2$.  The eigenvalues are degenerate, $\lambda_{\pm}=1.$  The matrix $M_{1}$ is not of finite order and is a parabolic transformation.
\item $k>2$.  The eigenvalues are real and reciprocals, $\lambda_{\pm}=\frac{k\pm \sqrt{k^{2}-4}}{2}$.  The matrix $M_{1}$ is not of finite order and is a hyperbolic transformation.
\end{itemize}
We consider each of these cases in turn.  We shall see that they exemplify three possible phenomena.  When $k=1$ the dual cones are finite in number and fill the entire charge lattice.  When $k=2,$ there are infinitely many dual cones and they accumulate.  The complement of the set of dual cones is a single ray.  Finally, when $k>2,$ there are infinitely many cones which accumulate and whose complement is an open set.

\subsubsection{$k=1:$ Argyres-Douglas Theory}
\label{adconesanalysis}
This is the case relevant to the Argyres-Douglas CFT.  From the initial seed $\{\gamma_{1},\gamma_{2}\}$ we generate five seeds by mutation.
\begin{equation}
\xy
(-10,0)+*{\{\gamma_{1},\gamma_{2}\}}="a"; (10,0)+*{\{\gamma_{1}+\gamma_{2}, -\gamma_{1}\}}="b";  (30,0)+*{\{\gamma_{2}, -\gamma_{1}-\gamma_{2}\}}="c"; (50,0)+*{\{-\gamma_{1},-\gamma_{2}\}}="d";  (20,-5)+*{\{-\gamma_{2},\gamma_{1}\}}="e";
(0,1.5)+*{\mu_{L_{1}}}="f"; (20,1.5)+*{\mu_{L_{1}}}="g";  (40,1.5)+*{\mu_{L_{1}}}="h";  (5,-3.9)+*{\mu_{L_{2}}}="i"; (35,-3.9)+*{\mu_{L_{2}}}="j";
\ar @{->} "a"; "b"
\ar @{->} "b"; "c"
\ar @{->} "c"; "d"
\ar @{->} "a"; "e"
\ar @{->} "e"; "d"
\endxy
\end{equation}

Associated to each seed indicated above is a dual cone $\check{\mathcal{C}}$ in the space of line operators.  We aim to exhibit these cones using the identification of the UV line operators with the charge lattice provided by the initial seed.  The dual cone associated to the initial seed is
\begin{equation}
\check{\mathcal{C}}_{\{\gamma_{1}, \gamma_{2}\}}=\left\{\phantom{\int}\hspace{-.18in}p \gamma_{1}-q\gamma_{2}| p,q\geq0\right\}.
\end{equation}

To obtain the cones dual to other seeds we first find the those charges which pair positively with the seed in question, and then we mutate back to our original seed to see how this cone is embedded in the fixed charge lattice.  Thus for instance
\begin{equation}
\check{\mathcal{C}}_{\{\gamma_{1}+\gamma_{2}, -\gamma_{1}\}}=\mu_{R_{2}}\left(\left\{\phantom{\int}\hspace{-.18in}p (\gamma_{1}+\gamma_{2})+q\gamma_{1}| p,q\geq0\right\}\right)=\left\{\phantom{\int}\hspace{-.18in}p \gamma_{2}+q\gamma_{1}| p,q\geq0\right\}.
\end{equation}
Similarly, the remaining dual cones are given by
\begin{eqnarray}
\check{\mathcal{C}}_{\{\gamma_{2}, -\gamma_{1}-\gamma_{2}\}} &= &\left\{\phantom{\int}\hspace{-.18in}-p \gamma_{1}+q\gamma_{2}| p,q\geq0\right\}, \nonumber\\
\check{\mathcal{C}}_{\{-\gamma_{1}, -\gamma_{2}\}}& = &\left\{\phantom{\int}\hspace{-.18in}-p(\gamma_{1}+ \gamma_{2})-q\gamma_{1}| p,q\geq0\right\}, \\
\check{\mathcal{C}}_{\{-\gamma_{2}, \gamma_{1}\}}& = &\left\{\phantom{\int}\hspace{-.18in}-p \gamma_{2}-q(\gamma_{1}+\gamma_{2})| p,q\geq0\right\}. \nonumber
\end{eqnarray}
The geometry of these five cones is illustrated in Figure \ref{fig:adcones}.
\begin{figure}[here!]
  \centering
 \includegraphics[width=0.45\textwidth]{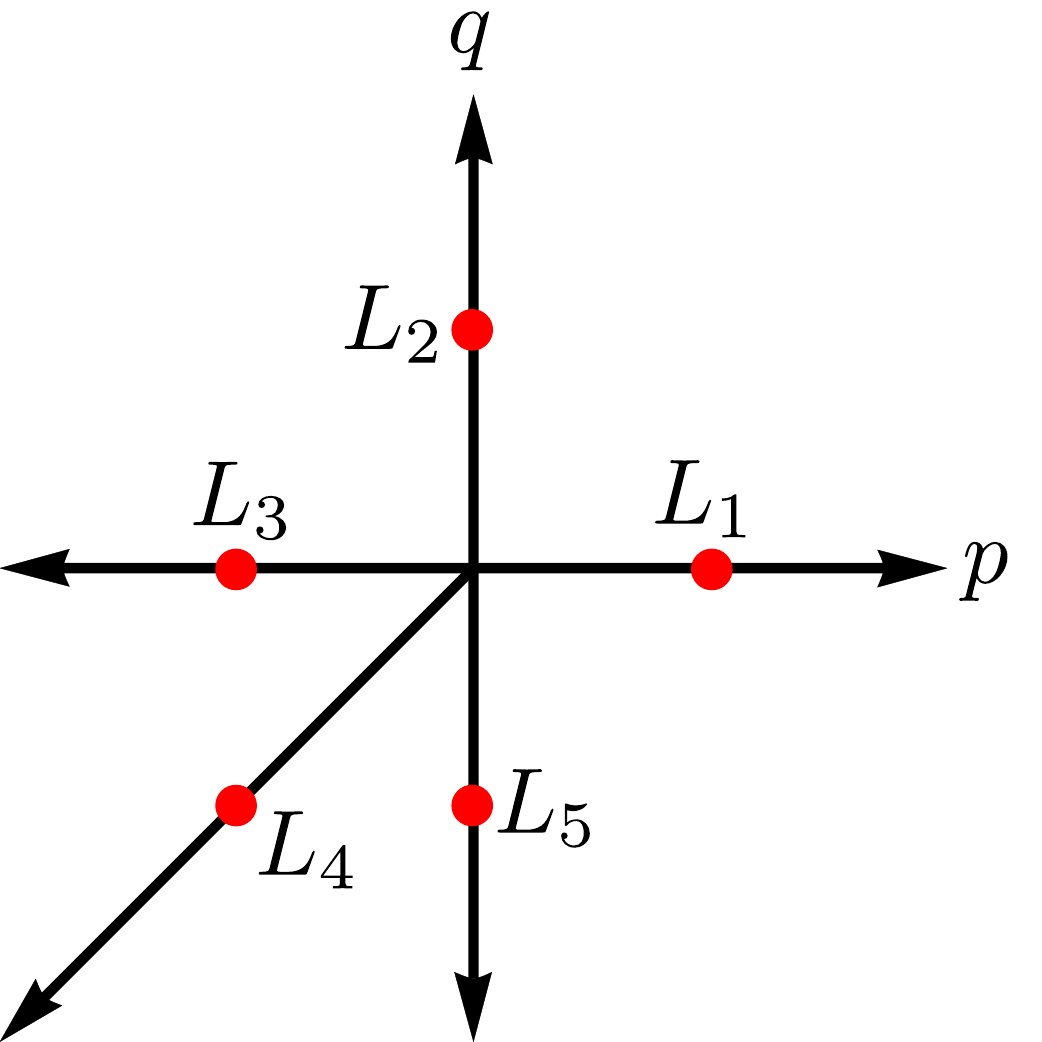}
  \caption{Geometry of dual cones for the Kronecker quiver $Q_{1}$.  The space $\Gamma\otimes_{\mathbb{Z}}\mathbb{R}$ is identified with $\mathbb{R}\oplus \mathbb{R}$ by expressing charges in terms of the seed as $\gamma=p\gamma_{1}+q\gamma_{2}.$ The boundaries of the cones are indicated by the black arrows.  The red dots denote the generators $L_{i}.$ }
  \label{fig:adcones}
\end{figure}

We observe that there are five generators of the boundary rays of the cones, and we label them $L_{i},$ with $i=1, \cdots 5$.  Using the $\mathbf{RG}$ map associated to the initial seed we may alternatively label these UV defects by their core charge as
\begin{equation}
\mathbf{RG}(L_{1})=\gamma_{1}, \hspace{.2in}\mathbf{RG}(L_{2})=\gamma_{2}, \hspace{.2in} \mathbf{RG}(L_{3})=-\gamma_{1}, \hspace{.2in}\mathbf{RG}(L_{4})=-\gamma_{1}-\gamma_{2}, \hspace{.2in}\mathbf{RG}(L_{5})=-\gamma_{2}. \label{ADgens}
\end{equation}
Together with the unit operator, the $L_{i}$ generate the OPE algebra of BPS line defects in the Argyres Douglas CFT.  In particular, as the five cones of Figure \ref{fig:adcones} completely fill the charge lattice $\Gamma,$ any line defect other than the $L_{i}$ occurs in the OPE of the $L_{i}$.  For instance according to the general arguments of \S \ref{dualseeds} we have
\begin{equation}
\mathbf{RG}(L)=p\gamma_{1}+q\gamma_{2}, \hspace{.2in} p,q\geq0 \Longrightarrow \underbrace{L_{1}*L_{1}* \cdots * L_{1}}_{p}*\underbrace{L_{2}* L_{2}*\cdots *L_{2}}_{q}=y^{pq}L.
\end{equation}
We complete our analysis of the OPE by determining the framed BPS spectra of the five generators in \S \ref{ADframedcomp}.

\subsubsection{$k=2: SU(2)$ SYM}
This is the case relevant to the $SU(2)$ SYM theory.  From the initial seed $\{\gamma_{1},\gamma_{2}\}$ we generate an infinite sequence of seeds by mutation.
\begin{equation}
\xy
(0,0)+*{\{\gamma_{1},\gamma_{2}\}}="a"; (0,7)+*{\{2\gamma_{1}+\gamma_{2}, -\gamma_{1}\}}="b";    (15,7)+*{\phantom{a}\cdots\phantom{a}}="c";  (40,7)+*{\{\gamma_{1}+n(\gamma_{1}+\gamma_{2}), \gamma_{2}-n(\gamma_{1}+\gamma_{2})\}}="d"; (65,7)+*{\phantom{a}\cdots\phantom{a}}="e"; (0,-7)+*{\{-\gamma_{2}, 2\gamma_{2}+\gamma_{1}\}}="f";    (15,-7)+*{\phantom{a}\cdots\phantom{a}}="g";  (40,-7)+*{\{\gamma_{1}-n(\gamma_{1}+\gamma_{2}), \gamma_{2}+n(\gamma_{1}+\gamma_{2})\}}="h"; (65,-7)+*{\phantom{a}\cdots\phantom{a}}="i"; (2.5,3.5)+*{\mu_{L_{1}}}="j"; (2.5,-3.5)+*{\mu_{R_{2}}}="k"; (9.5,5)+*{\mu_{L_{1}}}="l"; (9.5,-5)+*{\mu_{R_{2}}}="m" ; (22,5)+*{\mu_{L_{1}}}="n" ; (22,-5)+*{\mu_{R_{2}}}="o"; (58.5,5)+*{\mu_{L_{1}}}="p" ; (58.5,-5)+*{\mu_{R_{2}}}="q"
\ar @{->} "a"; "b"
\ar @{->} "b"; "c"
\ar @{->} "c"; "d"
\ar @{->} "d"; "e"
\ar @{->} "a"; "f"
\ar @{->} "f"; "g"
\ar @{->} "g"; "h"
\ar @{->} "h"; "i"
\endxy
\end{equation}
From these seeds we extract an infinite number of dual cones.  We embed these cones in the charge lattice $\Gamma$ using the $\mathbf{RG}$ map associated to the initial seed.  Let $n\geq0$, be a non-negative integer.  Then we may label the cones as
\begin{eqnarray}
\check{\mathcal{A}}&= &\left\{\phantom{\int}\hspace{-.18in}p \gamma_{1}+q\gamma_{2}| p,q\geq0\right\}, \nonumber\\
\check{\mathcal{B}}_{n}& = &\left\{\phantom{\int}\hspace{-.18in}p(\gamma_{1}-n(\gamma_{1}+ \gamma_{2}))-q(\gamma_{2}+n(\gamma_{1}+\gamma_{2}))| p,q\geq0\right\}, \\
\check{\mathcal{C}}_{n}& = &\left\{\phantom{\int}\hspace{-.18in}-p(\gamma_{1}+n(\gamma_{1}+ \gamma_{2}))+q(\gamma_{2}-n(\gamma_{1}+\gamma_{2}))| p,q\geq0\right\}. \nonumber
\end{eqnarray}

Notice that in this example the cones do not cover the entire charge lattice $\Gamma$.  Rather, the cones accumulate against the ray spanned by the vector $-(\gamma_{1}+\gamma_{2})$.  Physically, this is the ray spanned by purely electric Wilson lines in the ultraviolet.  The geometry of the cones is illustrated in Figure \ref{fig:su2cones}.
\begin{figure}[here!]
  \centering
 \includegraphics[width=0.45\textwidth]{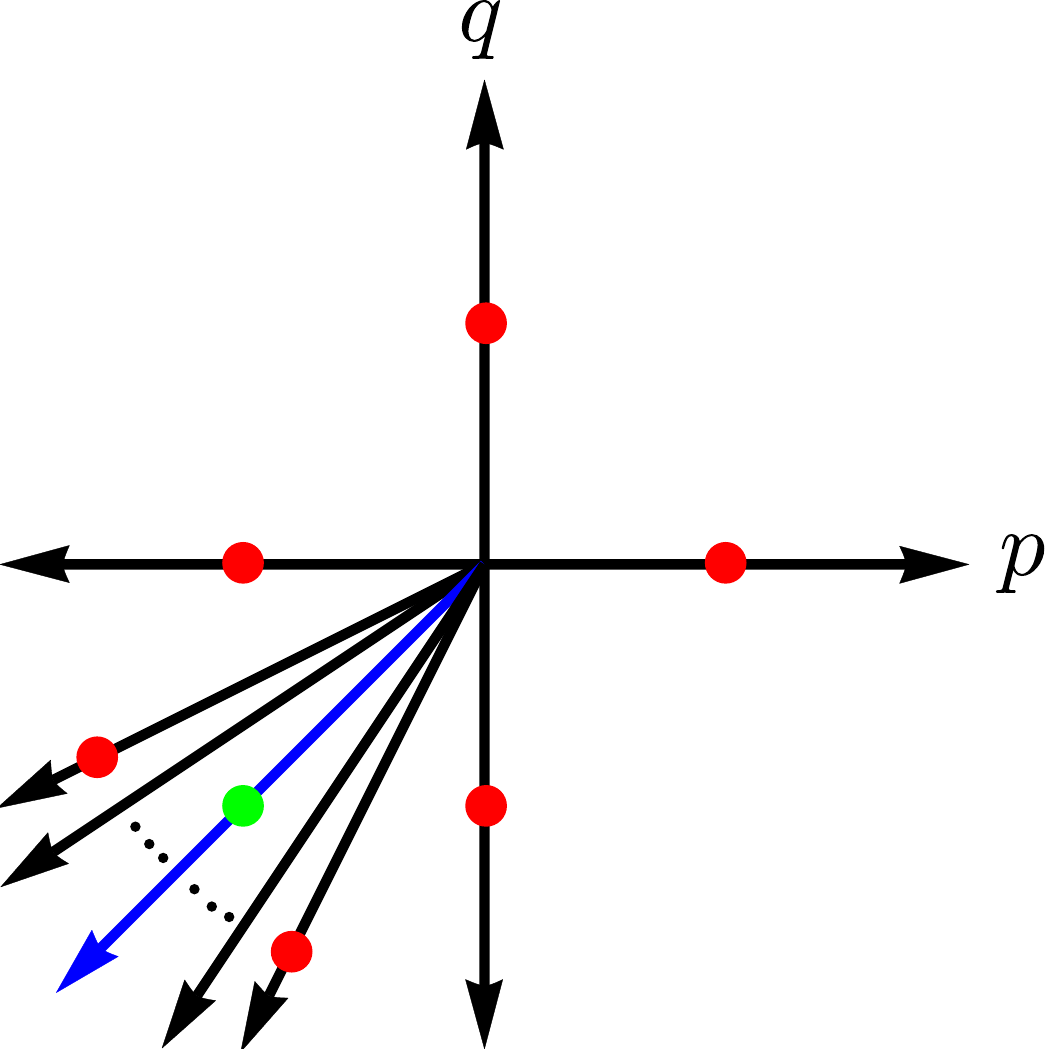}
  \caption{Geometry of dual cones for the Kronecker quiver $Q_{2}$.  The space $\Gamma\otimes_{\mathbb{Z}}\mathbb{R}$ is identified with $\mathbb{R}\oplus \mathbb{R}$ by expressing charges in terms of the seed as $\gamma=p\gamma_{1}+q\gamma_{2}.$ The boundaries of the cones are indicated by the black arrows.  The cones accumulate against the blue ray.  The red dots denote the integral generators on the boundaries of the cones. The green dot denotes an integral point on the limiting ray.}
  \label{fig:su2cones}
\end{figure}

The integral points on the boundaries of the cones, as well as the integral points on the limiting ray, are a natural generating set for the semiring of line defects.  They form three countably infinite families $L_{n}, L_{n}^{*},$ and $W_{n}.$ We label them by their core charge as
\begin{equation}
\mathbf{RG}(L_{n})=\gamma_{1}-n(\gamma_{1}+\gamma_{2}), \hspace{.2in}\mathbf{RG}(L^{*}_{n})=\gamma_{2}-n(\gamma_{1}+\gamma_{2}), \hspace{.2in}\mathbf{RG}(W_{n})=-\frac{n}{2}\left(\gamma_{1}+\gamma_{2}\right). \label{wilsondyons}
\end{equation}
The multiplication rules for any two defects in a fixed cone follows from the general analysis of \S \ref{dualseeds}.  We  study the framed states of the line defects $W_{n}$ occupying the limiting ray in \S \ref{wilsonlineex}.

\subsubsection{$k>2$}
This case is not known to describe any UV complete field theory but occurs as a subsector of many known models \cite{Galakhov:2013oja}.  From the initial seed $\{\gamma_{1},\gamma_{2}\}$ we again generate an infinite sequence of seeds by mutation.
\begin{equation}
\xy
(0,0)+*{\{\gamma_{1},\gamma_{2}\}}="a"; (0,7)+*{\{k\gamma_{1}+\gamma_{2}, -\gamma_{1}\}}="b";    (30,7)+*{\{(k^{2}-1)\gamma_{1}+k\gamma_{2}, -k\gamma_{1}-\gamma_{2}\}}="c";  (55,7)+*{\phantom{a}\cdots\phantom{a}}="d";  (0,-7)+*{\{-\gamma_{2}, k\gamma_{2}+\gamma_{1}\}}="f";    (30,-7)+*{\{(-k\gamma_{2}-\gamma_{1}, (k^{2}-1)\gamma_{2}+k\gamma_{1}\}}="g"; (55,-7)+*{\phantom{a}\cdots\phantom{a}}="h";  (2.5,3.5)+*{\mu_{L_{1}}}="j"; (2.5,-3.5)+*{\mu_{R_{2}}}="k"; (12.5,5)+*{\mu_{L_{1}}}="l"; (12.5,-5)+*{\mu_{R_{2}}}="m" ; (48.5,5)+*{\mu_{L_{1}}}="p" ; (48.5,-5)+*{\mu_{R_{2}}}="q"
\ar @{->} "a"; "b"
\ar @{->} "b"; "c"
\ar @{->} "c"; "d"
\ar @{->} "a"; "f"
\ar @{->} "f"; "g"
\ar @{->} "g"; "h"
\endxy
\end{equation}
By raising the matrices appearing in \eqref{mutationmatrices} to arbitrary powers, we deduce that the general seed obtained by a sequence of $\ell\geq0$ left or right mutations takes the form
\begin{eqnarray}
\mu_{L_{1}}^{\ell}\left(\phantom{\int}\hspace{-.18in}{\{\gamma_{1},\gamma_{2}\}}\right) & = & \{a_{\ell}\gamma_{1}+a_{\ell-1}\gamma_{2}, -a_{\ell-1}\gamma_{1}-a_{\ell-2}\gamma_{2}\}, \\
\mu_{R_{2}}^{\ell}\left(\phantom{\int}\hspace{-.18in}{\{\gamma_{1},\gamma_{2}\}}\right) & = & \{-a_{\ell-1}\gamma_{2}-a_{\ell-2}\gamma_{1},a_{\ell}\gamma_{2}+a_{\ell-1}\gamma_{1}\}, \nonumber
\end{eqnarray}
where the coefficients $a_{q}$ appearing above are defined by the following functions of $k$
\begin{equation}
a_{2n} \equiv  \sum_{j=0}^{n}\binom{n+j}{2j}(-1)^{n+j}k^{2j}-\delta_{-1,n}, \hspace{.5in}  a_{2n+1} \equiv  \sum_{j=0}^{n}\binom{n+j+1}{2j+1}(-1)^{n+j}k^{2j+1}.
\end{equation}
They satisfy the recurrence relation
\begin{equation}
a_{\ell+1}=ka_{\ell}-a_{\ell-1}.
\end{equation}

From these seeds we again extract an infinite number of dual cones.  We embed these cones in the charge lattice $\Gamma$ using the $\mathbf{RG}$ map associated to the initial seed.  Let $n\geq0$, be a non-negative integer.  Then we may label the cones as
\begin{eqnarray}
\check{\mathcal{A}}&= &\left\{\phantom{\int}\hspace{-.18in}p \gamma_{1}+q\gamma_{2}| p,q\geq0\right\}, \nonumber\\
\check{\mathcal{B}}_{n}& = &\left\{\phantom{\int}\hspace{-.18in}-p\left(a_{n-1}\gamma_{1}+a_{n}\gamma_{2}\right)-q\left(a_{n-2}\gamma_{1}+a_{n-1}\gamma_{2}\right)| p,q\geq0\right\}, \label{kcones} \\
\check{\mathcal{C}}_{n}& = &\left\{\phantom{\int}\hspace{-.18in}-p\left(a_{n}\gamma_{1}+a_{n-1}\gamma_{2}\right)-q\left(a_{n-1}\gamma_{1}+a_{n-2}\gamma_{2}\right)| p,q\geq0\right\}. \nonumber
\end{eqnarray}

The cones defined in \eqref{kcones} do not cover the entire charge lattice.  As $n$ tends to infinity, the cones  $\check{\mathcal{B}}_{n}$, and $\check{\mathcal{C}}_{n}$ degenerate to rays whose slope is controlled by the eigenvalues of the mutation matrix $\eqref{mutationmatrices}$
\begin{eqnarray}
\check{\mathcal{B}}_{n}&\longrightarrow& \check{\mathcal{B}}_{\infty}=\left\{\phantom{\int}\hspace{-.18in}-p \left(\gamma_{1}+\left(\frac{k+\sqrt{k^{2}-4}}{2}\right)\gamma_{2}\right)| p\geq0\right\},\\
\check{\mathcal{C}}_{n}&\longrightarrow &\check{\mathcal{C}}_{\infty}=\left\{\phantom{\int}\hspace{-.18in}-p \left(\gamma_{1}+\left(\frac{k-\sqrt{k^{2}-4}}{2}\right)\gamma_{2}\right)| p\geq0\right\}. \nonumber
\end{eqnarray}
These limiting rays have irrational slope and hence contain no integral points of $\Gamma$.  There is an open set $\mathcal{U}$ in $\Gamma \otimes_{\mathbb{Z}}\mathbb{R}$ which is spanned by the limiting rays
\begin{equation}
\mathcal{U}=\left\{\phantom{\int}\hspace{-.18in}-p \left(\gamma_{1}+\left(\frac{k+\sqrt{k^{2}-4}}{2}\right)\gamma_{2}\right)-q \left(\gamma_{1}+\left(\frac{k-\sqrt{k^{2}-4}}{2}\right)\gamma_{2}\right)| p,q>0\right\}.
\end{equation}
Defects with core charge in $\mathcal{U}$ do not lie in any cone.  The geometry is illustrated in Figure \ref{fig:kcones}.
\begin{figure}[here!]
  \centering
 \includegraphics[width=0.45\textwidth]{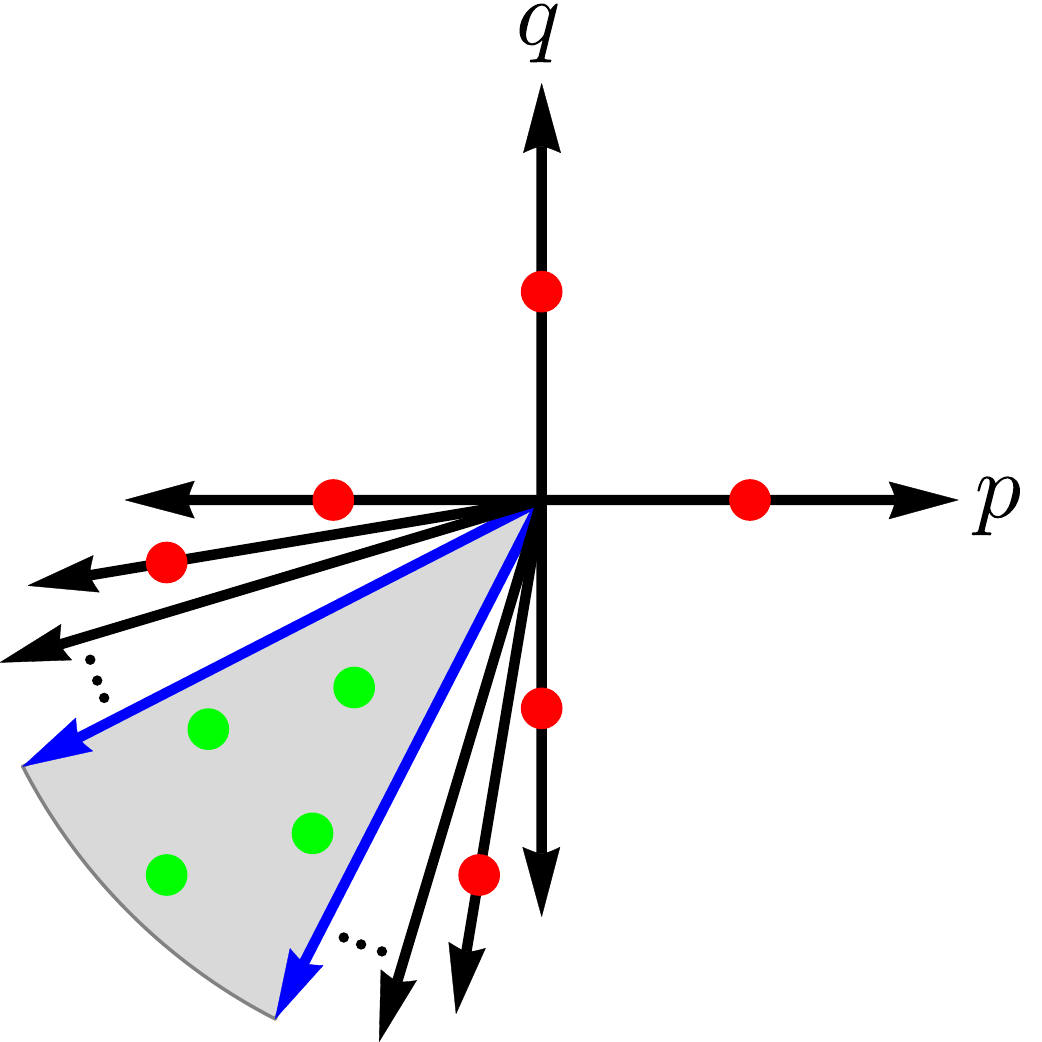}
  \caption{Geometry of dual cones for the Kronecker quiver $Q_{k}$ with $k>2$.  The space $\Gamma\otimes_{\mathbb{Z}}\mathbb{R}$ is identified with $\mathbb{R}\oplus \mathbb{R}$ by expressing charges in terms of the seed as $\gamma=p\gamma_{1}+q\gamma_{2}.$ The boundaries of the cones are indicated by the black arrows.  The cones accumulate against the blue rays.  The gray shaded area is the region $\mathcal{U}$ of core charges not contained in any cone.  The red dots denote the integral generators on the boundaries of the cones. The green dots denote integral points in $\mathcal{U}$.}
  \label{fig:kcones}
\end{figure}

As in our previous examples, the geometry of the cones suggests a natural generating set of the algebra of line defects.  We include the defects whose core charges are given by the integral points spanning the boundaries of the cones, denoted $L_{n}$ and $L_{n}^{*}$
\begin{equation}
\mathbf{RG}(L_{n})=-a_{n-2}\gamma_{1}-a_{n-1}\gamma_{2}, \hspace{.5in} \mathbf{RG}(L_{n}^{*})=-a_{n-1}\gamma_{1}-a_{n-2}\gamma_{2}.
\end{equation}
And, in addition, the generating set contains all defects $U_{r,s}$ whose core charge is an integral point in $\mathcal{U}$
\begin{equation}
\mathbf{RG}(U_{r,s})=r \gamma_{1}+s\gamma_{2}\in \mathcal{U}.
\end{equation}
The exploration of the algebra of these defects is an interesting topic for future analysis.

\subsection{Framed Spectra and OPEs}
\label{explicitspectra}

We now determine examples of framed spectra.  For the case $k=1$ we determine the spectra of the five generators $L_{i}.$  From this we extract the complete non-commutative OPE and verify its non-trivial predictions for the cohomology of quiver moduli.  For the case $k=2,$ we determine the framed states supported by the Wilson lines in $SU(2)$ SYM.  The latter requires input from a Coulomb branch analysis.

\subsubsection{$k=1:$ Argyres-Douglas Theory}
\label{ADframedcomp}

We first exhibit the spectrum of framed BPS states supported by the five defects $L_{i}$ defined in \eqref{ADgens}.  As each such defect lies in a dual cone, we may perform the analysis using a Higgs branch computation in quiver quantum mechanics.  We assume that the stability condition is chosen so that the framed quiver is right-adapted in the sense of \S \ref{adapted} so that
\begin{equation}
\arg(\zeta)<\arg(Z_{\gamma_{i}}), \hspace{.5in} \forall i.
\end{equation}

\begin{figure}[here!]
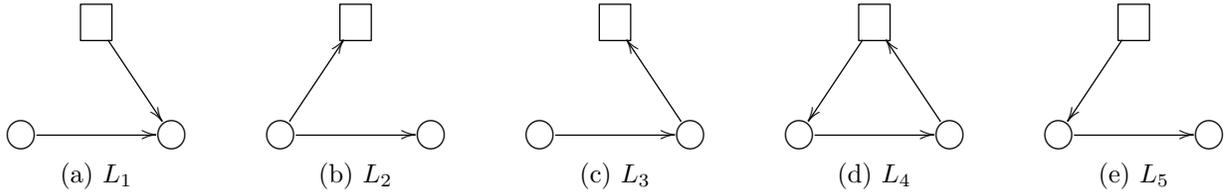

  \centering
 \subfloat[$L_{1}$]{\label{fig:L1}\xy
(0,0)+*{\bigcirc}="a"; (10,0)+*{\bigcirc}="b";  (10,15)*+{\color{white}{1}}*\frm{-} ="c";
\ar @{->} "a"; "b"
\ar @{->} "c"; "b"
\endxy}
\hspace{.35in}
 \subfloat[$L_{2}$]{\label{fig:L2}\xy
(0,0)+*{\bigcirc}="a"; (10,0)+*{\bigcirc}="b";  (10,15)*+{\color{white}{1}}*\frm{-} ="c";
\ar @{->} "a"; "b"
\ar @{->} "a"; "c"
\endxy}
\hspace{.35in}
 \subfloat[$L_{3}$]{\label{fig:L3}\xy
(0,0)+*{\bigcirc}="a"; (10,0)+*{\bigcirc}="b";  (10,15)*+{\color{white}{1}}*\frm{-} ="c";
\ar @{->} "a"; "b"
\ar @{->} "b"; "c"
\endxy}
\hspace{.35in}
 \subfloat[$L_{4}$]{\label{fig:L4}\xy
(0,0)+*{\bigcirc}="a"; (10,0)+*{\bigcirc}="b";  (10,15)*+{\color{white}{1}}*\frm{-} ="c";
\ar @{->} "a"; "b"
\ar @{->} "b"; "c"
\ar @{->} "c"; "a"
\endxy}
\hspace{.35in}
 \subfloat[$L_{5}$]{\label{fig:L5}\xy
(0,0)+*{\bigcirc}="a"; (10,0)+*{\bigcirc}="b";  (10,15)*+{\color{white}{1}}*\frm{-} ="c";
\ar @{->} "a"; "b"
\ar @{->} "c"; "a"
\endxy}
  \caption{The framed quivers for the five generators $L_{i}$ of the Argyres-Douglas theory.}
  \label{fig:adframed1}
\end{figure}

The five framed quivers are shown in Figure \ref{fig:adframed1}. All of these quivers are in the mutation class of the $A_{3}$ Dynkin quiver, and their representation theory is well known \cite{MR0332887}.  The framed BPS states in the cocyclic chamber are stated explicitly in \cite{Alim:2011kw} yielding the following generating functionals
\begin{eqnarray}
F(L_{1},\zeta) & = & X_{\gamma_{1}}, \nonumber \\
F(L_{2},\zeta) & = & X_{\gamma_{2}}+X_{\gamma_{1}+\gamma_{2}}, \nonumber \\
F(L_{3},\zeta) & = & X_{-\gamma_{1}}+X_{\gamma_{2}-\gamma_{1}}+X_{\gamma_{2}},  \label{adgensres} \\
F(L_{4},\zeta) & = & X_{-\gamma_{1}-\gamma_{2}}+X_{-\gamma_{1}}, \nonumber \\
F(L_{5},\zeta) & = & X_{-\gamma_{2}}. \nonumber
\end{eqnarray}
Notice that both $L_{1}$ and $L_{5}$ support only a single framed BPS state in this chamber consistent with the fact that these defects lie in the dual cone to the seed $\{\gamma_{1}, \gamma_{2}\}.$  The remaining $L_{i}$ support a non-trivial BPS spectrum consisting of spin zero states.

From the result $\eqref{adgensres}$ and the multiplication rules for the symbols $X_{\gamma}$ we may extract the non-commutative OPE coefficients .  Let us view the index $i$ as periodic mod $5$.  Then we find
\begin{equation}
L_{i}\ast L_{i+2}=1+yL_{i+1}, \hspace{.5in} L_{i}\ast L_{i-2}=1+y^{-1}L_{i-1},
\end{equation}
where in the above, the symbol $1$ denotes the line defect with vanishing core charge.  In the commutative case $y=1$ the above reproduces answers found in \cite{Gaiotto:2010be}.  Thus, we have recovered a non-trivial property of the Argyres-Douglas CFT using a purely infrared calculation in quiver quantum mechanics.

\paragraph{Example of the OPE}\mbox{} \\

As a further application of our general results, we now illustrate the OPE in an example of the general recursion formula \eqref{quiverrecursion}.  Recall from \S \ref{OPEsec}, that the operator product expansion of line defects is a chamber independent property of a quantum field theory.  It is interesting to see this in practice for explicit line defects in a chamber where the framed BPS spectrum is non-trivial.

 Let $p_{1}, p_{2}, q_{1}, q_{2}$ be non-negative integers.  We consider a defect whose core charge is $p_{i}\gamma_{1}+q_{i}\gamma_{2}$.  From the analysis of \S \ref{adconesanalysis} we know that such defects lie in a dual cone and hence must satisfy the algebra
\begin{equation}
L_{p_{1}\gamma_{1}+q_{1}\gamma_{2}}\ast L_{p_{2}\gamma_{1}+q_{2}\gamma_{2}} = \left(y^{p_{1}q_{2}-p_{2}q_{1}}\right)L_{(p_{1}+p_{2})\gamma_{1}+(q_{1}+q_{2})\gamma_{2}}. \label{adopetest}
\end{equation}
To exhibit this multiplication rule in an interesting setting we evaluate the framed BPS states directly for such defects in a chamber where the spectrum contains states of non-vanishing spin.

To carry out the calculation, we work in chamber which is left-adapted in the sense of \S\ref{adapted} so that
\begin{equation}
\arg(\zeta)>\arg(Z_{i}), \hspace{.5in}\forall i.
\end{equation}
Stable representations are then cyclic with respect to the framing subspace.   We consider a core and halo charge of the form
\begin{equation}
\gamma_{c}=p\gamma_{1}+q\gamma_{2}, \hspace{.5in}\gamma_{h}=n\gamma_{1}+m\gamma_{2}. \label{corehaloadex}
\end{equation}
The framed representations appear as
\begin{equation}
\xy
(-10,0)+*{\mathbb{C}^{n}}="a"; (10,0)+*{\mathbb{C}^{m}}="b"; (0,15)*+{\mathbb{C}} ="c";   (-6,4.5)+*{q};  (6,4.5)+*{p};
\ar @{->} "a"; "b"
\ar @{->} "a"; "c"
\ar @{->} "c"; "b"
\endxy
\label{adrepexample}
\end{equation}
where now $p$ and $q$ indicate the number of linear maps between the vector spaces in the diagram above.  Our aim is to enforce the cyclic property on this representation and extract the moduli space.

We first observe that if $n$ is non-zero then \eqref{adrepexample} admits a proper subrepresentation of the form shown in \eqref{destablize1}.
\begin{equation}
\xy
(-10,0)+*{\mathbb{C}^{n}}="a"; (-20,0)+*{0}="d"; (10,0)+*{\mathbb{C}^{m}}="b"; (20,0)+*{\mathbb{C}^{m}}="e"; (0,15)*+{\mathbb{C}} ="c";  (0,30)*+{\mathbb{C}} ="f";   (-6,4.5)+*{q};  (6,4.5)+*{p};  (12.5,10)+*{p}; (1.5,11)+*{I_{1}}; (15,1.5)+*{I_{m}};
\ar @{->} "a"; "b"
\ar @{->} "a"; "c"
\ar @{->} "c"; "b"
\ar @{->} "d"; "a"
\ar @{->} "e"; "b"
\ar @{->} "f"; "c"
\ar @/_4pc/  "d"; "e"
\ar @/^1pc/  "d"; "f"
\ar @/^1pc/  "f"; "e"
\endxy
\label{destablize1}
\end{equation}
Where in the above, $I_{k}$ indicates the $k$-dimensional identity map.  From this we conclude that if $n$ is non-zero the representation cannot be cyclic.  Hence $n$ equals zero, and we are reduced to studying representations of the form
\begin{equation}
\xy
(-9,0)+*{\mathbb{C}}="a"; (9,0)+*{\mathbb{C}^{m}.}="b";  (0,7.7)+*{A_{1}}; (0,3.5)+*{A_{2}};  (0,.5)+*{\vdots}; (0,-4)+*{A_{p}};
\ar @/^3pc/ "a"; "b"
\ar @/^1pc/ "a"; "b"
\ar @/_1pc/ "a"; "b"
\endxy
\end{equation}
Possible destabilizing subrepresentations are diagrams
\begin{equation}
\xy
(-9,0)+*{\mathbb{C}}="a"; (9,0)+*{\mathbb{C}^{m}}="b"; (0,3.5)+*{A_{1}};  (0,.5)+*{\vdots}; (0,-4)+*{A_{p}};
(-9,-13)+*{\mathbb{C}}="c"; (9,-13)+*{\mathbb{C}^{k}.}="d";  (0,-9.5)+*{a_{1}};  (0,-12.5)+*{\vdots}; (0,-17)+*{a_{p}};
\ar @/^1pc/ "a"; "b"
\ar @/_1pc/ "a"; "b"
\ar @/^1pc/ "c"; "d"
\ar @/_1pc/ "c"; "d"
\ar @{->} "c"; "a"
\ar @{->} "d"; "b"
\endxy
\end{equation}
where $k<m$, the vertical maps are injective and $a_{i}$ is the restriction of $A_{i}$ to the given subspace.  Such subrepresentations exist if and only if the span of the maps $A_{i}$ in $\mathbb{C}^{m}$, $\bigoplus \mathrm{Im}(A_{i})$, lies is a proper subspace.  Thus the representation is cyclic and hence stable if and only if $\bigoplus \mathrm{Im}(A_{i})$ spans $\mathbb{C}^{m}.$

We may now pass to the desired moduli space.  Assemble the linear maps $A_{i}$ into an $m\times p$ matrix.  According to the cyclic stability condition this matrix has maximal rank $m.$   Passing to the moduli space means quotienting this data by $Gl(m,\mathbb{C})$ acting by matrix multiplication.  Thus we conclude that the moduli space is a Grassmannian of $m$-planes in $\mathbb{C}^{p}$.

In summary the moduli space of framed BPS states for charges of the form \eqref{corehaloadex} is
\begin{equation}
\mathcal{M}_{\gamma_{c}=p\gamma_{1}q\gamma_{2}, \gamma_{h}=n\gamma_{1}+m\gamma_{2}}=\begin{cases}\mathrm{Empty} & \mathrm{if}  \  n>0 \ \mathrm{or} \  m>p, \\
Gr(m,p) & \mathrm{if}  \ n=0 \ \mathrm{and} \ m\leq p.
\end{cases} \label{grassmanian}
\end{equation}

From the result \eqref{grassmanian} we may extract the generating functional of framed BPS states in the cyclic chamber.   The Grassmannian has non-trivial cohomology $h^{p,q}$ only along the diagonal of the Hodge diamond where $p=q$, and its Poincar\'{e} polynomial takes the simple form
\begin{equation}
P_{Gr(m,p)}(t)=\sum_{s}b_{s}(Gr(m,p))t^{s}=\frac{\prod_{i=1}^{p}\left(1-t^{2i}\right)}{\prod_{i=1}^{m}\left(1-t^{2i}\right)\prod_{i=1}^{p-m}\left(1-t^{2i}\right)}
\end{equation}
It follows that the generating functional of framed BPS states is
\begin{equation}
F(L_{p\gamma_{1}+q\gamma_{2}}, \zeta)=\sum_{m=0}^{p}\left(\frac{y^{m(m-p)}\prod_{i=1}^{p}\left(1-y^{2i}\right)}{\prod_{i=1}^{m}\left(1-y^{2i}\right)\prod_{i=1}^{p-m}\left(1-y^{2i}\right)}\right)X_{p\gamma_{1}+(q+m)\gamma_{2}}. \label{adhighspin}
\end{equation}
The result \eqref{adhighspin} illustrates the no-exotics phenomenon of \S \ref{framedhilbert}: the coefficient of each monomial $X_{\gamma}$ is a Laurent polynomial in $y$ with positive integral coefficients.

The  generators $L_{\gamma_{1}}$ and $L_{\gamma_{2}}$ have the simple spectra
\begin{equation}
F(L_{\gamma_{1}}, \zeta)=X_{\gamma_{1}}+X_{\gamma_{1}+\gamma_{2}}, \hspace{.5in}F(L_{\gamma_{2}}, \zeta)=X_{\gamma_{2}}.
\end{equation}
Using these expressions, one may readily verify the multiplication formula
\begin{eqnarray}
F(L_{\gamma_{1}}, \zeta) \ast F(L_{p\gamma_{1}+q\gamma_{2}}, \zeta)& = & y^{q}F(L_{(p+1)\gamma_{1}+q\gamma_{2}}, \zeta), \\
F(L_{\gamma_{2}}, \zeta) \ast F(L_{p\gamma_{1}+q\gamma_{2}}, \zeta) &= &y^{-p}F(L_{p\gamma_{1}+(q+1)\gamma_{2}}, \zeta), \nonumber
\end{eqnarray}
thus proving the result \eqref{adopetest} by induction.

\subsubsection{$k=2: SU(2)$ SYM}
\label{wilsonlineex}

Next we study the framed spectra in the pure $SU(2)$ theory.  For line defects in dual cones, the situation is similar to that studied in the Argyres-Douglas theory.  In accordance with the general analysis of \S \ref{dualseeds}, there is a chamber where each such defect supports a single framed state and the spectrum in all remaining chambers follows via application of the wall-crossing formula.

By contrast, the line defects lying along the limiting ray of Figure \ref{fig:su2cones} do not occupy any dual cone.  We denote these defects by $W_{n},$ with $n\geq0.$  Physically, they describe the WIlson lines in the $n+1$ dimensional irreducible representation of $SU(2)$.  A correct calculation of the framed BPS states associated to these defects requires us to use our multi-centered formula \eqref{framedquivconj}, and hence provides an explicit example where the framed quiver construction of \S \ref{framedquivers} produces an incorrect answer for the spectrum.

To illustrate these phenomena in detail, we provide an independent calculation of the spectrum from the two branches.  First, a Higgs branch calculation, is carried out using representation theory of the framed quiver.  The resulting spectrum \eqref{wilsonhiggsspec} defines a candidate defect $W_{n}^{\mathrm{Higgs}}$.  We verify that the commutative OPE of these objects is
\begin{equation}
W_{n}^{\mathrm{Higgs}} W_{m}^{\mathrm{Higgs}}=W_{n+m}^{\mathrm{Higgs}}.  \label{hwilsonOPE}
\end{equation}
In the second calculation we explain how to geometrically extract the multi-centered Coulomb branch states appearing our general proposal \eqref{framedquivconj}.   The resulting spectrum \eqref{explicit1} defines the defect $W_{n}$.  We demonstrate that they satisfy the non-commutative OPE
\begin{equation}
W_{n}\ast W_{m}=\sum_{k=0}^{\mathrm{min}(n,m)}W_{|n-m|+2k}.  \label{cwilsonOPE}
\end{equation}

Of these two results \eqref{hwilsonOPE} and \eqref{cwilsonOPE}, it is the latter which is the correct multiplication rule.  Indeed, from the physical description of the defect $W_{n}$ as the Wilson line in the $n+1$ dimensional representation of $SU(2),$ we know two properties of these defects.
\begin{itemize}
\item If framed states carrying non-trivial magnetic charge are ignored, then the charges of the remaining purely electric states must comprise exactly the weight decomposition of the $n+1.$  This follows because on the Coulomb branch, $SU(2)$ is broken to its maximal torus $U(1)$.
\item The $SU(2)$ theory is asymptotically free.  Hence from a UV calculation of the OPE, it is clear that the defects $W_{n}$ must obey the algebra of $SU(2)$ representations.  This is true even though the framed BPS spectrum of these defects in general contains states with both non-trivial electric and magnetic charges.
\end{itemize}
The defects $W_{n}$ defined from our multi-centered prescription \eqref{framedquivconj} satisfy these properties, while the defects $W_{n}^{\mathrm{Higgs}}$ defined by framed quiver representations fail on both accounts.  These results are the most significant piece of evidence for our conjecture presented in this paper.

\paragraph{Higgs Branch Calculation}\mbox{} \\

We begin with the Higgs branch calculation of the framed spectra of the defect $W^{\mathrm{Higgs}}_{n}$.  We are to analyze representations of the quiver shown below.
\begin{equation}
\xy
(-10,0)+*{\bigcirc}="a"; (10,0)+*{\bigcirc}="b";   (0,15)*+{\color{white}{1}}*\frm{-} ="c";  (0,1.5)+*{2};  (-3.5,3.5)+*{n};  (3.5,3.5)+*{n};  (0,-2.5)+*{A_{i}};  (7.5,4.5)+*{B_{j}}; (-7.9,4.5)+*{C_{k}};
\ar @{->} "a"; "b"
\ar @{->} "b"; "c"
\ar @{->} "c"; "a"
\endxy
\label{framedwilson}
\end{equation}
The symbols $A_{1}, A_{2},$ $B_{1}, \cdots, B_{n},$ and $C_{1}, \cdots C_{n}$ denote the indicated arrows.  We work in region of stability space where
\begin{equation}
\arg(\zeta)<\arg(Z_{i}), \hspace{.5in}\forall i,
\end{equation}
so that the quiver is right-adapted in the sense of \S \ref{adapted}, and stability coincides with the algebraic cocyclic property.  This in turn is equivalent to the following injectivity conditions on the maps
\begin{equation}
\bigcap_{i=1}^{2}\mathrm{Ker}(A_{i})=0, \hspace{.5in}\bigcap_{i=1}^{n}\mathrm{Ker}(B_{i})=0. \label{su2stabres}
\end{equation}
Finally, the quiver \eqref{framedwilson} has non-trivial cycles and requires a potential $\mathcal{W}$.  We take this potential to be generic and cubic.  Without loss of generality we thus have
\begin{equation}
\mathcal{W}=\sum_{i=1}^{n}C_{i}\circ B_{i} \circ A_{1}-\sum_{i=1}^{n}\lambda_{i}C_{i}\circ B_{i}\circ A_{2},
\end{equation}
where in the above $\lambda_{i}$ are distinct non-zero complex numbers.  Thus, we impose the following relations on representations
\begin{eqnarray}
\frac{\partial\mathcal{W}}{\partial{A_{1}}}=0 & \Longrightarrow & \sum_{i=1}^{n}C_{i}\circ B_{i}=0, \label{abrel} \\
\frac{\partial\mathcal{W}}{\partial{A_{2}}}=0 & \Longrightarrow & \sum_{i=1}^{n}\lambda_{i}C_{i}\circ B_{i}=0, \\
\frac{\partial\mathcal{W}}{\partial{C_{i}}}=0 & \Longrightarrow & B_{i}\circ (A_{1}-\lambda_{i}A_{2}) =0,\label{barel}\\
\frac{\partial\mathcal{W}}{\partial{B_{i}}}=0 & \Longrightarrow &  (A_{1}-\lambda_{i}A_{2})\circ C_{i} =0. \label{carel}
\end{eqnarray}

Consider a representation with halo charge $\gamma_{h}=r{\gamma_{1}}+s\gamma_{2}$ shown below.
\begin{equation}
\xy
(-10,0)+*{\mathbb{C}^{r}}="a"; (10,0)+*{\mathbb{C}^{s}}="b";   (0,15)*+{\mathbb{C}} ="c";   (0,-2.5)+*{A_{i}};  (7.5,4.5)+*{B_{j}}; (-7.9,4.5)+*{C_{k}};
\ar @{->} "a"; "b"
\ar @{->} "b"; "c"
\ar @{->} "c"; "a"
\endxy
\end{equation}
We refer to the moduli space of stable, flat representations of the form shown above as $\mathcal{M}^{n}_{r,s}$.  We constrain the representations with the following results.

\begin{prop} \label{Ainj} The maps $A_{1}$ and $A_{2}$ are both injective.
\end{prop}
\begin{proof} Suppose on the contrary that $A_{1}$ has a kernel containing the non-zero vector $v$.  According to the stability condition \eqref{su2stabres}, $v$ cannot be in the kernel of $A_{2}$.  By the relation \eqref{barel}, we deduce that $B_{i}\circ A_{2}(v)=0$ for all $i,$ thus violating the constraint \eqref{su2stabres}.  By symmetry, $A_{2}$ must also be injective.
\end{proof}

\begin{prop} \label{Czero} The maps $C_{i}$ are zero for all $i$.
\end{prop}
\begin{proof}
Suppose that $C_{1}$ is nonzero.  Let $v$ denote a basis vector at the framing space.  By assumption, $w=C_{1}(v)$ is nonzero.  According to the relation \eqref{carel}, we have
\begin{equation}
A_{1}(w)=\lambda_{1}A_{2}(w).
\end{equation}
Applying \eqref{barel} we obtain
\begin{equation}
(\lambda_{1}-\lambda_{i})B_{i}\circ A_{2}(w)=0, \hspace{.5in} \forall i.
\end{equation}
By Proposition \ref{Ainj}, $A_{2}$ is injective; hence the above implies that
\begin{equation}
A_{2}(w)\in \bigcap_{i=2}^{b}\mathrm{Ker}(B_{i}),
\end{equation}
and from the stability condition \eqref{su2stabres} we deduce that $B_{1}\circ A_{2}(w) \neq0$.  Since the framing space is one-dimensional we must have $B_{1}\circ A_{2}(w)=\alpha v,$ with $\alpha$ a nonzero constant.  Finally, apply \eqref{abrel} to $A_{2}(w)$ to obtain
\begin{equation}
0= \sum_{i=1}^{n}C_{i}\circ B_{i}\circ A_{2}(w)=C_{1}\circ B_{1} \circ A_{2}(w) =\alpha C_{1}(v),
\end{equation}
thus contradicting the assumption that $C_{1}$ is non-vanishing.  By symmetry, $C_{i}$ therefore vanishes for all $i$.
\end{proof}

Via the above propositions, the quiver representations under investigation are reduced to the form of \eqref{simplifiesfsu2}.
\begin{equation}
\xy
(0,0)+*{\mathbb{C}^{r}}="a"; (10,0)+*{\mathbb{C}^{s}}="b";  (20,0)+*{\mathbb{C}} ="c";   (5,-2.5)+*{A_{i}};  (15,-2.5)+*{B_{j}};
\ar @{->} "a"; "b"
\ar @{->} "b"; "c"
\endxy
\label{simplifiesfsu2}
\end{equation}
Moreover as a consequence of the injectivity constraints, the dimension vectors are restricted to the range
\begin{equation}
0\leq r\leq s\leq n.
\end{equation}
To proceed further, we make use of the well-known classification \cite{MR1476671}  indecomposable representations of the quiver $Q_{2}.$\footnote{We thank Hugh Thomas for helpful comments on an earlier version of the following argument.}

\begin{prop} \label{kroneckermodule} The indecomposable representations of the quiver $Q_{2}$ where the maps $A_{i}$ are injective consist of the following.
\begin{itemize}
\item Representations $V_{\alpha}$ of dimension vector $(1,1),$
\begin{equation}
\xy
(0,0)+*{\mathbb{C}}="a"; (10,0)+*{\mathbb{C}}="b";    (5,-2.5)+*{A_{i}};
\ar @{->} "a"; "b"
\endxy
\end{equation}
Such representations are classified (up to isomorphism) by a point $\alpha \in \mathbb{C}^{\ast},$ where $A_{1}=\alpha$, and $A_{2}=1$.
\item Representations $J_{k}$ of dimension vector $(k,k)$ with $k>1$,
\begin{equation}
\xy
(0,0)+*{\mathbb{C}^{k}}="a"; (10,0)+*{\mathbb{C}^{k}}="b";    (5,-2.5)+*{A_{i}};
\ar @{->} "a"; "b"
\endxy
\end{equation}
Such representations are classified (up to isomorphism) by a point $\alpha \in \mathbb{C}^{\ast},$ where $A_{2}=I_{k}$, and
\begin{equation}
A_{1}=\left(\begin{array}{cccccc}
\alpha & 1 & 0 & \cdots &0 & 0\\
0 & \alpha & 1  & \cdots  &0 & 0\\
\vdots & \vdots & \vdots & \ddots  & \vdots & \vdots \\
0 & 0 & 0 & \cdots & \alpha & 1\\
0 & 0 & 0 & \cdots & 0 & \alpha
\end{array}
 \right),
\end{equation}
is a $k\times k$ Jordan block matrix with eigenvalue $\alpha$.
\item Representations $D_{k}$ of dimension vector $(k-1,k),$
\begin{equation}
\xy
(0,0)+*{\mathbb{C}^{k-1}}="a"; (10,0)+*{\mathbb{C}^{k}}="b";    (5,-2.5)+*{A_{i}};
\ar @{->} "a"; "b"
\endxy
\end{equation}
Such representations are rigid and labelled by $k>0$.  Up to isomorphism, we may assume that the maps $A_{i}$ take the form
\begin{equation}
A_{1}=\left(\begin{array}{c}I_{k-1}\\0\end{array}\right), \hspace{.5in}A_{2}=\left(\begin{array}{c}0\\I_{k-1}\end{array}\right), \label{aspecial}
\end{equation}
where in the above $I_{k}$ denotes the $k\times k$ identity matrix.
\end{itemize}
\end{prop}
Returning to the framed representations \eqref{simplifiesfsu2}, we apply Proposition \ref{kroneckermodule}, to decompose the portion of the representation involving the unframed nodes.  Let $R$ denote this restricted representation.

\begin{prop} \label{jordandestroy} The representations $J_{k}$ are excluded from appearing in $R$ by the relations following from the potential.
\end{prop}
\begin{proof}
Consider the relation
\begin{equation}
B_{i}\circ (A_{1}-\lambda_{i}A_{2})=0,
\end{equation}
applied to an eigenvector $v$ of the Jordan matrix $A_{1}$.  We obtain
\begin{equation}
(\alpha-\lambda_{i})B_{i}\circ A_{2}(v)=0.
\end{equation}
If the eigenvalue $\alpha$ is not equal to $\lambda_{i}$ for all $i$, then all the maps $B_{i}$ annihilate the non-zero vector $A_{2}(v)$ which violates the stability condition \eqref{su2stabres}.  Therefore, there exists an $i$ such that $\alpha =\lambda_{i},$ and $B_{i}\circ A_{2}(v)\neq 0$

Now, as $\lambda_{i}$ is equal to the eigenvalue for the Jordan matrix $A_{1},$ the map $A_{1}-\lambda_{i}A_{2}$ is a Jordan matrix with eigenvalue zero.  We may then pick a non-zero vector $w$ such that
\begin{equation}
\left(A_{1}-\lambda_{i}A_{2}\right)(w)=A_{2}(v).
\end{equation}
But then applying $B_{i}$ to the above equation we obtain
\begin{equation}
0=B_{i}\circ \left(A_{1}-\lambda_{i}A_{2}\right)(w)=B_{i}\circ A_{2}(v),
\end{equation}
which is a contradiction.
\end{proof}

As a consequence of Proposition \ref{jordandestroy}, the restricted representation $R$ can be specified by an integer, $\ell,$ with $0\leq \ell \leq r$, and a partition, $k,$ of $s-\ell$ into $s-r$ positive parts.  With $k=(k_{1}, k_{2}, \cdots, k_{s-r}).$ The decomposition is then
\begin{equation}
R=D_{k_{1}}\oplus D_{k_{2}}\oplus \cdots \oplus D_{k_{s-r}}\oplus V_{\alpha_{1}}\oplus V_{\alpha_{2}}\oplus \cdots \oplus V_{\alpha_{\ell}}. \label{decomprest}
\end{equation}

In the following, it will also be important to know the automorphism group of the above representation.  Each summand has as automorphim group $\mathbb {C}^{*}$ generated by simultaneous transformations proportional to the identity matrix at each node.  If each of the summands appearing in \eqref{decomprest} is distinct, then the automorphism group of such an $R$ is simply the product of the automorphims of each summand leading to $\left(\mathbb{C}^{*}\right)^{s-r+\ell}$.  More generally, suppose that the decomposition \eqref{decomprest} contains $a$ non-isomorphic representations each appearing with multiplicity $m_{a}$.  Then the automorphism group is enhanced to
\begin{equation}
\prod_{i=1}^{a}Gl(m_{a},\mathbb{C}), \label{isos}
\end{equation}
where in the above each factor acts non-trivially on the subspace spanned by a set of isomorphic representations.

We now apply the relations and stability conditions to determine how the remaining maps $B_{j}$ behave on the various subspaces in $R.$

\begin{prop} \label{lambdaalpha} For each $j$ there exists an $i$ such that $\alpha_{j}=\lambda_{i}$.
\end{prop}
\begin{proof}
Proceed as in Proposition \ref{jordandestroy}.
\end{proof}

Note that as a corollary we also see that the maps $B_{i}$ vanishes on $V_{\alpha_{j}}$ unless $\lambda_{i}=\alpha_{j}$.  Further, as each $B_{i}$ has rank one, we also learn that the $\alpha_{j}$ are distinct.

The above analysis solves the constraints implied by the potential on the subspace of \eqref{decomprest} spanned by the $V_{\alpha_{a}}$.  We now solve the relations on the space spanned by each $D_{k_{q}}$ appearing in $R$.  To describe the result it is useful to pick explicit bases for each of the vector spaces in the representation $D_{k_{q}}.$  Thus let $e^{q}_{1}, \cdots, e^{q}_{k_{q}-1}$ be a basis for the domain of $D_{k_{q}},$ and let $f^{q}_{1}, \cdots, f^{q}_{k_{q}}$ be a basis for the range.   According to Proposition \ref{kroneckermodule} we may assume
\begin{equation}
A_{1}(e^{q}_{j})=f^{q}_{j}, \hspace{.5in}A_{2}(e^{q}_{j})=f^{q}_{j+1}. \label{explicitbasisef}
\end{equation}

\begin{prop} \label{recursive}
Applied to the representation $D_{k_{q}},$ the map $B_{i}$ is completely determined by its value, $b_{i}^{q}$, on the vector $f_{k_{q}}^{q}$.
\end{prop}
\begin{proof}
Again we examine the equation
\begin{equation}
B_{i}\circ (A_{1}-\lambda_{i}A_{2})=0.
\end{equation}
Applied to a vector $e^{q}_{j}$ and using \eqref{explicitbasisef} we find
\begin{equation}
B_{i}(f^{q}_{j})=\lambda_{i}B_{i}(f^{q}_{j+1}),
\end{equation}
which has as general solution
\begin{equation}
B_{i}(f_{j}^{q})=b_{i}^{q}\lambda_{i}^{k_{q}-j}.
\end{equation}
\end{proof}

Propositions \ref{lambdaalpha} and \ref{recursive} along with the decomposition \eqref{decomprest} completely solve the relations implied by the potential as well as the stability condition which constrains both $A_{1}$ and $A_{2}$ to be injective.  The remaining restriction to study is that $\cap_{i}\mathrm{Ker}(B_{i})=0.$

Consider the decomposition \eqref{decomprest}.  According to the argument given for Proposition \ref{lambdaalpha} for each $V_{\alpha_{j}}$ there exists a unique $B_{i}$  which is nonzero on $ V_{\alpha_{j}}.$  Further, we may use the $\mathbb{C}^{*}$ isomorphism factor to scale the value of $B_{i}|_{V_{\alpha_{j}}}$ to unity.  As we must make a choice of assignment $j \rightarrow i$, we obtain $\binom{n}{\ell}$ distinct configurations of the $B$'s restricted to the direct sum of $V$'s.

Next examine the partition $k=(k_{1}, \cdots, k_{s-r})$ appearing in \eqref{decomprest}. Suppose that the partition contains $c$ distinct factors with multiplicities $m_{1}, \cdots, m_{c}$.  We consider the maps $B_{i}$ restricted to the subspace $D_{k_{1}}\oplus D_{k_{2}}\oplus \cdots \oplus D_{k_{m}}$ where $k_{1}=k_{2}=\cdots k_{m}$ and no other $k_{i}$ appearing in the partition $k$ are equal to $k_{1}.$ Via Proposition \ref{recursive} we know that the map $B_{i}$ on this subspace is determined by the numbers $b_{i}^{q}$ where $q$ ranges from $1$ to $m,$ the multiplicity in question.  A necessary condition for stability is that the $m \times n$ matrix with entries $b_{i}^{q}$ must have maximal rank $m$.  As described in \eqref{isos}, the automorphism group of $D_{k_{1}}\oplus D_{k_{2}}\oplus \cdots \oplus D_{k_{m}}$ is $Gl(m,\mathbb{C})$ and it acts on the matrix $b_{i}^{q}$ via standard matrix multiplication.  The quotient space of this action is the Grassmannian $Gr(m,n).$  Applying this result to the
entire partition, we have proved the following:
\begin{prop}\label{thmhiggs}
The set of stable representations whose restriction to the Kronecker subquiver appears as  in \eqref{decomprest} is naturally embedded inside the space
\begin{equation}
X_{k,\ell}=Gr(m_{1},n)\times Gr(m_{2},n) \times \cdots \times Gr(m_{c},n)\times \{p_{1}, p_{2}, \cdots , p_{\binom{n}{\ell}}\}, \label{prodgrass}
\end{equation}
where in the above $p_{i}$ is a point, and the partition $k$ contains $c$ distinct elements each with multiplicities $m_{1}, \cdots, m_{c}.$  Each point in $X_{k,\ell}$ solves the relations imposed by the potential, solves the injectivity constraint on the maps $A_{i}$ and has
\begin{equation}
\left(\bigcap_{i=1}^{n}\mathrm{Ker}(B_{i})\right)\bigcap\left(V_{\alpha_{1}}\oplus \cdots \oplus V_{\alpha_{\ell}}\right)=0, \hspace{.5in}\left(\bigcap_{i=1}^{n}\mathrm{Ker}(B_{i})\right)\bigcap\left(\langle f_{k_{d}}^{1},f_{k_{d}}^{2}, \cdots,f_{k_{d}}^{m_{d}} \rangle\right)=0.
\end{equation}
\end{prop}

The above result does not completely specify the moduli space of stable quiver representations. However combined with a localization argument it does provide us with enough information to extract the cohomology.

Consider the scaling which acts on the maps $B_{i}$ as $B_{i}\rightarrow \beta_{i}B_{i}$ where the $n$-tuple $(\beta_{1},\cdots , \beta_{n})$ lies in the algebraic torus $\left(\mathbb{C}^{*}\right)^{n}.$  This action preserves the relations $B_{i}\circ(A_{1}-\lambda_{i}A_{2})=0,$ and hence acts on the moduli space.  We extract the Euler characteristic of the moduli space by summing over fixed points.

The summation over fixed points is facilitated by observing the the torus in question acts on $X_{k,\ell}$ preserving each factor and further that the action on each Grassmannian is standard.  Thus, consider an integer $k_{1}$ occurring in the partition with multiplicity $m$ leading to a $Gr(m,n)$ factor in \eqref{prodgrass}.  A fixed point in this factor is given by choosing an $m\times m$ minor of the matrix $b_{i}^{q}$ and setting this minor to be the identity matrix with all other entries vanishing.  For instance
\begin{equation}
\left(\begin{array}{ccccccc}
b_{1}^{1} & b_{2}^{1} & \cdots& b_{m}^{1} &  b_{m+1}^{1}&\cdots & b_{n}^{1} \\
b_{1}^{2} & b_{2}^{2} & \cdots& b_{m}^{2} & b_{m+1}^{1}& \cdots & b_{n}^{2} \\
\vdots & \vdots & \ddots & \vdots & \vdots &\ddots &  \vdots \\
b_{1}^{m} & b_{2}^{m} & \cdots& b_{m}^{m} & b_{m+1}^{1}& \cdots & b_{n}^{m} \\
\end{array}\right)=
\left(\begin{array}{ccccccc}
1& 0& \cdots&0&  0&\cdots & 0 \\
0 & 1 & \cdots&0 & 0& \cdots &0 \\
\vdots & \vdots & \ddots & \vdots & \vdots &\ddots &  \vdots \\
0&0 & \cdots& 1& 0& \cdots & 0 \\
\end{array}\right),
\label{fixedex}
\end{equation}
is a fixed point.

Suppose at the fixed point above that $k_{1}>1.$  Then we claim that the maps $B_{i}$ do not satisfy the stability condition \eqref{su2stabres}.  Indeed, consider the summand in question
\begin{equation}
D_{k_{1}}\oplus D_{k_{2}} \oplus \cdots \oplus D_{k_{m}},
\end{equation}
where $k_{1}=k_{2}=\cdots k_{m}$.  At the fixed point \eqref{fixedex} we have
\begin{equation}
\mathrm{Ker}(B_{i})\supseteq\begin{cases}
D_{k_{1}}\oplus D_{k_{2}} \oplus  \cdots \oplus D_{k_{i-1}}\oplus S_{i}\oplus D_{k_{i+1}} \oplus D_{k_{m}} & \mathrm{If} \ i\leq m,\\
D_{k_{1}}\oplus D_{k_{2}} \oplus \cdots \oplus D_{k_{m}} & \mathrm{If} \ i>m.
\end{cases} \label{kertoobig}
\end{equation}
In the above $S_{i}$ is a codimension one subspace of $D_{k_{i}}$ and has non-zero dimension if $k_{1}>1$.  From \eqref{kertoobig} we learn that
\begin{equation}
\bigcap_{i=1}^{n}\mathrm{Ker}(B_{i}) \supseteq S_{1} \oplus S_{2} \oplus \cdots \oplus S_{m} \neq 0,
\end{equation}
violating stability.  From this argument we deduce:

\begin{prop} \label{simplefixed}
Fixed points in the moduli space may only occur when all $k_{i}$ in the partition $k$ are one.  Hence, the decomposition \eqref{decomprest} simplifies to
\begin{equation}
R=\underbrace{D_{1}\oplus D_{1}\oplus \cdots \oplus D_{1}}_{s-r}\oplus V_{\alpha_{1}}\oplus V_{\alpha_{2}}\oplus \cdots \oplus V_{\alpha_{r}}. \label{decompsimp}
\end{equation}
\end{prop}

Finally, we may conclude the calculation by counting fixed points directly.  We first pick $r$ of the $B_{i}$ to evaluate nontrivially on the $V_{\alpha_{i}}$ in accordance with Proposition \ref{lambdaalpha}.  Next, we must pick $s-r$ of the $B_{i}$ and set the matrix $b_{i}^{q}$ to be of the form \eqref{fixedex} with the non-degenerate minor spanned by the chosen $B_{i}$.  It is easy to see that there can be no overlap between the two sets of $B_{i}'$s.  Indeed, as in Proposition \ref{simplefixed} if there is an overlap with say $B_{i}$ having non-zero restriction to both a $V_{\alpha_{j}}$ and a $D_{1}$ summand in \eqref{decompsimp}, then the kernel of $B_{i}$ is a one-dimensional subspace contained in $V_{\alpha_{j}} \oplus D_{1},$ on which all other $B_{j}'$s also vanish.   Summing over contributions thus yields the following result.

\begin{thm}\label{eulercalc}
The Euler characteristic of the moduli space of stable, flat quiver representations of is given by
\begin{equation}
\chi(\mathcal{M}_{r,s}^{n})=\binom{n}{r}\binom{n-r}{s-r}.
\end{equation}
\end{thm}

In fact, with only slightly more work we obtain a more detailed picture of the cohomology.

\begin{thm} \label{algebraic}
The cohomology of $\mathcal{M}^{n}_{r,s}$ is all supported along the middle of the Hodge decomposition.  In other words, $h^{p,q}(\mathcal{M}^{n}_{r,s})$ vanishes if $p$ is different from $q$.
\end{thm}
\begin{proof}
Each fixed point in the analysis above contributes a given cohomology class in the cohomology of $\mathcal{M}^{n}_{r,s}$.  The Hodge bidegree of this class is determined by examining the tangent space to $\mathcal{M}^{n}_{r,s}$ at the the fixed point.  This vector space is a representation of the localizing torus $\C^{n}$ and the weight decomposition of this representation with respect to a generic element of the torus determines the Hodge decomposition via a standard argument from Morse theory.

Now observe that the moduli space $\mathcal{M}^{n}_{r,s}$ has complex dimension
\begin{equation}
\mathrm{dim}\left(\mathcal{M}^{n}_{r,s}\right) \equiv d(n,r,s)=(s-r)(n-s+r), \label{dimform}
\end{equation}
and this dimension coincides with that of $Gr(s-r,n).$  From Propositions  \ref{thmhiggs} and \ref{simplefixed}, we know that the locus in $\mathcal{M}^{n}_{r,s}$
which contains fixed points splits into a number of connected components, each of which is embedded in  $Gr(s-r,n).$  Moreover the action of the localizing torus on $\mathcal{M}^{n}_{r,s}$ is the same as the standard torus action on the Grassmannian.

It follows from this analysis that the local structure of the torus fixed points in $\mathcal{M}^{n}_{r,s}$, in particular the weight decomposition of the tangent space, agrees with that of the same fixed point when it is viewed as a point of the Grassmannian.  Since the Grassmannian has cohomology only along the middle of the Hodge decomposition, we conclude that $\mathcal{M}^{n}_{r,s}$ does as well.
\end{proof}

From Theorems \ref{eulercalc} and \ref{algebraic} we now pass to an explicit formula for the framed degeneracies.  Via a specialization of \eqref{indexform} we have
\begin{equation}
\fro_{\mathrm{Higgs}}\left(W_{n},\left(r-\frac{n}{2}\right)\gamma_{1}+\left(s-\frac{n}{2}\right)\gamma_{2},y=1\right)=\binom{n}{r}\binom{n-r}{s-r}. \label{wilsonhiggsspec}
\end{equation}
We assemble these degeneracies into a commutative generating functional
\begin{equation}
F(W_{n}^{\mathrm{Higgs}})= \sum_{s=0}^{n}\sum_{r=0}^{s}\binom{n}{r}\binom{n-r}{s-r}X_{\left(r-\frac{n}{2}\right)\gamma_{1}+\left(s-\frac{n}{2}\right)\gamma_{2}}. \label{wilsonhiggsfunc}
\end{equation}
These formulae are the final results of the Higgs branch calculation.

From \eqref{wilsonhiggsspec} and \eqref{wilsonhiggsfunc} one may readily see that the Higgs branch has produced an unphysical spectrum of framed BPS states. For example, the case $r=s$ above describes the purely electric states.  On physical grounds these should have degeneracy one, but instead the Higgs branch has yielded many more states.  A related observation is that the generating functionals \eqref{wilsonhiggsfunc} obey the incorrect multiplication law
\begin{equation}
F(W_{n}^{\mathrm{Higgs}})F(W_{m}^{\mathrm{Higgs}})=F(W_{n+m}^{\mathrm{Higgs}}), \label{incorrecthiggsmult}
\end{equation}
as may be readily verified using Pascal's recursion relations on binomial coefficients.  We conclude that, in general, the representation theory of a framed quiver does not correctly compute the framed BPS states.

Though \eqref{incorrecthiggsmult} is not the correct algebra of $SU(2)$ Wilson lines, it is nevertheless curiously simple: it is the algebra of highest weights of irreducible $SU(2)$ representations.  Why this should be so is presently unclear.

\paragraph{Coulomb Branch Calculation}\mbox{} \\

Now we study the spectrum of framed BPS states predicted by our general proposal \eqref{framedquivconj}.  This requires us to use the Coulomb branch formula \eqref{MPS}.  We have not been able to evaluate the result in closed form.  Fortunately the paper \cite{Manschot:2013sya} includes a Mathematica notebook which implements the Coulomb branch formula.  Using this notebook we have studied what our proposal gives for the framed BPS states attached to the line defects $W_n$ for small $n$.\footnote{In a previous preprint version of this paper, we proposed an alternative method of calculating the result, and verified that it gave the expected results for $n \le 4$.  However, Hugh Thomas pointed out that beginning at $n=5$, this alternative method would give a result which would not obey the tensor product algebra of $SU(2)$ representations.  Fortunately, as we report here, implementing our general
proposal \eqref{framedquivconj} directly does give the expected result for $n=5$, at least for those coefficients we were able to check --- this includes some coefficients which differ from what our alternative method predicted.
We have not understood precisely what was wrong with our alternative method. We thank Hugh Thomas for many useful comments and suggestions on this subject.}
As usual, we package the results into a generating functional:
\begin{equation}
F(W_{n})=\sum_{s=0}^{n}\sum_{r=0}^{s}\fro\left(W_{n},\left(r-\frac{n}{2}\right)\gamma_{1}+\left(s-\frac{n}{2}\right)\gamma_{2},y\right)X_{\left(r-\frac{n}{2}\right)\gamma_{1}+\left(s-\frac{n}{2}\right)\gamma_{2}}. \label{wilsoncoulombfunc}
\end{equation}
We find
\begin{eqnarray} \label{mma-results}
F(W_{0}) & = & \left[\phantom{\int}\hspace{-.15in}X_{0}\right], \nonumber\\
F(W_{1}) & = & \left[\phantom{\int}\hspace{-.15in}X_{-\frac{1}{2}(\gamma_{1}+\gamma_{2})}+X_{\frac{1}{2}(\gamma_{1}+\gamma_{2})}\right]+X_{-\frac{1}{2}\gamma_{1}+\frac{1}{2}\gamma_{2}}, \nonumber\\
F(W_{2}) & = & \left[\phantom{\int}\hspace{-.15in}X_{-(\gamma_{1}+\gamma_{2})}+X_{0}+X_{(\gamma_{1}+\gamma_{2})}\right]+\left(y^{-1}+y\right)X_{-\gamma_{1}}+X_{-\gamma_{1}+\gamma_{2}}+\left(y^{-1}+y\right)X_{\gamma_{2}},  \nonumber\\
F(W_{3}) & = & \left[\phantom{\int}\hspace{-.15in}X_{-\frac{3}{2}(\gamma_{1}+\gamma_{2})}+X_{-\frac{1}{2}(\gamma_{1}+\gamma_{2})}+X_{\frac{1}{2}(\gamma_{1}+\gamma_{2})}+X_{\frac{3}{2}(\gamma_{1}+\gamma_{2})}\right]+\left(y^{-2}+1+y^{2}\right) X_{-\frac{3}{2}\gamma_{1}-\frac{1}{2}\gamma_{2}}\nonumber\\
& + & \phantom{\int}\hspace{-.15in} \left(y^{-2}+1+y^{2}\right) X_{-\frac{3}{2}\gamma_{1}+\frac{1}{2}\gamma_{2}}+\left(y^{-2}+2+y^{2}\right) X_{-\frac{1}{2}\gamma_{1}+\frac{1}{2}\gamma_{2}}+ X_{-\frac{3}{2}\gamma_{1}+\frac{3}{2}\gamma_{2}}   \nonumber\\
& + & \phantom{\int}\hspace{-.15in}  \left(y^{-2}+1+y^{2}\right) X_{-\frac{1}{2}\gamma_{1}+\frac{3}{2}\gamma_{2}}+ \left(y^{-2}+1+y^{2}\right) X_{\frac{1}{2}\gamma_{1}+\frac{3}{2}\gamma_{2}},        \nonumber \\
F(W_{4}) & = &  \left[\phantom{\int}\hspace{-.15in}X_{-2(\gamma_{1}+\gamma_{2})}+X_{-(\gamma_{1}+\gamma_{2})}+X_{0}+X_{(\gamma_{1}+\gamma_{2})}+X_{2(\gamma_{1}+\gamma_{2})}\right]  \label{explicit1} \\
&+&   \phantom{\int}\hspace{-.15in}\left(y^{-3}+y^{-1}+y+y^{3}\right) X_{-2\gamma_{1}-\gamma_{2}} +\left(y^{-4}+y^{-2}+2+y^{2}+y^{4}\right) X_{-2\gamma_{1}}    \nonumber \\
& + & \phantom{\int}\hspace{-.15in} \left(y^{-3}+2y^{-1}+2y+y^{3}\right) X_{-\gamma_{1}} +\left(y^{-3}+y^{-1}+y+y^{3}\right) X_{-2\gamma_{1}+\gamma_{2}}  \nonumber \\
& + &  \phantom{\int}\hspace{-.15in}  \left(y^{-4}+2y^{-2}+3+2y^{2}+y^{4}\right) X_{-\gamma_{1}+\gamma_{2}} + \left(y^{-3}+2y^{-1}+2y+y^{3}\right) X_{\gamma_{2}}+X_{-2\gamma_{1}+2\gamma_{2}}  \nonumber \\
& + & \phantom{\int}\hspace{-.15in}    \left(y^{-3}+y^{-1}+y+y^{3}\right) X_{-\gamma_{1}+2\gamma_{2}}+  \left(y^{-4}+y^{-2}+2+y^{2}+y^{4}\right) X_{2\gamma_{2}}         \nonumber \\
& + &  \phantom{\int}\hspace{-.15in}    \left(y^{-3}+y^{-1}+y+y^{3}\right) X_{\gamma_{1}+2\gamma_{2}}. \nonumber
\end{eqnarray}
By direct evaluation from \eqref{mma-results} we then find that the generating functionals $F(W_{n})$ for
small $n$ indeed satisfy the tensor product algebra of $SU(2)$ representations:
\begin{eqnarray}
F(W_{1})*F(W_{1}) & = & F(W_{0})+F(W_{2}), \\
F(W_{1})*F(W_{2}) & = & F(W_{1})+F(W_{3}),  \nonumber  \\
F(W_{1})*F(W_{3}) & = & F(W_{2})+F(W_{4}). \nonumber
\end{eqnarray}
This verifies \eqref{cwilsonOPE} at least in these cases.   We have also evaluated all of the terms in $F(W_5)$ except for the coefficients of $X_{\frac32 \gamma_1 + \frac52 \gamma_2}$ and $X_{\frac52 \gamma_1 + \frac52 \gamma_2}$ (the running time required to compute these using the notebook of \cite{Manschot:2013sya} on the PC we used appears to be greater than two weeks.)  All terms which we did compute are consistent with the expected relation
\begin{eqnarray}
F(W_1)*F(W_4) &=& F(W_3) + F(W_5).
\end{eqnarray}

It would be very interesting to know whether the $F(W_n)$ produced by the Coulomb branch formula indeed obey \eqref{cwilsonOPE} in general.

\section*{Acknowledgements}
We thank S. Cecotti, F. Denef, A. Goncharov, J. Manschot, B. Pioline, and A. Sen for discussions.  We thank H. Thomas for discussions and for pointing out a problem in the first preprint version of this paper.  We thank the Simons Center for Geometry and Physics for hospitality during the 2012 Simons Workshop where this work was initiated. The work of C.C. is supported by a Junior Fellowship at the Harvard Society of Fellows.  The work of A.N. is supported by  National Science Foundation grant DMS-1151693.

\bibliography{FramedQuivers}

\providecommand{\href}[2]{#2}\begingroup\raggedright\begin{thebibliography}{10}

\bibitem{Wilson:1974sk}
K.~G. Wilson, ``{Confinement of Quarks},'' {\em Phys.Rev.} {\bf D10} (1974)
2445--2459.

\bibitem{'tHooft:1977hy}
G.~'t~Hooft, ``{On the Phase Transition Towards Permanent Quark Confinement},''
  {\em Nucl.Phys.} {\bf B138} (1978)
1.

\bibitem{Kapustin:2005py}
A.~Kapustin, ``{Wilson-'t Hooft operators in four-dimensional gauge theories
  and S-duality},'' {\em Phys.Rev.} {\bf D74} (2006) 025005,
\href{http://www.arXiv.org/abs/hep-th/0501015}{{\tt hep-th/0501015}}.

\bibitem{Drukker:2010jp}
N.~Drukker, D.~Gaiotto, and J.~Gomis, ``{The Virtue of Defects in 4D Gauge
  Theories and 2D CFTs},'' {\em JHEP} {\bf 1106} (2011) 025,
\href{http://www.arXiv.org/abs/1003.1112}{{\tt 1003.1112}}.

\bibitem{Aharony:2013hda}
O.~Aharony, N.~Seiberg, and Y.~Tachikawa, ``{Reading between the lines of
  four-dimensional gauge theories},''
\href{http://www.arXiv.org/abs/1305.0318}{{\tt 1305.0318}}.

\bibitem{Billo:2013jda}
M.~Bill—, M.~Caselle, D.~Gaiotto, F.~Gliozzi, M.~Meineri, {\em et al.}, ``{Line
  defects in the 3d Ising model},'' {\em JHEP} {\bf 1307} (2013) 055,
\href{http://www.arXiv.org/abs/1304.4110}{{\tt 1304.4110}}.

\bibitem{Kapustin:2006hi}
A.~Kapustin, ``{Holomorphic reduction of N=2 gauge theories, Wilson-'t Hooft
  operators, and S-duality},''
\href{http://www.arXiv.org/abs/hep-th/0612119}{{\tt hep-th/0612119}}.

\bibitem{Drukker:2009id}
N.~Drukker, J.~Gomis, T.~Okuda, and J.~Teschner, ``{Gauge Theory Loop Operators
  and Liouville Theory},''
\href{http://www.arXiv.org/abs/0909.1105}{{\tt 0909.1105}}.

\bibitem{Drukker:2009tz}
N.~Drukker, D.~R. Morrison, and T.~Okuda, ``{Loop operators and S-duality from
  curves on Riemann surfaces},'' {\em JHEP} {\bf 0909} (2009) 031,
\href{http://www.arXiv.org/abs/0907.2593}{{\tt 0907.2593}}.

\bibitem{Gaiotto:2010be}
D.~Gaiotto, G.~W. Moore, and A.~Neitzke, ``{Framed BPS States},''
\href{http://www.arXiv.org/abs/1006.0146}{{\tt 1006.0146}}.

\bibitem{Gaiotto:2009hg}
D.~Gaiotto, G.~W. Moore, and A.~Neitzke, ``{Wall-crossing, Hitchin Systems, and
  the WKB Approximation},''
\href{http://www.arXiv.org/abs/0907.3987}{{\tt 0907.3987}}.

\bibitem{MR2233852}
V.~Fock and A.~Goncharov, ``Moduli spaces of local systems and higher
  {T}eichm\"uller theory,'' {\em Publ. Math. Inst. Hautes \'Etudes Sci.}
  (2006), no.~103, 1--211, \href{http://www.arXiv.org/abs/math/0311149}{{\tt
  math/0311149}}.

\bibitem{MR2349682}
V.~V. Fock and A.~B. Goncharov, ``Dual {T}eichm\"uller and lamination spaces,''
  in {\em Handbook of {T}eichm\"uller theory. {V}ol. {I}}, vol.~11 of {\em IRMA
  Lect. Math. Theor. Phys.}, pp.~647--684.
\newblock Eur. Math. Soc., Z\"urich, 2007.
\newblock \href{http://www.arXiv.org/abs/math/0510312}{{\tt math/0510312}}.

\bibitem{Kapustin:2006pk}
A.~Kapustin and E.~Witten, ``{Electric-Magnetic Duality And The Geometric
  Langlands Program},'' {\em Commun.Num.Theor.Phys.} {\bf 1} (2007) 1--236,
\href{http://www.arXiv.org/abs/hep-th/0604151}{{\tt hep-th/0604151}}.

\bibitem{Gomis:2009ir}
J.~Gomis, T.~Okuda, and D.~Trancanelli, ``{Quantum 't Hooft operators and
  S-duality in N=4 super Yang-Mills},'' {\em Adv.Theor.Math.Phys.} {\bf 13}
  (2009) 1941--1981,
\href{http://www.arXiv.org/abs/0904.4486}{{\tt 0904.4486}}.

\bibitem{Gomis:2009xg}
J.~Gomis and T.~Okuda, ``{S-duality, 't Hooft operators and the operator
  product expansion},'' {\em JHEP} {\bf 0909} (2009) 072,
\href{http://www.arXiv.org/abs/0906.3011}{{\tt 0906.3011}}.

\bibitem{Cecotti:2010fi}
S.~Cecotti, A.~Neitzke, and C.~Vafa, ``{R-Twisting and 4d/2d
  Correspondences},''
\href{http://www.arXiv.org/abs/1006.3435}{{\tt 1006.3435}}.

\bibitem{Ito:2011ea}
Y.~Ito, T.~Okuda, and M.~Taki, ``{Line operators on $S^1\times R^3$ and
  quantization of the Hitchin moduli space},'' {\em JHEP} {\bf 1204} (2012)
  010,
\href{http://www.arXiv.org/abs/1111.4221}{{\tt 1111.4221}}.

\bibitem{Saulina:2011qr}
N.~Saulina, ``{A note on Wilson-'t Hooft operators},'' {\em Nucl.Phys.} {\bf
  B857} (2012) 153--171,
\href{http://www.arXiv.org/abs/1110.3354}{{\tt 1110.3354}}.

\bibitem{Moraru:2012nu}
R.~Moraru and N.~Saulina, ``{OPE of Wilson-'t Hooft operators in N=4 and N=2
  SYM with gauge group G=PSU(3)},''
\href{http://www.arXiv.org/abs/1206.6896}{{\tt 1206.6896}}.

\bibitem{Xie:2013lca}
D.~Xie, ``{Higher laminations, webs and N=2 line operators},''
\href{http://www.arXiv.org/abs/1304.2390}{{\tt 1304.2390}}.

\bibitem{Douglas:1996sw}
M.~R. Douglas and G.~W. Moore, ``{D-branes, quivers, and ALE instantons},''
\href{http://www.arXiv.org/abs/hep-th/9603167}{{\tt hep-th/9603167}}.

\bibitem{Douglas:2000ah}
M.~R. Douglas, B.~Fiol, and C.~Romelsberger, ``{Stability and BPS branes},''
  {\em JHEP} {\bf 0509} (2005) 006,
\href{http://www.arXiv.org/abs/hep-th/0002037}{{\tt hep-th/0002037}}.

\bibitem{Douglas:2000qw}
M.~R. Douglas, B.~Fiol, and C.~Romelsberger, ``{The Spectrum of BPS branes on a
  noncompact Calabi-Yau},'' {\em JHEP} {\bf 0509} (2005) 057,
\href{http://www.arXiv.org/abs/hep-th/0003263}{{\tt hep-th/0003263}}.

\bibitem{Alim:2011kw}
M.~Alim, S.~Cecotti, C.~Cordova, S.~Espahbodi, A.~Rastogi, {\em et al.}, ``{N=2
  Quantum Field Theories and Their BPS Quivers},''
\href{http://www.arXiv.org/abs/1112.3984}{{\tt 1112.3984}}.

\bibitem{Fiol:2000pd}
B.~Fiol, ``{The BPS spectrum of N=2 SU(N) SYM and parton branes},''
\href{http://www.arXiv.org/abs/hep-th/0012079}{{\tt hep-th/0012079}}.

\bibitem{Fiol:2000wx}
B.~Fiol and M.~Marino, ``{BPS states and algebras from quivers},'' {\em JHEP}
  {\bf 0007} (2000) 031,
\href{http://www.arXiv.org/abs/hep-th/0006189}{{\tt hep-th/0006189}}.

\bibitem{Fiol:2006jz}
B.~Fiol, ``{The BPS spectrum of N=2 SU(N) SYM},'' {\em JHEP} {\bf 0602} (2006)
065.

\bibitem{Cecotti:2011rv}
S.~Cecotti and C.~Vafa, ``{Classification of complete N=2 supersymmetric
  theories in 4 dimensions},'' {\em Surveys in differential geometry, vol} {\bf
  18} (2013)
\href{http://www.arXiv.org/abs/1103.5832}{{\tt 1103.5832}}.

\bibitem{Cecotti:2011gu}
S.~Cecotti and M.~Del~Zotto, ``{On Arnold's 14 `exceptional' N=2 superconformal
  gauge theories},'' {\em JHEP} {\bf 1110} (2011) 099,
\href{http://www.arXiv.org/abs/1107.5747}{{\tt 1107.5747}}.

\bibitem{Alim:2011ae}
M.~Alim, S.~Cecotti, C.~Cordova, S.~Espahbodi, A.~Rastogi, {\em et al.}, ``{BPS
  Quivers and Spectra of Complete N=2 Quantum Field Theories},''
\href{http://www.arXiv.org/abs/1109.4941}{{\tt 1109.4941}}.

\bibitem{DelZotto:2011an}
M.~Del~Zotto, ``{More Arnold's N = 2 superconformal gauge theories},'' {\em
  JHEP} {\bf 1111} (2011) 115,
\href{http://www.arXiv.org/abs/1110.3826}{{\tt 1110.3826}}.

\bibitem{Xie:2012dw}
D.~Xie, ``{Network, Cluster coordinates and N=2 theory I},''
\href{http://www.arXiv.org/abs/1203.4573}{{\tt 1203.4573}}.

\bibitem{Cecotti:2012va}
S.~Cecotti, ``{Categorical Tinkertoys for N=2 Gauge Theories},'' {\em
  Int.J.Mod.Phys.} {\bf A28} (2013) 1330006,
\href{http://www.arXiv.org/abs/1203.6734}{{\tt 1203.6734}}.

\bibitem{Cecotti:2012sf}
S.~Cecotti and M.~Del~Zotto, ``{Half-Hypers and Quivers},'' {\em JHEP} {\bf
  1209} (2012) 135,
\href{http://www.arXiv.org/abs/1207.2275}{{\tt 1207.2275}}.

\bibitem{Cecotti:2012gh}
S.~Cecotti and M.~Del~Zotto, ``{4d N=2 Gauge Theories and Quivers: the
  Non-Simply Laced Case},'' {\em JHEP} {\bf 1210} (2012) 190,
\href{http://www.arXiv.org/abs/1207.7205}{{\tt 1207.7205}}.

\bibitem{Saidi:2012gi}
E.~H. Saidi, ``{Weak Coupling Chambers in N=2 BPS Quiver Theory},'' {\em
  Nucl.Phys.} {\bf B864} (2012) 190--202,
\href{http://www.arXiv.org/abs/1208.2887}{{\tt 1208.2887}}.

\bibitem{Cecotti:2012jx}
S.~Cecotti and M.~Del~Zotto, ``{Infinitely many N=2 SCFT with ADE flavor
  symmetry},'' {\em JHEP} {\bf 1301} (2013) 191,
\href{http://www.arXiv.org/abs/1210.2886}{{\tt 1210.2886}}.

\bibitem{Xie:2012jd}
D.~Xie, ``{Network, cluster coordinates and N = 2 theory II: Irregular
  singularity},''
\href{http://www.arXiv.org/abs/1207.6112}{{\tt 1207.6112}}.

\bibitem{Cecotti:2012se}
S.~Cecotti, ``{The quiver approach to the BPS spectrum of a 4d N=2 gauge
  theory},''
\href{http://www.arXiv.org/abs/1212.3431}{{\tt 1212.3431}}.

\bibitem{Cecotti:2012kv}
S.~Cecotti, ``{N=2 SUSY and representation theory: An introduction},'' {\em
  PoS} {\bf ICMP2012} (2012)
005.

\bibitem{Cecotti:2013lda}
S.~Cecotti, M.~Del~Zotto, and S.~Giacomelli, ``{More on the N=2 superconformal
  systems of type $D_p(G)$},''
\href{http://www.arXiv.org/abs/1303.3149}{{\tt 1303.3149}}.

\bibitem{Cecotti:2013sza}
S.~Cecotti and M.~Del~Zotto, ``{The BPS spectrum of the 4d N=2 SCFT's $H_1,
  H_2, D_4, E_6, E_7, E_8$},'' {\em JHEP} {\bf 1306} (2013) 075,
\href{http://www.arXiv.org/abs/1304.0614}{{\tt 1304.0614}}.

\bibitem{Galakhov:2013oja}
D.~Galakhov, P.~Longhi, T.~Mainiero, G.~W. Moore, and A.~Neitzke, ``{Wild Wall
  Crossing and BPS Giants},''
\href{http://www.arXiv.org/abs/1305.5454}{{\tt 1305.5454}}.

\bibitem{Chuang:2013wt}
W.-y. Chuang, D.-E. Diaconescu, J.~Manschot, G.~W. Moore, and Y.~Soibelman,
  ``{Geometric engineering of (framed) BPS states},''
\href{http://www.arXiv.org/abs/1301.3065}{{\tt 1301.3065}}.

\bibitem{Cirafici:2013bha}
M.~Cirafici, ``{Line defects and (framed) BPS quivers},''
\href{http://www.arXiv.org/abs/1307.7134}{{\tt 1307.7134}}.

\bibitem{MR2403807}
B.~Szendr{\H{o}}i, ``Non-commutative {D}onaldson-{T}homas invariants and the
  conifold,'' {\em Geom. Topol.} {\bf 12} (2008), no.~2, 1171--1202.

\bibitem{MR2592501}
S.~Mozgovoy and M.~Reineke, ``On the noncommutative {D}onaldson-{T}homas
  invariants arising from brane tilings,'' {\em Adv. Math.} {\bf 223} (2010),
  no.~5, 1521--1544.

\bibitem{Ooguri:2008yb}
H.~Ooguri and M.~Yamazaki, ``{Crystal Melting and Toric Calabi-Yau
  Manifolds},'' {\em Commun.Math.Phys.} {\bf 292} (2009) 179--199,
\href{http://www.arXiv.org/abs/0811.2801}{{\tt 0811.2801}}.

\bibitem{Denef:2002ru}
F.~Denef, ``{Quantum quivers and Hall / hole halos},'' {\em JHEP} {\bf 0210}
  (2002) 023,
\href{http://www.arXiv.org/abs/hep-th/0206072}{{\tt hep-th/0206072}}.

\bibitem{Denef:2000nb}
F.~Denef, ``{Supergravity flows and D-brane stability},'' {\em JHEP} {\bf 0008}
  (2000) 050,
\href{http://www.arXiv.org/abs/hep-th/0005049}{{\tt hep-th/0005049}}.

\bibitem{Denef:2007vg}
F.~Denef and G.~W. Moore, ``{Split states, entropy enigmas, holes and halos},''
  {\em JHEP} {\bf 1111} (2011) 129,
\href{http://www.arXiv.org/abs/hep-th/0702146}{{\tt hep-th/0702146}}.

\bibitem{deBoer:2008zn}
J.~de~Boer, S.~El-Showk, I.~Messamah, and D.~Van~den Bleeken, ``{Quantizing N=2
  Multicenter Solutions},'' {\em JHEP} {\bf 0905} (2009) 002,
\href{http://www.arXiv.org/abs/0807.4556}{{\tt 0807.4556}}.

\bibitem{Manschot:2011xc}
J.~Manschot, B.~Pioline, and A.~Sen, ``{A Fixed point formula for the index of
  multi-centered N=2 black holes},'' {\em JHEP} {\bf 1105} (2011) 057,
\href{http://www.arXiv.org/abs/1103.1887}{{\tt 1103.1887}}.

\bibitem{Manschot:2012rx}
J.~Manschot, B.~Pioline, and A.~Sen, ``{From Black Holes to Quivers},'' {\em
  JHEP} {\bf 1211} (2012) 023,
\href{http://www.arXiv.org/abs/1207.2230}{{\tt 1207.2230}}.

\bibitem{Manschot:2013sya}
J.~Manschot, B.~Pioline, and A.~Sen, ``{On the Coulomb and Higgs branch
  formulae for multi-centered black holes and quiver invariants},'' {\em JHEP}
  {\bf 1305} (2013) 166,
\href{http://www.arXiv.org/abs/1302.5498}{{\tt 1302.5498}}.

\bibitem{KS1}
M.~Kontsevich and Y.~Soibelman, ``{Stability structures, motivic
  Donaldson-Thomas invariants and cluster transformations},''
\href{http://www.arXiv.org/abs/0811.2435}{{\tt 0811.2435}}.

\bibitem{KS2}
M.~Kontsevich and Y.~Soibelman, ``{Motivic Donaldson-Thomas invariants: summary
  of results},''
\href{http://www.arXiv.org/abs/0910.4315}{{\tt 0910.4315}}.

\bibitem{KS3}
M.~Kontsevich and Y.~Soibelman, ``{Cohomological Hall algebra, exponential
  Hodge structures and motivic Donaldson-Thomas invariants},''
\href{http://www.arXiv.org/abs/1006.2706}{{\tt 1006.2706}}.

\bibitem{JS}
D.~Joyce and Y.~Song, ``{A theory of generalized Donaldson-Thomas
  invariants},''
\href{http://www.arXiv.org/abs/0810.5645}{{\tt 0810.5645}}.

\bibitem{J}
D.~Joyce, ``{Generalized Donaldson-Thomas invariants},''
\href{http://www.arXiv.org/abs/0910.0105}{{\tt 0910.0105}}.

\bibitem{Argyres:1995jj}
P.~C. Argyres and M.~R. Douglas, ``{New phenomena in SU(3) supersymmetric gauge
  theory},'' {\em Nucl.Phys.} {\bf B448} (1995) 93--126,
\href{http://www.arXiv.org/abs/hep-th/9505062}{{\tt hep-th/9505062}}.

\bibitem{Argyres:1995xn}
P.~C. Argyres, M.~R. Plesser, N.~Seiberg, and E.~Witten, ``{New N=2
  superconformal field theories in four-dimensions},'' {\em Nucl.Phys.} {\bf
  B461} (1996) 71--84,
\href{http://www.arXiv.org/abs/hep-th/9511154}{{\tt hep-th/9511154}}.

\bibitem{Seiberg:1994rs}
N.~Seiberg and E.~Witten, ``{Electric - magnetic duality, monopole
  condensation, and confinement in N=2 supersymmetric Yang-Mills theory},''
  {\em Nucl.Phys.} {\bf B426} (1994) 19--52,
\href{http://www.arXiv.org/abs/hep-th/9407087}{{\tt hep-th/9407087}}.

\bibitem{Seiberg:1994aj}
N.~Seiberg and E.~Witten, ``{Monopoles, duality and chiral symmetry breaking in
  N=2 supersymmetric QCD},'' {\em Nucl.Phys.} {\bf B431} (1994) 484--550,
\href{http://www.arXiv.org/abs/hep-th/9408099}{{\tt hep-th/9408099}}.

\bibitem{MR1887642}
S.~Fomin and A.~Zelevinsky, ``Cluster algebras. {I}. {F}oundations,'' {\em J.
  Amer. Math. Soc.} {\bf 15} (2002), no.~2, 497--529 (electronic).

\bibitem{MR2295199}
S.~Fomin and A.~Zelevinsky, ``Cluster algebras. {IV}. {C}oefficients,'' {\em
  Compos. Math.} {\bf 143} (2007), no.~2, 112--164,
  \href{http://www.arXiv.org/abs/math/0602259}{{\tt math/0602259}}.

\bibitem{ghki}
M.~Gross, P.~Hacking, and S.~Keel, ``{Mirror symmetry for log Calabi-Yau
  surfaces I},'' \href{http://www.arXiv.org/abs/1106.4977}{{\tt 1106.4977}}.

\bibitem{goncharovtoappear}
A.~Goncharov, 2013.
\newblock To appear.

\bibitem{Seiberg:1996nz}
N.~Seiberg and E.~Witten, ``{Gauge dynamics and compactification to
  three-dimensions},''
\href{http://www.arXiv.org/abs/hep-th/9607163}{{\tt hep-th/9607163}}.

\bibitem{Gaiotto:2008cd}
D.~Gaiotto, G.~W. Moore, and A.~Neitzke, ``{Four-dimensional wall-crossing via
  three-dimensional field theory},'' {\em Commun.Math.Phys.} {\bf 299} (2010)
  163--224,
\href{http://www.arXiv.org/abs/0807.4723}{{\tt 0807.4723}}.

\bibitem{Bershadsky:1995vm}
M.~Bershadsky, A.~Johansen, V.~Sadov, and C.~Vafa, ``{Topological reduction of
  4-d SYM to 2-d sigma models},'' {\em Nucl.Phys.} {\bf B448} (1995) 166--186,
\href{http://www.arXiv.org/abs/hep-th/9501096}{{\tt hep-th/9501096}}.

\bibitem{Harvey:1995tg}
J.~A. Harvey, G.~W. Moore, and A.~Strominger, ``{Reducing S duality to T
  duality},'' {\em Phys.Rev.} {\bf D52} (1995) 7161--7167,
\href{http://www.arXiv.org/abs/hep-th/9501022}{{\tt hep-th/9501022}}.

\bibitem{Bridgeland1}
T.~Bridgeland and I.~Smith, ``{ Quadratic differentials as stability
  conditions},''
\href{http://www.arXiv.org/abs/1302.7030}{{\tt 1302.7030}}.

\bibitem{MR0364254}
G.~Harder and M.~S. Narasimhan, ``On the cohomology groups of moduli spaces of
  vector bundles on curves,'' {\em Math. Ann.} {\bf 212} (1974/75) 215--248.

\bibitem{Reineke1}
M.~Reineke, ``{The Harder-Narasimhan system in quantum groups and cohomology of
  quiver moduli},''
\href{http://www.arXiv.org/abs/math.QA/0204059}{{\tt math.QA/0204059}}.

\bibitem{Manschot:2010qz}
J.~Manschot, B.~Pioline, and A.~Sen, ``{Wall Crossing from Boltzmann Black Hole
  Halos},'' {\em JHEP} {\bf 1107} (2011) 059,
\href{http://www.arXiv.org/abs/1011.1258}{{\tt 1011.1258}}.

\bibitem{Sen:2011aa}
A.~Sen, ``{Equivalence of Three Wall Crossing Formulae},''
\href{http://www.arXiv.org/abs/1112.2515}{{\tt 1112.2515}}.

\bibitem{Reineke5}
M.~Reineke and S.~Mozgovoy, ``{ Abelian quiver invariants and marginal
  wall-crossing},''
\href{http://www.arXiv.org/abs/1212.0410}{{\tt 1212.0410}}.

\bibitem{Bena:2012hf}
I.~Bena, M.~Berkooz, J.~de~Boer, S.~El-Showk, and D.~Van~den Bleeken,
  ``{Scaling BPS Solutions and pure-Higgs States},'' {\em JHEP} {\bf 1211}
  (2012) 171,
\href{http://www.arXiv.org/abs/1205.5023}{{\tt 1205.5023}}.

\bibitem{Lee:2012sc}
S.-J. Lee, Z.-L. Wang, and P.~Yi, ``{Quiver Invariants from Intrinsic Higgs
  States},'' {\em JHEP} {\bf 1207} (2012) 169,
\href{http://www.arXiv.org/abs/1205.6511}{{\tt 1205.6511}}.

\bibitem{Lee:2012naa}
S.-J. Lee, Z.-L. Wang, and P.~Yi, ``{BPS States, Refined Indices, and Quiver
  Invariants},'' {\em JHEP} {\bf 1210} (2012) 094,
\href{http://www.arXiv.org/abs/1207.0821}{{\tt 1207.0821}}.

\bibitem{Keller2}
B.~Keller and D.~Yang, ``{Derived equivalences from mutations of quivers with
  potential},''
\href{http://www.arXiv.org/abs/0906.0761}{{\tt 0906.0761}}.

\bibitem{Andriyash:2010yf}
E.~Andriyash, F.~Denef, D.~L. Jafferis, and G.~W. Moore, ``{Bound state
  transformation walls},'' {\em JHEP} {\bf 1203} (2012) 007,
\href{http://www.arXiv.org/abs/1008.3555}{{\tt 1008.3555}}.

\bibitem{MR0332887}
P.~Gabriel, ``Unzerlegbare {D}arstellungen. {I},'' {\em Manuscripta Math.} {\bf
  6} (1972) 71--103; correction, ibid. 6 (1972), 309.

\bibitem{MR1476671}
M.~Auslander, I.~Reiten, and S.~O. Smal{\o}, {\em Representation theory of
  {A}rtin algebras}, vol.~36 of {\em Cambridge Studies in Advanced
  Mathematics}.
\newblock Cambridge University Press, Cambridge, 1997.
\newblock Corrected reprint of the 1995 original.

\end{thebibliography}\endgroup

\end{document}